\documentclass[11pt]{report}
\usepackage[dvipdfmx]{graphicx}
\usepackage{amsmath}%
\usepackage{amsfonts}%
\usepackage{amssymb}%
\newtheorem{theorem}{ Theorem}[chapter]

\newtheorem{lemma}{ Lemma}[chapter]
\newtheorem{conjecture}{ Conjecture}[chapter]
\newtheorem{corollary}[theorem]{ Corollary}
\newenvironment{proof}%
 {\par\noindent{ \underline{Proof} \quad}}{\hfill$\Box$\bigskip}
\newenvironment{proofth}%
 {\par\noindent{ \underline{Proof} of the theorem\quad}}{\hfill$\Box$\bigskip}
\newenvironment{remark}%
 {\par\smallskip\noindent{ \underline{{\it Remark}} \quad}}{\par\smallskip}
\newenvironment{fact}%
 {\par\smallskip\noindent{ \underline{{\it Fact}} \quad}}{\par\smallskip}
\newenvironment{example}%
 {\par\smallskip\noindent{ \underline{{\it Example}} \quad}}{\par\smallskip}
 {\par\smallskip\noindent{{\it Assumotion} \quad}}{\par\smallskip}
 {\par\smallskip\noindent{{\it Condition} \quad}}{\par\smallskip}

\newcommand{\ds}{\displaystyle}
\newcommand{\ol}{\overline}

\newcommand{\tr}{{\rm tr}}
\newcommand{\Tr}{{\rm Tr}}
\newcommand{\R}{{\bf R}}
\newcommand{\C}{{\bf C}}
\newcommand{\argmin}{{\rm argmin}}
\newcommand{\argmax}{{\rm argmax}}
\newcommand{\rank}{{\rm rank}}
\newcommand{\abs}{{\rm abs}}
\newcommand{\diag}{{\rm diag}}
\newcommand{\CR}{{\rm CR}}

\newcommand{\pte}{\mathop{{\coprod}^{(e)}}}
\newcommand{\ptm}{\mathop{{\coprod}^{(m)}}}
\newcommand{\nablae}{\nabla^{(e)}}

\newcommand{\Te}{T^{(e)}}

\newcommand{\pttw}{\mathop{{\coprod}^{(w)}}}
\newcommand{\nablaw}{\nabla^{(w)}}

\newcommand{\Ttw}{T^{(w)}}

\newcommand{\hlift}{\mathop{h}}
\newcommand{\rgl}{\rangle}
\newcommand{\lgl}{\langle}
\newcommand{\rgll}{\rangle\rangle}
\newcommand{\lgll}{\langle\langle}
\newcommand{\rrgll}{\right\rangle\right\rangle}
\newcommand{\llgll}{\left\langle\left\langle}

\begin{document}

\begin{titlepage}
\vspace*{1.5cm}
\begin{center}
  \Huge A Geometrical Approach to Quantum Estimation Theory\\
\end{center}
\vspace{5cm}
\begin{center}
  \Huge Keiji Matsumoto\\
\end{center}
\end{titlepage}

\pagenumbering{roman}
\tableofcontents
\pagebreak
\pagenumbering{arabic}



\chapter{Introduction}
\section{The purposes of the thesis}
The most important purpose of the thesis is 
pursuit for  the geometrical theory of 
statistical  estimation of the quantum mechanical state. 

In the  statistical theory of the probability distribution,
S. Amari and his coworkers have formulated
a geometrical theory,  so called {\it information geometry},
and successfully applied to various statistical problems \cite{Amari}.
H. Nagaoka, one of Amari's  coworkers, pointed out that
the {\it duality} between $e$- and $m$-connections sits at the heart of
the information geometry.
He also formulated {\it quantum information geometry} by use of
this idea of the mutually dual connections, and applied to
characterization of the model which has the efficient estimator
\cite{Nagaoka:1989:2}\cite{Nagaoka:1994}.

However, his information geometry does not give any insight
into the problem 
of determination of the attainable CR (Cramer-Rao) type bound,
nor characterize the condition for the SLD CR bound,
which is of special interest for some reasons,
is attainable.
After all, Nagaoka's geometry deal with
the global properties of the model,
while the attainable CR type bound is related local properties
of the model.
Hence, it semms that another geometric structure
is needed for the thorough description 
of the quantum estimation theory.

On the other hand, Berry's phase, discovered by M. V. Berry
as the  non-integrable phase factor in the adiabatic motion\cite{Berry},
is naturally understood as a curvature of 
the natural connection in the principle fiber bundle
over the space of pure states whose structure group is $U(1)$ \cite{AA}.
In 1986, Uhlmann generalized the geometry to the space of mixed states.
Though the physical meaning of Uhlmann's geometry
is not known, Berry's phase is 
applied to the explanation of various phenomina\cite{Shapere}. 

The author conjectures
that Berry-Uhlmann's curvature
reflects local properties of the model.
To prove the statement, it is needed to determine
the attainable CR type bound for arbitrary models,
which is far out of our reach.
However, for the 2-parameter pure state model,
the author  presents complete answer to the problem.
In addition, for  the faithful state model and pure state model,
it is shown that SLD CR bound is attainable if and only if
the model is free of Berry-Uhlmann curvature.

Furthermore, we try a kind of unification of the two geometries, 
Nagaoka's information geometry and Uhlmann's parallelism.

Second most important purpose is the determination of
the attainable CR type bound, and the development  of new methodology
for that purpose. In the pure state models, 
this purpose is achieved, though not completely, 
to large extent.
A new methodology {\it direct approach} is formulated and 
successfully applied to the 2-parameter pure state model,
the coherent model.
For the arbitrary pure state model,
calculated is the attainable CR type bound whose weight matrix is 
the SLD Fisher information matrix.
Looking back, no one has ever determined 
the attainable CR type bound for this wide range of models.

Third, some considerations about such physical problems
as the uncertainty principle are done.
The time-energy uncertainty is nicely formulated
as a hypothesis test, and the position-momentum uncertainty
as a estimation of the mean values 
of the position and the momentum operators.
In this formulation of the position-momentum uncertainty , 
it is shown that the mean values 
of the position and the momentum operators
are simultaneously estimated up to arbitrarily high efficiency,
if the particle is prepared carefully.

\section{Organization of the thesis}
The thesis is divided into three parts: 
the faithful model theory, the pure state theory, the general model theory;
The reason for this organization is that
the extent of the achievement of the purposes is
different in these three cases.

Before these three parts,
chapter 2 gives brief review of the estimation theory of
probability distributions, the quantum mechanical theory of the measurement,
and the quantum estimation theory,
and chapter 3 gives
the geometrical and the estimation-theoretical
framework, commonly used in any of the following three parts.

\chapter{Preliminaries}
In this chapter, statistical estimation theory and
quantum mechanics are reviewed briefly.
For the thorough description of estimation theory,
see, for example, 
Ref. \cite{Lehmann:1983}.
As for quantum mechanics, see Ref. \cite{Sakurai}, or other text books.
Some basic concepts in quantum estimation theory are  introduced also.

\section{Classical estimation theory}
Throughout the thesis, the usual estimation theory,
or the estimation theory of the probability distribution
is called classical estimation theory,
in the sense that the theory is not quantum mechanical.

The theme of the classical estimation theory is
identification of the probability distribution
from which the $N$ data $x_1, x_2, ..., x_N $ is produced.
Usually, the probability distribution is assumed to be a member of 
a {\it model}, or a family 
\begin{eqnarray}
{\cal M}=\{p(x|\theta)|\theta\in\Theta\subset\R^m\}
\nonumber
\end{eqnarray}
of probability distributions
and that the finite dimensional parameter $\theta\in\Theta\subset\R^m$ 
is to be estimated statistically.

{\it Unbiased estimator} $\hat\theta=\hat\theta(x_1,x_2,...,x_N)$
 of parameter $\theta$ the estimate which satisfies
\begin{eqnarray}
E_{\theta}[\hat\theta]
&\equiv &
\int dx_1dx_2,...,dx_N
\hat\theta(x_1,x_2,...,x_N)\prod_{i=1}^{N} p(x_i|\theta)\nonumber\\
&=&\theta,
\label{eqn:clunbiased}
\end{eqnarray}
that is, the estimate which gives the true value of parameter
in average.
For the technical reason, we also define 
{\it locally unbiased estimator $\hat\theta$ at $\theta_0$}
by
\begin{eqnarray}
&&E_{\theta_0}[\hat\theta]=\theta_0,\nonumber\\
&&\left. \partial_j E_{\theta}[\hat\theta^i]\right|_{\theta=\theta_0}
=\delta^i_j.
\nonumber
\end{eqnarray}
The estimator is unbiased iff it is locally unbiased
at every $\theta\in\Theta$.

For the variance of locally unbiased estimator at $\theta$, the following
theorem gives bound of efficiency of the estimation.

\begin{theorem}(Cramer-Rao inequality)
For any locally unbiased estimate $\hat\theta$ at $\theta$,
\begin{eqnarray}
V_{\theta}[\hat\theta]\ge \frac{1}{N} J^{-1}(\theta).
\label{eqn:clcrN}
\end{eqnarray}
Here, $N$ is the number of the data and $J(\theta)$
is $m\times m$ real symmetric matrix defined by
\begin{eqnarray}
J(\theta)\equiv
\left[\int dx p(x|\theta)
 \partial_i \ln p(x|\theta)\partial_j \ln p(x|\theta) 
\right],
\label{eqn:clbest}
\end{eqnarray}
where $\partial_i$ stands for $\partial/\partial\theta^i$.

The best estimator, or the estimator $\hat\theta$ satisfying
$(\ref{eqn:clcr})$, is given by
\begin{eqnarray}
\hat\theta^i(x_1,...,x_N)&=&\hat\theta_{(\theta)}^i(x_1,...,x_N)\nonumber\\
&\equiv& \theta^i+
\sum_{j=1}^m [J^{-1}(\theta)]^{ij}\partial_j \ln\,\prod_{k=1}^N p(x_k|\theta).
\nonumber
\end{eqnarray}
\end{theorem}

$J(\theta)$ is called {\it Fisher information matrix},
because the larger the $J(\theta)$ is, the more precise estimate
can be done with the same number of data. Metaphorically speaking,
we obtain as much information as $J(\theta)$ per data.
Actually, as easily seen by putting $N=1$ in Cramer-Rao (CR) inequality,
we can obtain $J(\theta)$ as the minimum variance of 
locally unbiased estimate when only one data is given.
\begin{eqnarray}
V_{\theta}[\hat\theta]\ge  J^{-1}(\theta).
\label{eqn:clcr}
\end{eqnarray}

The trouble with  the $(\ref{eqn:clcr})$ is 
that the best estimator $\hat\theta_{(\theta)}$ is
dependent on the true value of the parameter $\theta$,
which is unknown to us.
When the true value of parameter is not $\theta_0$,
the estimate $\hat\theta_{(\theta_0)}$ is not even locally unbiased
at the true value of the parameter.
To avoid this dilemma,
we give up with the unbiased estimator,
and focus on the {\it consistent estimator} defined by
\begin{eqnarray}
\lim_{N\rightarrow\infty} E_{\theta}[\hat\theta(x_1,x_2,...,x_N)]=\theta.
\nonumber
\end{eqnarray}
For the consistent estimator, we also have the following theorem.

\begin{theorem}
If the estimator is consistent,
\begin{eqnarray}
V_{\theta}[\hat\theta]\ge 
\frac{1}{N} J^{-1}(\theta)+o\left(\frac{1}{N}\right)
\label{eqn:clcscr}
\end{eqnarray}
holds true.

The maximum likelihood estimator $\hat\theta_{MLE}$,
which is defined by,
\begin{eqnarray}
\hat\theta_{MLE}\equiv
\argmax \left\{
\left.\sum_{j=1}^N \ln p(x_i|\theta)\:
\right|\:\theta\in\Theta\subset\R^m
\right\}.
\nonumber
\end{eqnarray}
is consistent and achieves the equality in $(\ref{eqn:clcscr})$.
\end{theorem}

Notice that to obtain $\hat\theta_{MLE}$, 
we need no information about the true value of the parameter beforehand.
Hence, the Fisher information matrix is 
a good measure of the efficiency of the optimal consistent estimator.

\section{ Quantum mechanics and measurement theory}
In the quantum mechanics, 
the state of physical system is described 
by the density operator $\rho$, which is 
a non-negative Hermitian operator whose trace is equal to 1,
in a separable Hilbert space ${\cal H}$, whose dimension is
denoted by $d\leq\infty$ hereafter.
We denote by ${\cal P}({\cal H})$
the space of  density operators in ${\cal H}$,
by  ${\cal P}_r ({\cal H})$  
the the space of density operators whose rank is $r$,
and by ${\cal P}_+({\cal H})$ the space of strictly positive definite 
density operators.
${\cal P}({\cal H})$, ${\cal P}_r ({\cal H})$, and  ${\cal P}_+({\cal H})$
are often simply denoted by  ${\cal P}$,
${\cal P}_r $, and  ${\cal P}_+$, respectively .

Let $\Omega$ be a space of all possible outcomes of 
an experiment,
and $\sigma(\Omega)$ be a $\sigma$- field
in $\Omega$.
When the density operator of the system is $\rho$,
the probability  that the data $\omega \in \Omega$
lies in $B\in\sigma(\Omega)$ writes 
\begin{eqnarray}
{\rm Pr}\{ \omega \in B|\rho \} =\tr \rho M(B),
\label{eqn:pdm}
\end{eqnarray}
by use of the map  $M$ from $\sigma(\Omega)$
to  nonnegative Hermitian operator which satisfies
\begin{eqnarray}
&&M(\phi)=O, M(\Omega)=I,\nonumber\\
&&M(\bigcup_{i=1}^{\infty} B_i),
=\sum_{i=1}^{\infty}M(B_i)
\;\;(B_i\cap B_j=\phi,i\neq j),
\label{eqn:maxiom}
\end{eqnarray}
so that $(\ref{eqn:pdm})$ define a probability measure
(see Ref.{\rm \cite{Helstrom:1976}}, p.53 
and Ref.{\rm \cite{Holevo:1982}}, p.50).
We call the map $M$ the {\it measurement}, because
there always exist an physical experiment corresponds to the map $M$
which satisfies $(\ref{eqn:maxiom})$
\cite{Steinspring:1955}\cite{Ozawa:1984}.

\section{Unbiased estimator in quantum estimation theory}
The purpose of the quantum estimation is to 
identify the density operator of the given physical system
from the data obtained by the appropriately designed experiment.
For simplicity, we usually assume that
the density operator is a member of a {\it model}, or a manifold of
${\cal M}
=\{\rho(\theta)|\theta\in \Theta \subset {\bf R}^m\}\subset{\cal P}$,
and that
the parameter $\theta$ is to be estimated statistically.
For example, ${\cal M}$ is 
the set of spin states with given wave function part and unknown spin part.

To estimate the parameter, we performs an experiment
to obtain the data $\omega$ by which we calculate
an estimator $\hat\theta$ by {\it estimator } $\hat\theta(\omega)$.
A pair $(\hat\theta,M,\Omega)$ of a space $\Omega$ of data, a measurement $M$,
and an estimator $\hat\theta(*)$ is also called an {\it estimator}. 
The expectation of $f(\omega)$ with respect to the probability measure 
$(\ref{eqn:pdm})$ is denoted by $E_\theta[f(\omega)|M]$.

We have seen that the locally unbiased estimator played
a key role in classical estimation theory.
Hence we try to keep the same track also 
in the quantum estimation theory.

The estimator $(\hat\theta,M,\Omega)$ is said to be  {\it unbiased} 
if 
\begin{eqnarray}
E_\theta[\hat\theta(\omega)|M]=\theta
\label{eqn:unbiased}
\end{eqnarray}
holds for all $\theta\in\Theta$.
If $(\ref{eqn:unbiased})$ and
\begin{eqnarray}
\partial_i E_\theta[\theta^j(\omega)|M]=\delta^j_i\:(i,j=1,...,m)
\nonumber
\end{eqnarray}
hold at a 
particular  $\theta$, $(\hat\theta,M,\Omega)$ is called {\it locally unbiased} 
at $\theta$.

It is also reasonable to include calculation of the estimate from data 
into the process of measurement.
In this point of view,
the estimate $\hat\theta$ itself is produced 
by the measurement process, and the data space $\Omega$
is $\R^m$.
Therefore, by the term `estimator' 
we also  mean  the measurement which takes value on  $\R^m$.
In this case, 
the unbiased estimator is a measurement 
 which takes value on  $\R^m$ which satisfies
\begin{eqnarray}
E_\theta[M]=\theta
\label{eqn:unbiased2}
\end{eqnarray}
holds for all $\theta\in\Theta$,
where, 
\begin{eqnarray}
E_\theta[M]\equiv\int \hat\theta \tr\rho(\theta)M(d\hat\theta ).
\nonumber
\end{eqnarray}
If $(\ref{eqn:unbiased2})$ and
\begin{eqnarray}
[\partial_i E_\theta[M]]^j=\delta^j_i\:(i,j=1,...,m)
\nonumber
\end{eqnarray}
hold at a 
particular  $\theta$, $M$ is called {\it locally unbiased} 
at $\theta$.
We denote by $V_{\theta}[M]$ the 
covariance matrix of the estimator $M$ 
when the true value of the parameter is $\theta$.

Obviously, these two definition of the estimator
 are 
equivalent. 
Therefore,
in some situations, we prefer the former to the latter,
while in other situations the latter is 
preferred for the sake of simplicity.

\chapter{Conceptual framework}
\section{Horizontal lift and SLD}

In this thesis, except for the pure state model theory,
$d\equiv\dim{\cal H}$ is
assumed to be finite  for the sake of clarity.
The author believe the essence of the discussion will not be damaged 
by this restriction.

Let ${\cal W}_r$ be
the space of  $d\times r$ complex and full-rank matrix $W$ such that
\begin{eqnarray}
\tr WW^*=1,
\nonumber
\end{eqnarray}
${\cal P}_r$ the space of density operators whose rank is $r$,
and 
$\pi$ the map from ${\cal W}_r$
to ${\cal P}_r$ such that
\begin{eqnarray}
\rho=\pi(W)\equiv WW^*.
\nonumber
\end{eqnarray}
Because 
$\pi(WU)$ is identical to $\pi(W)$ iff $U$ is a $r \times r$ 
unitary matrix,
it is natural to see the  space 
${\cal W}_r$ as the total space of 
the {\it principal fiber bundle}
with the {\it base space} ${\cal P}_r$ 
and the ${\it structure group}$ $U(d)$ \cite{KobayashiN}.
One possible physical interpretation of
$W$ is a representation of a state vector 
$|W\rgl$ in a bigger
Hilbert space ${\cal H}\otimes{\cal H}'$.
Here, the dimension of ${\cal H}'$ is $r$ and   
the operation $\pi(*)$ corresponds to  the partial trace of
$|W\rgl\lgl W|$ over ${\cal H}'$.

In this section, basic concepts about 
the tangent bundle ${\cal T}({\cal W}_r)$ over 
${\cal W}_r$, which is a real manifold 
with the real parameter 
$\zeta=(\zeta^1,...,\zeta^{2rd-1})^T$,
are introduced.

The {\it matrix representation} 
${\bf M}(\partial/\partial \zeta^i)$ of 
the tangent vector $\partial/\partial \zeta^i$ 
(throughout the thesis, the tangent vector is understood as 
the differential operator) is
a $d\times r$ complex matrix  such that
\begin{eqnarray}
{\bf M}\left(\frac{\partial}{\partial\zeta^i}\right)\equiv
2\frac{\partial}{\partial\zeta^i} W(\zeta).
\nonumber
\end{eqnarray}
The real span of the  matrix representations
is  
\begin{eqnarray}
\{X\: |\: {\rm Re}\,\tr XW^*(\zeta)=0,X\in M(d,r,\C)\}.
\nonumber
\end{eqnarray}

We introduce the inner product $\lgll *,*\rgll_W$ to 
${\cal T}({\cal W}_r)$ such that,
\begin{eqnarray}
&&\lgll\hat{X},\hat{Y}\rgll_W\nonumber\\
&\equiv&\sum_{i,j}({\rm Re}({\bf M}\hat{X})_{ij}{\rm Re}({\bf M}\hat{Y})_{ij}
          +{\rm Im}({\bf M}\hat{X})_{ij}{\rm Im}({\bf M}\hat{X})_{ij})
\nonumber\\
&=&{\rm Re}\,\tr (\,({\bf M}\hat{X})({\bf M}\hat{Y})^*\,),
\nonumber
\end{eqnarray}
which is invariant under the action of 
$U\in U(r)$ to the matrix representation  of 
the tangent vector from right side,
\begin{eqnarray}
&&U\in U(r),\nonumber\\
&&\lgll\, ({\bf M}\hat{X}) U, ({\bf M}\hat{Y})U\rgll_{WU}
=\lgll {\bf M}\hat{X}, {\bf M}\hat{Y}\rgll_W
\nonumber
\end{eqnarray}

Let us decompose ${\cal T}_W({\cal W}_r)$ into
the direct sum of the {\it horizontal subspace} ${\cal LS}_W$ and 
the {\it vertical subspace} ${\cal K}_W$ 
where ${\cal LS}_W$ is defined by
\begin{eqnarray}
{\cal LS}_W\equiv
\{\hat{X}\: |\: W^*({\bf M}\hat{X})=({\bf M}\hat{X})^*W\},
\label{eqn:defLS2}
\end{eqnarray}
and 
${\cal K}_W$ is the orthogonal complement space 
${\cal T}_W({\cal W}_r)\ominus{\cal LS}_W$
with respect to the inner product $\lgll *,* \rgll_W$.
$\hat{X}\in{\cal K}_W$ satisfies
\begin{eqnarray}
({\bf M}\hat{X})W^*+ W({\bf M}\hat{X})^*
=0,
\label{eqn:defk2}
\end{eqnarray}
or its equivalence,
\begin{eqnarray}
\pi_*(\hat{X})=0,
\label{eqn:pi*=0}
\end{eqnarray}
where $\pi_*$ is the differential map of $\pi$.
A member of the horizontal subspace and the vertical subspace are
called a {\it horizontal vector} and {\it vertical vector},
respectively.
The image of $\hat{X}\in{\cal T}_W({\cal W}_r)$ by the projection
onto the horizontal subspace ${\cal LS}_W$
is called the {\it horizontal component},
while the image by the projection onto 
the vertical  subspace  ${\cal K}_W$ is called 
{\it vertical component}.

The {\it horizontal lift} $\hlift_W$
is a mapping from  ${\cal T}_{\pi(W)}({\cal P}_r) $
to ${\cal T}_W({\cal W}_r) $ such that
\begin{eqnarray}
\pi_*\left(\,\hlift_W(X)\,\right)=X\nonumber,\\
\hlift_W(X)\in {\cal LS}_W.
\nonumber
\end{eqnarray}
Because of the following theorem,
the matrix representation of the horizontal lift $\pi_*(\hlift_W(X))$ 
is a representation of 
the tangent vector $X\in{\cal T}_{\pi(W)}({\cal P}_r) $.

\begin{theorem}
$\hlift_W$ is a isomorphism from
${\cal T}_{\pi(W)}({\cal P}_r)$ to ${\cal LS}_W$.
\label{theorem:uqhl}
\end{theorem}

\begin{proof}
First, notice that
for any $\hat{Y}\in{\cal T}_{W}({\cal W}_r)$,
$W+ \varepsilon {\bf M}(\hat{Y})$ also is a member
of ${\cal W}_r$ if $\varepsilon$ is small enough.
Therefore, we have
\begin{eqnarray}
\pi_*(\, {\cal T}_{W}({\cal W}_r)\,)\subset{\cal T}_{\pi(W)}({\cal P}_r).
\nonumber
\end{eqnarray}

Second, we prove that 
the map $\pi_*|_{{\cal LS}_W}$ is a one to one map
from ${\cal LS}_W$ to ${\cal T}_{\pi(W)}({\cal P}_r)$.
For that, it is sufficient to prove that
$\hat{X}=0$ when  $\hat{X}\in {\cal K}_W$.
This statement is proved to be true because
${\cal K}_W$ is orthogonal to ${\cal LS}_W$.

Finally, checking the dimension of ${\cal T}_{\pi(W)}({\cal P}_r)$
is equal to ${\cal LS}_W$, we have the theorem.
\end{proof}

Using the horizontal lift, the inner product $\lgl *,*\rgl$ 
in ${\cal T}({\cal P}_r)$
is deduced from 
$\lgll *,*\rgll$:
\begin{eqnarray}
\lgl X,Y \rgl_{\pi(W)}=\llgll  \hlift_W(X),\hlift_W(Y) \rrgll_W.
\nonumber
\end{eqnarray}

The horizontal lift $\hlift$ satisfies
the following equality so that the above definition of
the inner product $\lgl *,*\rgl$ is self-consistent:
\begin{eqnarray}
 \llgll \hlift_W X,\hlift_{W} Y \rrgll
 =\llgll \hlift_{WU} X ,\hlift_{WU} Y \rrgll,\:
(U,U'\in U(n)\;).
\nonumber
\end{eqnarray}

The {\it symmetrized logarithmic derivative} (SLD, in short)
of $X\in {\cal T}_{\pi(W)}({\cal P}_r)$ is
the Hermitian operator $L^S_X$ in ${\cal H}$ defined by the equation
\begin{eqnarray}
X\rho(\theta)=\frac{1}{2}(L^S_X\rho(\theta)+\rho(\theta)L^S_X),
\label{eqn:defsldx}
\end{eqnarray}
where $\theta$ is a real parameter which is assigned to 
a member of  ${\cal P}_r $.
Iff the density operator is strictly positive, 
SLD is uniquely defined  by $(\ref{eqn:defsldx})$.
$L^S_{\partial/\partial \theta^i}$ is often denoted simply by
$L^S_i$.

SLD is closely related to the horizontal lift by the following equation:
\begin{eqnarray}
{\bf M}\left(\hlift_W{X}\right)=L^S_X W.
\label{eqn:hsld}
\end{eqnarray}

\section{Definition of Uhlmann's parallelism}
\label{sec:defuhlmann}

Berry's phase, by far confirmed by several experiments,
is a holonomy of a natural connection in the line bundle
over the space of pure states \cite{AA}\cite{Berry}. 
In 1986, Uhlmann generalized the theory to include
mixed states in the Hilbert space ${\cal H}$
\cite{Uhlmann:1986}\cite{Uhlmann:1992}
\cite{Uhlmann:1993}.
Throughout this chapter, for the sake of clarity,
$d\equiv\dim{\cal H}$ is
assumed to be finite.
For notational simplicity, the argument $\theta$ is omitted, 
as long as the omission is not misleading.

Define a {\it horizontal lift} of 
a curve $C=\{\rho(t)|t\in {\bf R}\}$ in ${\cal P}_r$
as a curve $C_h=\{W(t)|t\in {\bf R}\}$ in ${\cal W}_r$ 
which satisfies
$C=\pi(C_h)$ and
\begin{eqnarray}
\frac{dW(t)}{dt}
={\bf M}\left(\hlift_{W(t)}\left(\frac{d}{dt}\right)\, \right).
\label{horizontal}
\end{eqnarray}
Then, the {\it relative phase factor} (RPF) 
between  $\rho(t_0)$ and $\rho(t_1)$
along the curve $C$
is 
the unitary matrix $U$ which satisfies
the equation 
\begin{eqnarray}
W(t_1)=\hat{W}_1 U, 
\nonumber
\end{eqnarray}
where 
$\hat{W}_1$ 
satisfies $\rho(t_1)=\pi(\hat{W}_1)$ and
\begin{eqnarray} 
\hat{W}_1^*W(t_0)=W^*(t_0)\hat{W}_1.
\nonumber
\end{eqnarray}
RPF is said to vanish when it is equal to the identity.  

\section{RPF for infinitesimal loop}
The RPF for the infinitesimal loop
\begin{eqnarray}
\begin{array}{ccc}
(\theta^1,..., \theta^i,..., \theta^j+d\theta^j,..., \theta^m)
&\leftarrow&
(\theta^1,..., \theta^i+d\theta^i,..., \theta^j+d\theta^j,...., \theta^m)
\\
\downarrow& &\uparrow
\\
\theta=(\theta^1,..., \theta^i,..., \theta^j,...., \theta^m)
&\rightarrow&
(\theta^1,..., \theta^i+d\theta^i,..., \theta^j,...., \theta^m)
\end{array}
\label{loop}
\end{eqnarray}
is calculated up to the second order of $d\theta$
by expanding the solution of 
the equation $(\ref{horizontal})$ 
to that order:
\begin{eqnarray}
I+\frac{1}{2}W^{-1}F_{ij}W\;d\theta^i d\theta^j
+o(d\theta)^2,\nonumber\\
F_{ij}
=(\partial_i L^S_j-\partial_j L^S_i)
-\frac{1}{2}[L^S_i,L^S_j].
\label{eqn:Fij}
\end{eqnarray}
Note that $F_{ij}$ is a representation of  the curvature form,
and that RPF  for any closed loop vanishes  iff  $F_{ij}$ is zero
at any point in ${\cal M}$.

\section{The SLD Cramer-Rao inequality}
In parallel with the classical estimation theory,
 in the quantum estimation theory,
 we have the following {\it SLD CR inequality},
which is proved for the faithful state model 
by Helstrom \cite{Helstrom:1967}\cite{Helstrom:1976},
for the pure state model by Fujiwara and Nagaoka
\cite{FujiwaraNagaoka:1995}, 
and for the general case by Fujiwara and Matsumoto \cite{fujiwara2}:
\begin{eqnarray}
V_{\theta}[\hat\theta(\omega)\, |\, M]\geq(J^S(\theta))^{-1},
\label{eqn:genCR}
\end{eqnarray}
{\it i.e.}, $V_{\theta}[\hat\theta(\omega)\, |\, M]-(J^S(\theta))^{-1}$ 
is non-negative definite.
Here $ V_{\theta}[\hat\theta(\omega)\, |\, M]$ is a covariance matrix of 
an unbiased estimator  $(\hat\theta, M, \Omega)$,
and $J^S(\theta)$ is called {\it SLD Fisher information matrix}, 
and is defined by
\begin{eqnarray}
J^S(\theta)&\equiv &
\left[\left\lgl \frac{\partial}{\partial \theta^i}, 
 \frac{\partial}{\partial \theta^j}\right\rgl_{\rho(\theta)} \right],
\nonumber\\
&=& [{\rm Re}\, \tr \rho(\theta)L^S_i(\theta))L^S_j(\theta)],
\label{eqn:defgenfshr}
\end{eqnarray}
which is nothing but
 the metric tensor of the inner product $\lgl *, *\rgl$.

The inequality $(\ref{eqn:genCR})$ is of special interest,
because 
$J^{S-1}$, called {\it SLD CR bound},
 is the one of the best bounds in the sense 
explained later.

To prove the inequality $(\ref{eqn:genCR})$,
 we set some notations, and present some lemmas
For unbiased estimator $(\hat\theta, M, \Omega)$,
we define  the notation ${\bf M}$ as follows,
\begin{eqnarray}
{\bf M}^i(\hat\theta,\, M,\,W)
\equiv\int (\hat\theta^i(\omega)-\theta^i) M(d\omega)W.
\nonumber
\end{eqnarray}
$Z[\hat\theta, M]$ is $m\times m$ matrix defined by
\begin{eqnarray}
Z[\hat\theta, M]=
\left[\; \tr\:{\bf M}^i(\hat\theta,\, M,\,W)\,
		({\bf M}^j(\hat\theta,\, M,\,W)\,)^*\; \right]
\label{eqn:defz}
\end{eqnarray}

\begin{lemma}
Following two inequalities are valid:
\begin{equation}
  V[\hat\theta\, |\, M]\ge {\rm Re}\, Z[\hat\theta,M].
\label{eqn:genVZ}
\end{equation}
\begin{equation}
 V[\hat\theta\, |\, M] \ge Z[\hat\theta,M].
\label{eqn:genVZ2}
\end{equation}
\label{lemma:genVZ}
\end{lemma}

\begin{lemma}
\begin{eqnarray}
{\rm Re}\, Z[\hat\theta,M]\ge J^{S-1}
\label{eqn:genZJ}
\end{eqnarray}
holds. The equality is valid iff
\begin{eqnarray}
{\bf M}^j(\hat\theta,\, M,\,W)
 &=&\sum_k [J^{S-1}]^{j,k} L^S_k W,
\nonumber\\
 &=&\sum_k [J^{S-1}]^{j,k} 
{\bf M}\left(\hlift_W \frac{\partial}{\partial \theta^k} \right),
\label{eqn:genxx=J}
\end{eqnarray}
\label{lemma:genZJ}
\end{lemma}

They are proved in almost the same manner as the strictly positive case
(see Ref.\cite{Holevo:1982} p.88 and p.274 respectively).
Lemma $\ref{lemma:genVZ}$ and $\ref{lemma:genZJ}$ lead to 
the SLD CR inequality $(\ref{eqn:genCR})$.

\begin{theorem}
SLD Fisher information gives a lower bound of covariance matrix
of an unbiased measurement, {\it i.e.},
$(\ref{eqn:genCR})$ holds true.
\end{theorem}

The SLD CR bound $(\ref{eqn:genCR})$ is the best bound 
in the following sense.
\begin{theorem}
Letting $A$ be a real hermitian matrix which is larger than $J^{S-1}$,
that is, $A>J^{S-1}$, there exists 
such an unbiased estimator $(\hat\theta, M, \Omega)$
that $V[\hat\theta\, |\, M]$ is not smaller than $A$.
\label{theorem:gensldbest}
\end{theorem}

\begin{proof}
Let ${\bf v}\in{\R}^m$ be the real vector 
such that
\begin{eqnarray}
\exists C\subset \R\;\;\;
\frac{\partial}{\partial \theta^i} 
\int_C \tr \rho(\theta) E_{\bf v}(dt)\neq 0,
\: (i=1,...,m),
\label{eqn:dpneq0}
\end{eqnarray}
where $E_{\bf v}$ is a projection valued measure obtained
by the spectral decomposition
of $\sum_{i,j} v_j(\theta^j I + [J^{S-1}]^{i,j} L_j^S)$,
where $I$ and $v_i$ denotes identity and $i$th component of ${\bf v}$.

The condition $(\ref{eqn:dpneq0})$ is
implies the existence of  an estimator  $\hat\theta_{\bf v}(\omega)$ 
which makes
the triplet $(\hat\theta_{\bf v}, E_{\bf v}, \Omega)$ 
locally unbiased at $\theta$.
For that triplet $(\hat\theta, E_{\bf v}, \Omega)$,
we have
\begin{eqnarray}
{\bf v}^T \left(V[\hat\theta\, |\, E_{\bf v}]-J^{S-1}\right){\bf v}=0.
\nonumber
\end{eqnarray}

If $\varepsilon$ is small enough, 
for any real vector ${\bf v}\in{\R}^m$,
${\bf v}+\varepsilon{\bf v}_0 $ satisfies 
the condition $(\ref{eqn:dpneq0})$,
or its equivalence,
\begin{eqnarray}
({\bf v}+\varepsilon{\bf v}_0 )^T
\left(V[\hat\theta_{{\bf v}+\varepsilon{\bf v}_0 }\, |
\, E_{{\bf v}+\varepsilon{\bf v}_0 }]-J^{S-1}\right)
({\bf v}+\varepsilon{\bf v}_0) =0.
\label{eqn:vevjve}
\end{eqnarray}

Let us assume that there exists a real matrix $A$ which satisfies
\begin{eqnarray}
V[\hat\theta|M]\ge A >J^{S-1}
\nonumber
\end{eqnarray}
 for any unbiased estimator.
Then, by virtue of $(\ref{eqn:vevjve})$, we have
for any real vector ${\bf v}_0\in{\R}^m$ 
and enough small $\varepsilon$,
\begin{eqnarray}
({\bf v}+\varepsilon{\bf v}_0 )^T
\left(A-J^{S-1}\right)
({\bf v}+\varepsilon{\bf v}_0) =0,
\nonumber
\end{eqnarray}
whose second derivative with respect to $\varepsilon$
yields
\begin{eqnarray}
{\bf v}_0 ^T\left(A-J^{S-1}\right){\bf v}_0 =0.
\nonumber
\end{eqnarray}
Because  ${\bf v}_0\in{\R}^m$ is arbitrary, we have
\begin{eqnarray}
A-J^{S-1}=0,
\nonumber
\end{eqnarray}
which contradicts with the assumption $A>J^{S-1}$.
\end{proof}

\begin{theorem}
If the model ${\cal M}$ has only one parameter,
the equality in $(\ref{eqn:genCR})$ is achievable.
\label{th:gen1pCR}
\end{theorem}
\begin{proof}
Let $E$ and $\hat\theta(*)$ be a projection valued measurement
and an estimator which satisfies,
\begin{eqnarray}
\left[\int(\hat\theta(\omega)-\theta)M(d\omega)\right]^i
=\sum_j \left[J^{S-1}\right]^{ji}L^S_j.
\nonumber
\end{eqnarray}
Then, the triplet $(\hat\theta, E, \Omega)$ is locally unbiased at $\theta$
and attains SLD CR bound.
\end{proof}

$(\ref{th:gen1pCR})$ 
implies the statistical significance of the natural metric
$\rgl *,*\lgl$.
A possible geometrical interpretation of $(\ref{th:gen1pCR})$:
the closer 
	two states
			$\rho(t)$ and $\rho(t+dt)$
				are,
the harder
	it is
	  to distinguish $\rho(t)$ 
		 from $\rho(t+dt)$.

The SLD CR inequality $(\ref{eqn:genCR})$ looks quite analogical to
CR inequality in classical estimation theory.
However, as will be  found out later,
the equality is not generally attainable.

\section{The attainable Cramer-Rao type bound}
\label{sec:nonclassical}

In the previous section,
SLD CR bound  is proved to be the best bound.
However, as will be turned out, this best bound 
is attainable only in the special cases, that is,
the case when the model is locally quasi-classical. 
In general case, therefore, we must give  up   
to find a tight lower bound of covariance matrix 
in the form of the matrix inequality.
Instead, we determine the region ${\cal V}_{\theta}({\cal M})$ 
of the map $V[*]$ 
from unbiased estimators 
to $m \times m$ real positive symmetric matrices
(so far as no confusion is expected, we write ${\cal V}$ 
for ${\cal V}_{\theta}({\cal M})$).
Especially, the boundary $bd {\cal V}$ is of  interest,
because  ${\cal V}$ is convex as in the following lemma.

\begin{lemma}
${\cal V}$ is convex.
\label{lemma:vconvex}
\end{lemma}
\begin{proof}
In this proof, we define the estimator to be
the  measurement which takes value in $\R^m$.
Let $M$ and $M'$ be an unbiased estimator.
Because 
\begin{eqnarray}
\lambda V[M] + (1-\lambda) V[M']=V[\lambda M+(1-\lambda)M']
\nonumber
\end{eqnarray}
holds true and $\lambda M+(1-\lambda)M'$ is an unbiased estimator,
we have the lemma.
\end{proof}

\begin{lemma}
If a matrix $V$ is a member of ${\cal V}$,
$V+V_0$ is also a member of  ${\cal V}$
for any arbitrary nonnegative real symmetric matrix $V_0$.
\label{lemma:vunbound}
\end{lemma}
\begin{proof}
Let $(\hat\theta, M, \Omega)$ be a locally unbiased estimator
whose covariance matrix is $V$,
and define $\varepsilon(\omega)$ by
\begin{eqnarray}
\varepsilon(\omega)\equiv
V_0^{1/2}(V[\hat\theta|M])^{-1/2}\left(\hat\theta(\omega)-\theta\right).
\nonumber
\end{eqnarray}
Then,  
$\hat\theta^*(\omega)=\hat\theta(\omega)+\varepsilon(\omega)$ is 
also a locally unbiased estimator, and its covariance matrix 
is  equal to $V+V_0$.
\end{proof}

To obtain $bd {\cal V}$, 
the following procedure is used in this thesis. 
Define an inner product of two real symmetric matrices
$A$ and $B$ by $\Tr AB$. 
Then, the set $\{V| \Tr VG=const.\}$ is a hyperplain 
perpendicular to the vector $G$.
Because of lemmas $\ref{lemma:vconvex}$-$\ref{lemma:vunbound}$,
$bd{\cal V}$ is the collection of all the matrices 
$V\in{\cal V}$
which  achieve the minimum of $\Tr GV$		
for a certain symmetric real nonnegative definite matrix $G$.

However, when the model has too many parameters,
the dimension of the space $Sym(m)$ of real symmetric matrices 
is so  large that  ${\cal V}_{\theta}({\cal M})$ 
is extremely hard to  determine.
In such cases, we calculate  
\begin{eqnarray} 
\CR (G,\theta,{\cal M})\equiv
\min\{\Tr GV\: |\: V\in{\cal V}_{\theta}({\cal M})\}
\label{eqn:crtype}
\end{eqnarray}
for an arbitrary nonnegative symmetric real matrix $G$,
and call it the {\it attainable CR (Cramer-Rao) type bound}.
The matrix $G$ is called {\it weight matrix}.
Often, we drop $G$,  $\theta$ and/or ${\cal M}$ 
when no confusion is expected.
If $\CR (G,\theta,{\cal M})$ is smaller than 
$\CR (G,\theta',{\cal M}')$ for any weight matrix $G$, 
the ${\cal V}_{\theta}({\cal M})$ of 
is located in the `lower part' of $Sym(m)$
compared with that of ${\cal V}_{\theta'}({\cal M}')$.

To make the estimational meaning of $(\ref{eqn:crtype})$ clear, let us
consider a  diagonal weight matrix $G=diag(g_1,g_2...,g_m)$. 
Letting $v_{ii}$ be the $(i,i)$-th component of $V[M]$,
\begin{eqnarray}
        \Tr GV[M]=\sum_i g_i v_{ii}, 
\nonumber
\end{eqnarray}
is the weighed sum of the variances of the estimations of
the parameters $\theta^i\,(i=1,...m)$. 
 If the accuracy of estimation of, for example, the parameter $\theta^1$ 
is required
more than other parameters,
then $g_1$ is set larger than any other $g_i$,
and the estimator which minimize $\sum_i g_i v_{ii}$ is to be used.

\section{A conjecture about the quantum MLE}

In the classical estimation theory,
the inverse of the Fisher information matrix is attained
globally up to the order of $1/N$, where $N$ is the number of data.
An example of estimators which attain the bound is 
the maximum likelihood estimator.
On the other hand, in the quantum estimation theory,
Nagaoka $\cite{Nagaoka:1989:2}$ conjectured that
the $G$-quantum maximum likelihood estimate, 
defined below, has similar property.

Suppose that $N$ copies of the state of the system are given.
Then, we  define $\hat\theta_{(k)}$ recursively by the equation
\begin{eqnarray}
\hat\theta_{(k)}
= \argmax\left\{
\left.
\sum_{i=1}^k \ln p_i(x_i\: | \: \theta)\: \right|
\: \theta\in\R^m
\right\},
\nonumber
\end{eqnarray}
where 
\begin{eqnarray}
&&p_i(x\: | \: \theta)dx = \tr \rho(\theta)M_{(i)}(dx),\\
&&M_{(i)}=\left\{
\begin{array}{cc}
\argmin\{\tr GV_{\theta_{(i-1)}}[M]\: |\: \mbox{ $M$ is locally unbiased}\},
		& (i\geq 2)\\
\mbox{an arbitrary locally unbiased measurement}& (i=1)
\end{array}
\right.
\nonumber
\end{eqnarray}
and $x_i$ is the data produced by  $M_{(i)}$ from $i$th copy. 
When the number of sample is equal to $N$ and the weight matrix is $G$, 
the $G$-quantum maximum likelihood estimate $\hat\theta_{GqMLE}(x_1,...,x_N)$
is defined to be $\hat\theta_{GqMLE}(x_1,...,x_N)=\hat\theta_{(N)}$.

\begin{conjecture}
$\hat\theta_{GqMLE}$ is consistent, or
\begin{eqnarray}
\lim_{N\rightarrow\infty} E_{\theta}[\hat\theta_{GqMLE}(x_1,...,x_N)]=\theta.
\nonumber
\end{eqnarray}
\end{conjecture}
\begin{conjecture}
\begin{eqnarray}
\Tr G V_{\theta}[\hat\theta_{GqMLE}(x_1,...,x_N)]=
\frac{1}{N}\CR (G, \theta)+o\left(\frac{1}{N} \right).
\nonumber
\end{eqnarray}
\end{conjecture}

\part*{Part I \\ 
The faithful model theory}
\addcontentsline{toc}{chapter}%
{Part I :
The faithful model theory}

\chapter{The estimation theory of the faithful model}
\section{The locally quasi-classical model and the quasi-classical model}
\label{sec:questf}
In this section, 
the condition for SLD CR bound to be attainable is reviewed briefly,
in the case of 
the {\it faithful model},
any member of which is
faithful, {\it i.e.}, a reversible operator.
We denote the space of the faithful states by 
${\cal P}_+({\cal H})$, or simply by ${\cal P}_+$.
For mathematical simplicity, the dimension $d$ of the Hilbert space
${\cal H}$ is assumed to be finite.
The author believes the essence of the discussion will not be damaged
by this restriction.

As for the equality in  $(\ref{eqn:genCR})$ 
in the faithful model,
we have the following theorem, which is proved 
by Nagaoka \cite{Nagaoka:1996}.
\begin{theorem}
The equality in $(\ref{eqn:genCR})$ is attainable at $\theta$
iff $[L^S_i(\theta),L^S_j(\theta)]=0$ for any $i,j$.
Letting
$|\omega\rgl$ be a simultaneous eigenvector of 
the matrices $\{L^S_j (\theta)|j=1,...,m\}$
and 
$\lambda_i(\omega)$ be the eigenvalue of $L^S_i (\theta)$
corresponding to $|\omega\rgl$,
the equality is attained by the estimator
$(\hat\theta_{(\theta)},M_{(\theta)},\Omega)$
such that
\begin{eqnarray}
&&\Omega=\{\omega| \omega=1,...,n\},\nonumber\\
&&M_{(\theta)}(\omega)=|\omega\rgl\lgl \omega|,\nonumber\\
&&\hat\theta_{(\theta)}^j(\omega)
= \theta^j + \sum_{k=1}^d [(J^S)^{-1}]^{jk} \lambda_k(\omega).
\label{eqn:bests}
\end{eqnarray}
\label{theorem:mpCR}
\end{theorem}

The model ${\cal M}$ is said to be 
{\it locally quasi-classical} at $\theta$ iff
$L^S_i(\theta)$ and $L^S_j(\theta)$ commute for any $i,j$,
because in this case $(\ref{eqn:genCR})$
gives the attainable lower bound of the covariance matrix
of the unbiased estimator,
as its classical counterpart does.
However, it should be noted that
even if the model is locally quasi-classical at any $\theta\in\Theta$,
theorem $\ref{theorem:mpCR}$ do not tell us 
the optimal experiment scheme,
because
the measurement $M_{(\theta)}$ in $(\ref{eqn:bests})$
generally depends on the unknown parameter $\theta$ 
(so does the best experiment scheme).

Therefore, let us move to easier case.
Suppose  that
$L^S_i(\theta)$ and $L^S_j(\theta')$
commute even when $\theta\neq\theta'$,
in addition to being locally quasi-classical at any $\theta\in\Theta$.
Then,
the measurement $M_{(\theta)}$ in $(\ref{eqn:bests})$,
denoted by $M_{best}$ in the remainder of the section,
is uniformly optimal for all $\theta$
(so is the corresponding scheme).
We say such a model is {\it quasi-classical} \cite{Yung}.

After the best experiment is done,
the rest of our task is to estimate the value of the parameter $\theta$
in the probability distribution $(\ref{eqn:pdm})$, 
where $M_{best}$ is substituted into $M$.
Hence, in this case, the quantum estimation  reduces to
the classical estimation.

\section{Estimation of the temperature}
Because quantum mechanics is applicable to 
the measuring process  of the temperature,
the measurement of the temperature can be described 
with the generalized measurement.
Therefore, we formulate the measurement of the temperature
as estimation of the parameter $T$
in the {\it canonical  model}
  ${\cal M}$  which is defined by
\begin{eqnarray}
{\cal M}=\left\{\rho(T)\: 
	\left|\: 
\rho (T)= \exp\left(-\frac{1}{k_B T}(H-F(T)\,)\right),\, 
T\in \R \right.\right\},
\nonumber
\end{eqnarray}
where $k_B$ is the Boltzmann constant, $H$ the Hamiltonian of the system,
$F(T)$ the free energy.

Simple calculations yield he SLD Fisher information $J^S(T)$,
\begin{eqnarray}
J^S(T)=\frac{1}{k_B T^2}C(T),
\nonumber
\end{eqnarray}
and the best estimator,
\begin{eqnarray}
&&\Omega =\{\omega \: |\: \omega=1,...,d\},\\
&&M_{T} (\omega) = |\omega \rgl\lgl \omega |, \\
&&\hat{T}_{T}(\omega) = T+ \frac{E_{\omega}-\lgl H \rgl_T}{C(T)},
\nonumber
\end{eqnarray}
where  $C(T)$ is the specific heat of the system,
$|\omega\rgl$ the $\omega$th eigenvector of the Hamiltonian,
$E_{\omega}$ the $\omega$th eigenvalue of the Hamiltonian,
and $\lgl H \rgl_T$ denotes $\tr \rho(T) H$. 

Because the model is quasi-classical,
by use of the maximum likelyhood estimator of the parameter $T$
in  the induced family of probability distributions,
\begin{eqnarray}
\{p(\omega|T)\: |\: p(\omega|T)=\lgl \omega| \rho(T)|\omega \rgl \}, 
\nonumber
\end{eqnarray}
we can attain asymptotically the SLD CR bound.

When the temperature is high, 
the specific heat behaves like constant independent of $T$,
which implies that SLD Fisher information tends to zero
as $T$ tends to infinity.
In other words, in the high temperature, the parameter 
$T$ is hard to estimate.

In the low temperature, the dependency of the specific heat on $T$ 
differs from system to system, 
and so does  the dependency of $J^S(T)$.

For example, if the system is the mass of the phonon,
\begin{eqnarray}
C(T)\propto T^3,
\nonumber
\end{eqnarray}
which implies that SLD Fisher information tends to zero
as $T$ tends to zero.
On the other hand, the electronic specific heat is
\begin{eqnarray}
C(T)\propto T ,
\nonumber
\end{eqnarray}
which implies that SLD Fisher information tends to infinity
as $T$ tends to zero.

\chapter{Uhlmann connection and the estimation theory}
\label{chap:uhlmann}

\section{Some new facts about RPF}
\label{section:RPF}
In this section, 
we derive conditions for RPF to vanish, which is used to
characterize 
the classes of manifold defined in the previous section.
For notational simplicity, the argument $\theta$ is omitted, 
as long as the omission is not misleading.

\begin{theorem}
RPF  for any closed loop vanishes
iff $[\, L^S_i(\theta),\, L^S_j(\theta)\,]=0$ for any $\theta \in \Theta$.
In other words,
\begin{eqnarray}
F_{ij}(\theta)=0\Longleftrightarrow
\left[\, L^S_i(\theta),\, L^S_j(\theta)\,\right]=0.
\label{eqn:f0lsls0}
\end{eqnarray}
\label{theorem:f0lsls0}
\end{theorem}

\begin{proof}
If $F_{ij}$ equals zero,  
then both of the two terms in the left-hand side of $(\ref{eqn:Fij})$ 
must vanish,
because the first term is Hermitian and the second term is skew Hermitian.
Hence, if $F_{ij}=0$, $[\, L^S_i,\, L^S_j\, ]$ vanishes.

On the other hand, the identity
$\partial_i\partial_j\rho-\partial_j\partial_i\rho=0$, 
or its equivalence
\begin{eqnarray}
\left(\partial_iL^S_j-\partial_jL^S_i 
-\frac{1}{2}[ L^S_i,\, L^S_j ]\right)\rho+
\rho\left(\partial_iL^S_j-\partial_jL^S_i +\frac{1}{2}[ L^S_i,\, L^S_j ]\right)
=0,\nonumber
\label{eqn:didj}
\end{eqnarray}
implies
that
$\partial_i L^S_j-\partial_j L^S_i$ vanishes if $[\, L^S_i,\, L^S_j\, ]=0$,
because 
$\partial_i L^S_j-\partial_j L^S_i$ is Hermitian
and $\rho$ is positive definite.
Thus we see $F_{ij}=0$  if $L^S_i$ and $L^S_j$ commute.
\end{proof}

A manifold ${\cal M}$ is said to be {\it parallel} 
when the RPF between any two points along any curve vanishes.
From the definition, if ${\cal M}$ is parallel, RPF along
any closed loop vanishes, but the reverse is not necessarily true.
The following theorem is a generalization of Uhlmann's 
theory of $\Omega$-horizontal real plane \cite{Uhlmann:1993}.

\begin{theorem}
The following three conditions are equivalent.
\begin{itemize}
\item[{\rm (1)}] ${\cal M}$ is parallel.
\item[{\rm (2)}] Any element $\rho(\theta)$ of ${\cal M}$ writes
		$M(\theta)\,\rho_0\, M(\theta)$,
		where $M(\theta)$ is Hermitian  and 
		$M(\theta_0)\,M(\theta_1)=M(\theta_1)\,M(\theta_0)$
		for any $\theta_0, \theta_1 \in\Theta$.
\item[{\rm (3)}]$
		\forall i,j,\:\forall \theta_0,\,\theta_1\in\Theta,\:\;
		\left[\,L^S_i(\theta_0),\, L^S_j(\theta_1)\,\right]=0.
		$
\end{itemize}
\label{theorem:m-inphase}
\end{theorem}

\begin{proof}
Let $W(\theta_t)=M(\theta_t)W_0$ be a horizontal lift of
$\{\rho(\theta_t)\, |\, t\in{\bf R}\}\subset{\cal M}$.
Then, $W^*_0\, W(\theta_t)=W^*(\theta_t)\, W_0$ implies
$M(\theta_t)=M^*(\theta_t)$, and
$W^*(\theta_{t_0})\, W(\theta_{t_1})
=W^*(\theta_{t_1})\, W(\theta_{t_0})$ implies 
$M(\theta_{t_0})\, M(\theta_{t_1})=M(\theta_{t_1})\, M(\theta_{t_0})$.
Thus we get $(1)\Rightarrow(2)$. Obviously, the reverse also holds true.
For the proof of $(2)\Leftrightarrow(3)$,
see Ref. \cite{Yung}, pp.31-33.
\end{proof}

\section{Uhlmann's parallelism in the quantum estimation theory}
\label{sec:estphase}
To conclude the chapter, we present the theorems which 
geometrically characterize 
the locally quasi-classical model and
quasi-classical model,
by the  vanishing conditions of RPF, 
implying the close tie between
Uhlmann parallel transport and the quantum estimation theory.
They are straightforward consequences of
the definitions of the terminologies  and
theorems \ref{theorem:mpCR} -\ref{theorem:m-inphase}.
\begin{theorem}
${\cal M}$ is locally quasi-classical at $\theta$
iff $F_{ij}(\theta)=0$ for any $i,j$.
${\cal M}$ is locally quasi-classical at any $\theta\in\Theta$
iff the RPF for any loop vanishes.
\label{th:estuhl1}
\end{theorem}
\begin{theorem}
 ${\cal M}$ is quasi-classical
iff ${\cal M}$ is parallel.
\end{theorem}

\chapter{Nagaoka's quantum information geometry 
and Uhlmann's parallelism}
\section{Nagaoka's quantum information geometry}

In this section, we give brief review of 
Nagaoka's quantum information geometry,
which is another geometrical theory of the quantum 
statistical model than Uhlmann's parallelism.

In the Nagaoka's geometry, metric tensor is 
chosen to be SLD Fisher information matrix.
Letting ${\bf L}^S$ denote the linear mapping  from
the tangent vector to its SLD,
{\it $e$-parallel transport} $\pte$ is 
defined as follows:
\begin{eqnarray}
{\bf L}^S\left(\pte_{\rho\rightarrow\sigma} X\right)
&\equiv&{\bf L}^SX  
-\tr (\sigma {\bf L}^S X).
\label{eqn:ptte}
\end{eqnarray}
Note that in the faithful model, ${\bf L}^S$ is one to one mapping,
and the equation $(\ref{eqn:ptte})$ defines the connection
uniquely.
The dual $\ptm$ of $e$-parallel transport
with respect to SLD inner product 
is called {\it $m$-parallel transport},
\begin{eqnarray}
\forall X,Y\in{\cal T}_{\rho}({\cal P}_+)\:\:
\left\lgl \ptm_{\rho\rightarrow\sigma}X,
\pte_{\rho\rightarrow\sigma}Y \right\rgl_{\sigma}
=\lgl X,Y \rgl_{\sigma}.
\nonumber
\end{eqnarray}

For the  autoparallel manifold in $e$-connection,
the $e$- covariant derivative is calculated as
\begin{eqnarray}
{\bf L}^S\left(\,\nablae_X Y\,\right)=XL^S_Y-\tr \rho XL^S_Y,
\nonumber
\end{eqnarray}
and  
the tortion of $e$-connection $\Te$ as,
\begin{eqnarray}
{\bf L}^S 
\left(\,\Te(X,Y)\,\right)
&=&{\bf L}^S\left(\,\nablae_X Y-\nabla_Y X -[X,Y]\,\right)\nonumber\\
&=& X{\bf L}^S(Y)-Y{\bf L}^S(X),
\label{eqn:etortion}
\end{eqnarray}
or, equivalently,
\begin{eqnarray}
\Te(X,Y)\rho(\theta)
&=&\frac{1}{4}[[L^S_X,L^S_Y], \rho],
\label{eqn:etortion2}
\end{eqnarray}
where $X$ and $Y$ are understood as differential operators.

Nagaoka showed that 
these $e$- and $m$-connections 
nicely characterize the estimation theoretical properties of models,
in a different manner than Uhlmann's parallelism.
Is there any relation between the two geometrical structures?

\section{$w$-connection in ${\cal W}$}

To elucidate the relations between Uhlmann's parallelism and 
Nagaoka's information geometry, we consider
the geometry of the tangent bundle over the total space ${\cal W}_d$.

the logarithmic derivative ${\bf L}(\hat{X}) $ of 
$X\in{\cal T}_W({\cal W}_d)$ is a $d\times d$ complex matrix 
which satisfies the equation,
\begin{eqnarray}
{\bf L}(\hat{X}) W
={\bf M}(\hat{X}).
\label{eqn:deflogd}
\end{eqnarray}
${\bf L}(\partial/\partial\zeta^i)$ is often denoted by
$L_{\partial/\partial \zeta^i}$ for simplicity.
Notice the logarithmic derivative is uniquely defined iff
the rank $r$ of $W$ is equal to the dimension $d$ of 
the Hilbert space ${\cal H}$.

The real span of the logarithmic derivatives of at $\zeta$ is
\begin{eqnarray}
\{L\: |\: {\rm Re}\,\tr(L \pi(W(\zeta)))=0,
	L\in M_d(\C)\}.
\nonumber
\end{eqnarray}

We introduce the {\it $w$-connection} in ${\cal T}({\cal W})$
by the {\it $w$-parallel transport} defined by
\begin{eqnarray}
{\bf L}\left(\pttw_{W_0\rightarrow W_1}\hat{X}\right)
\equiv {\bf L}\hat{X}  -{\rm Re}\,\tr(\pi(W_1){\bf L}\hat{X}).
\label{eqn:defpttw}
\end{eqnarray} 
This $w$-parallel transport is left invariant 
 under the action of $U\in U(d)$ in the following sense:
\begin{eqnarray}
{\bf L}\left(\pttw_{W_0\rightarrow W_1}\hat{X}\right)
={\bf L}\left(\pttw_{W_0U\rightarrow W_1}\hat{X}\right)
={\bf L}\left(\pttw_{W_0\rightarrow W_1U}\hat{X}\right).
\nonumber
\end{eqnarray}

The covariant derivative $\nablaw$ 
and the tortion $w$-tortion $\Ttw(\hat{X},\hat{Y})$
are  easily calculated for the $w$-autoparallel 
submanifold of ${\cal W}$ as
\begin{eqnarray}
{\bf L}\left(\nablaw_{\hat{X}}\hat{Y} \right)
&=&\hat{X}{\bf L}(\hat{Y})
-{\rm Re}\,\tr \left(\pi(W)\hat{X}{\bf L}(\hat{Y})\right)I
\nonumber\\
{\bf L}\Ttw(\hat{X},\hat{Y})
&=&
{\bf L}\left(\,\nablaw_{\hat{X}} \hat{Y}-\nablaw_{\hat{Y}}\hat{X}
-[\hat{X},\,\hat{Y}]\,\right)\nonumber\\
&=&\frac{1}{2}[L_X,L_Y]
-{\rm Re}\,\tr \left\{\pi(W)(\hat{X}L_Y-\hat{Y}L_X)\right\}I,
\nonumber\\
\label{eqn:defTtwmtrx}
\end{eqnarray}
where 
the tangent vector $\hat{X}$ is understood as a differential operator,
and $I$ is the identity in the Hilbert space ${\cal H}$.
The latter equation is equivalent to
\begin{eqnarray}
{\bf M}\Ttw(\hat{X},\hat{Y})
=
\frac{1}{2}[L_X,L_Y]W
-{\rm Re}\,\tr \left\{ \pi(W)(\hat{X}L_Y-\hat{Y}L_X)\right\}W,
\label{eqn:mttw}
\end{eqnarray}
which is of use when the theory is generalized to non-faithful models.

\section{Projection of geometric structures}
In the beginning, we show that the Nagaoka's information geometry
is naturally induced from the geometry of the ${\cal T}({\cal W})$.

As in the definition,
the metric $\lgl *,* \rgl$ in Nagaoka's quantum information 
geometry is  induced from 
the natural metric $\lgll *,* \rgll$ in ${\cal T}({\cal W})$.

Not only the metric $\lgl *,* \rgl$, but also
the transport $\pte$ is  induced from 
$\pttw$ :
\begin{eqnarray}
\begin{array}{ccc}
         \rho=\pi(W)&         &\sigma=\pi(V)\\
  &\pttw&               \\
 \ds{\hat{X}=\hlift_{W}(X)\in{\cal T}_W({\cal W})}&--\longrightarrow
&\ds{\pttw_{W\rightarrow V}}\hat{X}\in{\cal T}_V({\cal W})
\\
\hlift\uparrow& &\downarrow\pi_*
\\
X\in{\cal T}_{\rho}({\cal P}_+)&--\longrightarrow
&\ds{\pte_{\rho\rightarrow\sigma}}
X\in{\cal T}_{\sigma}({\cal P}_+)
\\
  &\pte&
\end{array}
\label{eqn:ptwte}
\end{eqnarray}

The horizontal lift $\hlift$ satisfies
the following requirements so that the definition of
the transport $\pte$ by the diagrams above are
consistent:
\begin{eqnarray}
 \pi_*\left(\pttw_{W\rightarrow V}\hlift_{W}(X)\:\right)
       =\pi_*\left(\pttw_{WU\rightarrow VU'}\hlift_{WU}(X)\;\right),
\nonumber
\end{eqnarray}

Because of the diagram,
it is quite easy to see that
if the submanifold ${\cal N}$ of ${\cal W}$
is $w$-autoparallel,
the model ${\cal M}=\pi({\cal N})$ is $e$-autoparallel.

Some elementary calculations leads to  the following theorem,
which illustrates the relation between the two geometries.

\begin{theorem}
Let ${\cal M}$ be a submanifold of ${\cal P}_+$
which is induced from $w$-autoparallel submanifold 
${\cal N}$ of ${\cal W}$ by 
${\cal M}=\pi({\cal N})$
and 
$X$ and $Y$  tangent vectors to ${\cal P}_+$ such that,
\begin{eqnarray}
X&=&\sum_i x^i\frac{\partial}{\partial\theta^i},\nonumber\\
Y&=&\sum_i y^i\frac{\partial}{\partial\theta^i}.
\nonumber
\end{eqnarray}
$\Ttw (\hlift(X),\hlift(Y))$
is decomposed into the sum such that 
\begin{eqnarray}
{\bf L}\Ttw (\hlift(X),\hlift(Y))
&=&{\bf L}^S\Te(X,Y) - F_{XY},\nonumber\\
F_{XY}&=& \sum_{i,j}F_{i j}x^i y^j,
\nonumber
\end{eqnarray}
where ${\bf L}^{-1}({\bf L}^S\Te(X,Y))$ is a horizontal vector
and ${\bf L}^{-1} F_{XY}$ is a vertical subspace.
\label{theorem:lslsdecom}
\end{theorem}

In other words, the horizontal component of the $w$-tortion  is
the $e$-tortion and the vertical component is 
the curvature form of the Uhlmann parallelism.

\chapter{The duality between the observed system and the hidden system}
\section{The duality of SLD and RLD}
First, we define the {\it right logarithmic derivative} (RLD, in short),
which played quite important role 
in the estimation theory of the Gaussian model,
which is a superposition of the coherent states by the Gaussian kernel
( see Ref. \cite{YuenLax:1973}).

RLD $L^R_i(\theta)$ of the parameter $\theta^i$ is defined by
the equation
\begin{eqnarray}
  \frac{\partial \rho(\theta)}{\partial \theta^i}
=L^R_i(\theta)\rho(\theta),
\nonumber
\end{eqnarray}
and RLD ${\bf L}^R X$ of the tangent vector 
by the equation
\begin{eqnarray}
\hat{X}\rho(\theta)
=({\bf L}^R X)\rho(\theta).
\nonumber
\end{eqnarray}

Our question is why we need two types of logarithmic derivatives,
 SLD and RLD namely, and what the relations between them.
To answer the question, we interpret
the total space ${\cal W}_d$,
like in the section $\ref{sec:defuhlmann}$,
as the space of the state vector $|\Phi\rgl$ 
in the bigger Hilbert space ${\cal H}\otimes{\cal H}'$,
where we took the dimension of ${\cal H}'$ to be $d$, 
the dimension of ${\cal H}$.
Let us call ${\cal H}$ the observed system, and ${\cal H}'$ the hidden system,
and the partial trace over ${\cal H}$ and ${\cal H}'$ are 
denoted by $\pi$ and $\pi'$, respectively.

In terms of $W$, $\pi$ and $\pi'$ write
\begin{eqnarray}
\pi(W)&=&WW^*=\rho,\nonumber\\
\pi'(W)&=&W^*W=\sigma,
\nonumber
\end{eqnarray}
and $\pi_*$ and $\pi'_*$ write
\begin{eqnarray}
\pi_*(\hat{X})\rho&=&
	\frac{1}{2}(\:
	({\bf M}\hat{X})W^*+W({\bf M}\hat{X})^*\: ),\nonumber\\
\pi'_*(\hat{X})\sigma&=&
	\frac{1}{2}(\:
	({\bf M}\hat{X})^*W+W^*({\bf M}\hat{X})\: ),
\label{eqn:defpi'*}
\end{eqnarray}

The RLD subspace ${\cal LR}_W$ of ${\cal T}_W({\cal W}_d)$ is
the space of all vectors which satisfies,
\begin{eqnarray}
{\bf L}(\hat{X})={\bf L}^R(\pi_*(\hat{X})),
\label{eqn:defLR1}
\end{eqnarray}
or, its equivalence,
\begin{eqnarray}
	({\bf M}\hat{X})W^*=W({\bf M}\hat{X})^*.
\label{eqn:defLR2}
\end{eqnarray}

Then, from $(\ref{eqn:defpi'*})$,
$(\ref{eqn:defLR2})$ and $(\ref{eqn:defLS2})$,
we have
\begin{eqnarray}
{\cal LR}_W&=&
\{\hat{X}\; |\;{\bf L}(\hat{X})={\bf L}^S(\pi'_*(\hat{X}))\},\\
{\cal LS}_W&=&
\{\hat{X}\; |\;{\bf L}(\hat{X})={\bf L}^R(\pi'_*(\hat{X}))\}.
\label{eqn:dualsr}
\end{eqnarray}
In other words, looking from the hidden system,
the RLD subspace looks like the horizontal subspace,
and the horizontal subspace looks like the RLD subspace.
We call the fact $(\ref{eqn:dualsr})$ 
the {\it duality between SLD and RLD}.

The orthogonal complement subspace ${\cal SK}_W$ of the RLD subspace
is also 
dual of   ${\cal K}_W$ in the following sense.
${\cal SK}_W$ is the space of all the tangent vectors which satisfies
\begin{eqnarray}
{\bf L}\hat{X}= -({\bf L}\hat{X})^*
\label{eqn:defsk1}
\end{eqnarray}
or, equivalence,
\begin{eqnarray}
({\bf M}\hat{X})W^*+W({\bf M}\hat{X})^*=0,
\label{eqn:defsk2}
\end{eqnarray}
which yields
\begin{eqnarray}
\pi'_*(\hat{X})=0.
\label{eqn:pi'*=0}
\end{eqnarray}
Metaphorically speaking,
$(\ref{eqn:pi*=0})$ and $(\ref{eqn:pi'*=0})$
implies that ${\cal SK}_W$ looks like ${\cal K}_W$
seen from the hidden system.

$(\ref{eqn:defsk1})$ means that 
for any member $\hat{X}$ of ${\cal SK}_W$,
$\pi_* (\hat{X})$ corresponds to  a unitary motion of the observed system.
The dual of this statement is also valid:
for any member $\hat{X}$ of ${\cal K}_W$,
$\pi'_*(\hat{X})$ corresponds to a unitary motion of the hidden system.
This statement reflects the physical fact that the unitary motion of 
the hidden system do not affect the observed system.

\begin{lemma}
The intersection of the ${\cal LS}_W$ and ${\cal LR}_W$ is
given by
\begin{eqnarray}
&&{\cal LS}_W\cap{\cal LR}_W\nonumber\\
&=&\{\hat{X}\: |\:  [{\bf L}^S (\pi_*(\hat{X})),\pi(W)]=0,\, 
				({\bf L}\hat{X})^*={\bf L}\hat{X},\,\}\nonumber\\
&=&\{\hat{X}\: |\:  [{\bf L}^S (\pi'_*(\hat{X})),\pi'(W)]=0,\, 
				({\bf L}\hat{X})^*={\bf L}\hat{X},\,\}.
\label{eqn:lscaplr}
\end{eqnarray}

The intersection of the ${\cal K}_W$ and ${\cal SK}_W$ is
given by
\begin{eqnarray}
&&{\cal K}_W\cap{\cal SK}_W\nonumber\\
&=&\{\hat{X}\: |\:  [{\bf L}(\hat{X}),\pi(W)]=0,\, 
				({\bf L}\hat{X})^*=-{\bf L}\hat{X},\,\}\nonumber\\
&=&\{\hat{X}\: |\:  [{\bf L}'(\hat{X}),\pi'(W)]=0,\, 
				({\bf L}\hat{X})^*=-{\bf L}\hat{X},\,\},
\label{eqn:kcapsk}
\end{eqnarray}
where ${\bf L}'(\hat{X})$ is defined by
\begin{eqnarray}
{\bf L}'(\hat{X})\equiv W^{-1}{\bf M}(\hat{X})
\nonumber
\end{eqnarray}
\label{lemma:lscaplr}
\end{lemma}

The former statement of the theorem
means that for any vector $\hat{X}\in {\cal LS}_W \cap {\cal LR}_W$,
$\pi_*(\hat{X})$ and $\pi'_*(\hat{X})$ correspond
to the change of eigenvalues of $\pi(W)$ and $\pi'(W)$, respectively.
The latter statement implies that
 for any vector $\hat{X}\in {\cal K}_W\cap {\cal SK}_W$,
$\pi_*(\hat{X})$ and $\pi'_*(\hat{X})$ correspond
to the change of the phase of the eigenvectors of of
$\pi(W)$ and $\pi'(W)$, respectively.

\begin{proof}
$(\ref{eqn:defLR1})$
yields 
\begin{eqnarray}
\pi_*(\hat{X})=({\bf L}\hat{X})\pi(W)=\pi(W)({\bf L}\hat{X})^*,
\nonumber
\end{eqnarray}
which, combined with $(\ref{eqn:hsld})$, yields
\begin{eqnarray}
[({\bf L}\hat{X}),\pi(W)]=0
\nonumber
\end{eqnarray}
Because  ${\bf L}(\hat{X})={\bf L}^S (\pi_*(\hat{X})$ holds true
for any $\hat{X}\in{\cal LS}_W$, we have the first equality
in $(\ref{eqn:lscaplr})$. The second equality in $(\ref{eqn:lscaplr})$ 
and the equalities in 
$(\ref{eqn:kcapsk})$ are obtained in the same manner.
\end{proof}

\section{Canonical distribution}
In this section, as an application of the duality of SLD and RLD,
we try an estimation theoretical characterization 
of the canonical  model.

One conspicuous feature of the canonical distribution model
is that only the eigenvalue of the density matrix 
is dependent on the parameter, and that the eigenvector
is left unchanged even if the parameter changed.
Other thermodynamical models, 
for example, the $T-p$ model 
\begin{eqnarray}
\left\{\rho(T,p)\: \left|\: 
\rho(T,p)=
\sum_{\omega} |\omega\rgl\lgl \omega|
\exp\left[-\frac{1}{k_B T}
\left(E_{\omega}-pV_{\omega} -G(T,p)\,\right)\right] 
\right.\right\}
\nonumber
\end{eqnarray}
and 
the grand canonical model
\begin{eqnarray}
\left\{\rho(T,\mu)\: \left|\: 
\rho(T,\mu)=
\sum_{\omega} |\omega\rgl\lgl \omega|
\exp\left[ -\frac{1}{k_B T}
\left(E_{\omega}-\mu N_{\omega}+Y(T, \mu)\,\right)\right]
\right.\right\}
\nonumber
\end{eqnarray}
also share this feature.
Here,  $H$ is the Hamiltonian and 
$E_{\omega}$  the   $\omega$th eigenvalue of $H$, 
and $|\omega\rgl$ the  $\omega$th  eigenvector of $H$.
We call the model which has this feature the {\it classical model}.

Let us require  first that the canonical distribution is
 a pure state in the composite Hilbert space
${\cal H}\otimes {\cal H}'$, 
where ${\cal H}$ is for the system and ${\cal H}'$  for the heat bath
( taking trace over the heat bath, we have the canonical distribution).
In other words, we assume that  the pure state model
\begin{eqnarray}
{\cal N}\equiv
\left\{|W(T)\rgl\lgl W(T)|\:\: |\: |W(T)\rgl\in {\cal H}\otimes {\cal H}'
\right\},
\label{eqn:totmdl}
\end{eqnarray}
the partial trace $\pi$ over the hidden system ${\cal H}'$
reduces to the canonical model ${\cal M}$.
We denote by ${\cal M}'$ the model induced from
${\cal N}$ by the partial trace $\pi'$ 
over the observed system ${\cal H}$.

Second, we assume the following situation: 
the optimization of the measurement in the estimator of 
the temperature over any of the following three range  
\begin{enumerate}
\item all the measurements in ${\cal H}$
\item all the measurements in ${\cal H}'$
\item all the measurements in ${\cal H}\otimes{\cal H}'$
\end{enumerate}
achieves  exactly the same extent of the efficiency.
In usual situation, we can achieve more efficiency in the case of $3$ than
the other cases, for the range of the optimization is larger.
However, as for the macroscopic parameter like 
the temperature, it is natural to assume that
the measurement of the total system do not bring about more information
than the measurement of the system.
In addition, the measurement of the temperature in the system and 
the heat bath must yield same amount of information, 
because they are in the thermal equilibrium.
In this situation, we  say that  
the model ${\cal M}$ and ${\cal N}$ are {\it maximally entangled}.

\begin{theorem}
The  model ${\cal M}$ and ${\cal M}'$ 
induced by the projection $\pi$ and $\pi'$ from 
the pure state model ${\cal N}$ are classical
iff they are maximally entangled.
\end{theorem}

To prove the theorem, we need the following fact 
in the pure state estimation theory,
which is explained in the later chapters in detail.

\begin{fact}
The attainable lower bound of the variance of 
the unbiased estimator of the pure state model $(\ref{eqn:totmdl})$
is given by $1/\tr A(T)A^*(T)$, where $A(T)$ is 
the matrix representation of 
the tangent vector $d/dT|_{W}\in {\cal T}_W({\cal W})$
to the ${\cal W}$.
\end{fact}

\begin{proofth}
Here, we assume $\dim {\cal H}'$ is equal to $d=\dim {\cal H}$.
For the efficiency of the estimation in the case of $1$
is equal to that in the case of $3$, 
$d/dT|_{W}\in {\cal T}_W({\cal W})$
need to be a member of the horizontal subspace ${\cal LS}_W$,
because 
the length of the horizontal component
of $d/dT|_{W}\in {\cal T}_W({\cal W})$ gives
the SLD Fisher information of the model ${\cal M}$.

Mostly in the same manner, it can be proved that
$d/dT|_{W}\in {\cal T}_W({\cal W})$ is the member of ${\cal LR}_W$
for the efficiency of the estimation in the case of $2$
to be equal to that in the case of $3$.
Therefore,
\begin{eqnarray}
\left. \frac{d}{dT}\right|_{W} \in{\cal LS}_W\cap {\cal LR}_W,
\nonumber
\end{eqnarray}
which, mixed with lemma $\ref{lemma:lscaplr}$
 leads to the statement of the theorem.
\end{proofth}

This theorem, 
which also applies to the grand canonical model 
and the $T-p$ model,
 characterize the entanglement between
the heat bath and the system.


\part*{Part II \\ 
The pure state  model theory}
\addcontentsline{toc}{chapter}%
{Part II :
The pure state  model theory}

\chapter{The pure state estimation theory}
\section{Histrical review of the theory and the purpose of the chapter}
First, we review the history of quantum estimation theory
to clarify the purpose of the chapter.

In parallel with the classical estimation theory,
in 1967, Helstrom show that in the faithful state model,
the covariance matrix is larger than or equal to the inverse of 
SLD Fisher information matrix, and that in the one parameter faithful model,
the bound is attainable \cite{Helstrom:1967}\cite{Helstrom:1976}.

On the other hand, in the multi-parameter model,
it is proved 
that there is no matrix 
which makes attainable lower bound of covariance matrix,
because of non-commutative nature of quantum theory.
Therefore, we deal with the attainable CR type bound defined 
in the section $\ref{sec:nonclassical}$.
Lower bounds of $\Tr GV_{\theta}[M]$
is, attainable or not,  called {\it Cramer-Rao (CR) type bound}.

Five years after Helstrom's work, 
Yuen and Lax found out the exact form of the attainable CR type bound 
of the Gaussian state model,
which is a faithful 2-parameter model obtained by superposition of
coherent states by Gaussian kernel\cite{YuenLax:1973}.
Their work is remarkable not only 
because it was first calculation of the attainable CR type bound of 
a multi-parameter model, 
but also because 
they established a kind of methodology,
which we call {\it indirect approach} hereafter,
 to calculate the attainable CR type bound.
First an auxiliary bound which is not generally attainable
is found out 
and then it is proved to be attained in the specific cases.
In their work, they used so-called RLD bound, 
which was used and generalized by several authors.

Holevo completed their work by solving analytically subtle problems
and generalizing RLD bound
\cite{Holevo:1982}. 
In 1998, Nagaoka calculated the attainable CR type bound of
the faithful 2-parameter spin-$1/2$ model
using Nagaoka bound, which is also another auxiliary bound
\cite{Nagaoka:1989}.

The models had been assumed to be faithful
till
Fujiwara and Nagaoka 
formulated the problem in the pure state model,
and calculated the CR type bound for the 1-parameter model
by use of generalized SLD bound 
and that of the 2-parameter coherent model 
by use of generalized RLD bound in 1995
\cite{FujiwaraNagaoka:1995}\cite{FujiwaraNagaoka:1996}.

The approach in this chapter,  called {\it direct approach} 
in contrast with {\it indirect approach},
is essentially different from the approaches of other authors.
We reduce the given minimization problem to the problem
which is easy enough to be solved directly
by elementary calculus.
The methodology is successfully applied to 
the general 2-parameter pure state model,
coherent model with arbitrary number of parameters,
and the minimization of $\Tr J^S(\theta)V_{\theta}[M]$
for  arbitrary sure state model.
These 
are relatively general category in comparison with the cases
treated by other authors. 
In the 2-parameter pure state model,
the existence of the order parameter $\beta$ which
is a good index of noncommutative nature between the parameters.

As a by-product, we have remarkable corollary,
which asserts that even for non-commutative cases,
a simple measurement attains the lower bound.

\section{Notations}
\label{sec:sld}
In this chapter, we consider data space $\Omega$ to be $\R^m$,
and each $\omega\in\Omega$ to be the estimate $\hat\theta$.

The pure state  model ${\cal M}$ is assumed to 
be induced by ${\cal M}=\pi({\cal N})$ from the manifold 
\begin{eqnarray}
{\cal N}=\{ |\phi(\theta)\rgl\, |\, |\phi(\theta)\rgl\in \tilde{\cal H}\},
\nonumber
\end{eqnarray}
where $\tilde{\cal H}$ is the set of the members of ${\cal H}$
with unit length,
\begin{eqnarray}
\tilde{\cal H}=\{|\phi\rgl\:|\:|\phi\rgl\in{\cal H},\lgl\phi|\phi\rgl=1\}.
\nonumber
\end{eqnarray}
If this assumption is made, the horizontal lift of the tangent vector
is taken for granted.
Uniqueness is proved mostly in the same way as
the proof of theorem $\ref{theorem:uqhl}$.

We denote by $|l_X\rgl$ 
the matrix representation of the horizontal lift
of $X\in {\cal T}_{\rho(\theta)}({\cal M})$,
and $|l_i\rgl$ is short for  $|l_{\partial/\partial \theta^i}\rgl$. 
$|l_X\rgl$ satisfies
\begin{eqnarray}
X\rho(\theta)=
\frac{1}{2}(|l_X\rgl\lgl\phi(\theta)|+|\phi(\theta)\rgl\lgl l_X|),
\label{eqn:sld:lift}
\end{eqnarray}
and
\begin{eqnarray}
\lgl l_X|\phi(\theta)\rgl=0.
\label{eqn:sld:horizontal}
\end{eqnarray}

Notice that 
$span_{\R}\{|l_i\rgl\;|\: i=1,...,m\}$ is a representation of 
${\cal T}_{\rho(\theta)}({\cal M})$ 
because of unique existence of $|l_X\rgl$.
Therefore, we often also call the matrix representation $|l_X\rgl$
horizontal lift.

We call ${\bf M}^i(M,\, |\phi(\theta)\,)$
 {\it estimation vector} of the parameter $\theta^i$ 
by a measurement $M$ at $|\phi(\theta)\rgl$, 
An estimation vector ${\bf M}^i(M,\, |\phi(\theta)\,)$
is said to be {\it locally unbiased}
iff $M$ is locally unbiased.
The local unbiasedness conditions for estimating vectors writes
\begin{eqnarray}
\lgl \phi(\theta)|{\bf M}^i(M,\, |\phi(\theta)\rgl \,)&=&0,\\
{\rm Re}\lgl l_j(\theta)|{\bf M}^i(M,\,|\phi(\theta)\rgl\,)
&=&\delta^i_j
\:(i,j=1,...,m).
\label{eqn:sld:hunbiased}
\end{eqnarray}
Often, we omit the argument  $\theta$ is
$| l_j(\theta)\rgl,|\phi(\theta)\rgl, \rho(\theta)$, and $J^S(\theta)$
and denote them simply by $| l_j\rgl,|\phi\rgl, \rho, J^S$.

We denote the ordered pair of vectors 
\begin{eqnarray}
[\,|l_1\rgl,\,|l_2\rgl,...\,|l_m\rgl\,]
\nonumber
\end{eqnarray}
by ${\sf L}$> In this notation, 
the SLD Fisher information matrix $J^S$ writes
\begin{eqnarray}
J^S={\rm Re}{\sf L}^*{\sf L}\equiv {\rm Re}\,\left[\,\lgl l_i|l_j\rgl\,\right],
\nonumber
\end{eqnarray}
and the imaginary part of ${\sf L}^*{\sf L}$
is denoted by $\tilde{J}$.
Generally,
for the ordered pairs 
\begin{eqnarray}
{\sf X}=[\,|x^1\rgl, \,|x^2\rgl,...,\,|x^m\rgl\,],\;\;
{\sf Y}=[\,|y^1\rgl, \,|y^2\rgl,...,\,|y^m\rgl\,],\nonumber
\nonumber
\end{eqnarray}
of vectors,
we define
\begin{eqnarray}
{\sf X}^*{\sf Y}=\left[\,\lgl x^i|y^j\rgl\,\right]\nonumber
\end{eqnarray}
for notational simplicity.

Then, letting ${\sf X}$ be
\begin{eqnarray}
\left[\,{\bf M}^1(M,\,|\phi\rgl\,), {\bf M}^2(M,\,|\phi\rgl\,),
	...,{\bf M}^m(M,\,|\phi\rgl\,)\,\right],
\nonumber
\end{eqnarray}
the unbiasedness conditions  $(\ref{eqn:sld:hunbiased})$ 
writes
\begin{eqnarray}
{\rm Re}\,{\sf X}^*{\sf L}=I_m,\:
{\rm Re}\,\left\{\,({\bf M}^i(M,\, |\phi\rgl\,)\,)^*|l^j \rgl\,\right\}=I_m,
\label{eqn:lagrange:restriction1}
\end{eqnarray}
where $I_m$ is the $m\times m$ unit matrix,
and 
the matrix $Z_{\theta}[M]$ defined
 in the equation $(\ref{eqn:defz})$ writes
\begin{eqnarray}
Z_{\theta}[M]={\sf X}^*{\sf X}
\nonumber
\end{eqnarray}

\section{The commuting theorem and the locally quasi-classical model}
\label{sec:classical}
In this section,
the necessary and sufficient condition 
for SLD CR bound to be attainable is studied. 
Fujiwara proved the following theorem \cite{Fujiwara}.

\begin{theorem}(Fujiwara \cite{Fujiwara})
SLD CR bound is attainable iff
 SLD's $\{L^S_i|i=1,...,m\}$ can be chosen so that
\begin{eqnarray}
[L^S_i, L^S_j]=0, (i,j=1...,m).
\nonumber
\end{eqnarray}
\end{theorem}

We prove another necessary and sufficient condition which is much easier
to check for the concrete examples, by use of
the following {\it commuting theorem}, which
plays key role in our direct approach to pure state estimation.

\begin{theorem}
If there exists a unbiased measurement $M$ such that
\begin{eqnarray}
V[M]={\rm Re}{\sf X}^*{\sf X},
\label{eqn:V=Z}
\end{eqnarray}
where
the ordered pair ${\sf X}$ is
\begin{eqnarray}
 {\sf X}=\left[{\bf M}^1(M,\,|\phi\rgl\,), {\bf M}^2(M,\,|\phi\rgl\,),
		...,{\bf M}^m(M,\,|\phi\rgl\,) \right],
\label{eqn:xm}
\end{eqnarray}
then,
\begin{eqnarray}
{\rm Im}{\sf X}^*{\sf X}=0
\label{eqn:imxx=0}
\end{eqnarray}
holds true.
conversely, if $(\ref{eqn:imxx=0})$ holds true for some ordered pair ${\bf X}$
of vectors, then 
there exists  a simple, or projection valued, 
unbiased estimator $E$ which satisfies 
$(\ref{eqn:V=Z})$, $(\ref{eqn:xm})$, and
\begin{eqnarray}
&&E(\{\hat\theta_{\kappa}\})^2 = E(\{\hat\theta_{\kappa}\}),\nonumber\\
&&E(\{\hat\theta_{0}\}) = E_0,\nonumber\\
&&E\left( {\R}^m/\bigcup_{\kappa=0}^m \{\hat\theta_{\kappa}\} 
                                                 \right)=0,
\label{eqn:commes}
\end{eqnarray}
for some 
$\{\hat\theta_{\kappa}|\hat\theta_{\kappa}\in {\R}^m,\kappa = 0,...,m+1\}$,
where $E_0$ is a projection onto orthogonal complement subspace
of $span_{\C}\{{\sf X}\}$.
\label{theorem:commute}
\end{theorem}

\begin{proof}
If $(\ref{eqn:V=Z})$ holds,
inequality $(\ref{eqn:genVZ2})$ in lemma $\ref{lemma:genVZ}$
leads to
\begin{eqnarray}
{\rm Re}{\sf X}^* {\sf X}\ge {\sf X}^*{\sf X},
\nonumber
\end{eqnarray}
or
\begin{eqnarray}
0 \ge  i{\rm Im}{\sf X}^*{\sf X},
\nonumber
\end{eqnarray}
which implies  ${\rm Im}{\sf X}^*{\sf X}=0$.

Conversely, Let us assume that  $(\ref{eqn:imxx=0})$ holds true.
Applying Schmidt's orthogonalization to the system
$\{|\phi\rgl, {\sf X}\}$ of vectors, 
we obtain the orthonormal system $\{|b^i\rgl\;|\; i=1,...,m+1\}$
by which the system ${\sf X}$ of vectors write
of vectors such that,
\begin{eqnarray}
{\sf X}=\left[\sum_{j=1}^{m+1} \lambda^1_j|b^j\rgl, 
\sum_{j=1}^{m+1} \lambda^2_j|b^j\rgl,...,
\sum_{j=1}^{m+1} \lambda^m_j|b^j\rgl\right],
\nonumber
\end{eqnarray}
where $\lambda^i_j\, (i=1,...,m,\,j=1,...,m+1)$ are
real numbers.
Letting  $O=[o^i_j]$ be a $(m+1)\times (m+1)$ real orthogonal matrix
such that
\begin{eqnarray}
\lgl\phi|\sum_{j=1}^{m+1} o^i_j |b^j\rgl\neq 0,
\nonumber
\end{eqnarray}
and denoting $\sum_{j=1}^{m+1} o^i_j |b^j\rgl$ by $|b'^{i}\rgl$,
$i$th member of the ordered pair ${\sf X}$ writes
\begin{eqnarray}
&&\sum_{j=1}^{m+1} \lambda^i_j\sum_{k=1}^{m+1} o^k_j |b'^k\rgl\nonumber\\
&=&\sum_{k=1}^{m+1}
	\left(\sum_{j=1}^{m+1} \lambda^i_j o^k_j \right)|b'^k\rgl\nonumber\\
&=&\sum_{k=1}^{m+1}
	\frac{\sum_{j=1}^{m+1} \lambda^i_j o^k_j }{\lgl b'^k|\phi\rgl}
		|b'^k\rgl\lgl b'^k|\phi\rgl.
\nonumber
\end{eqnarray}
Therefore, since the system 
$\{|b'^i\rgl\:|\: i=1,...,m+1\}$ of vectors
is orthonormal,
we obtain an unbiased measurement which satisfies $(\ref{eqn:commes})$
as follows:
\begin{eqnarray}
&&\hat\theta_{\kappa}=
\frac{\sum_{j=1}^{m+1} \lambda^i_j o^{\kappa}_j }{\lgl b'^{\kappa}|\phi\rgl},
\:\:\kappa =1,...,m+1,\nonumber\\
&&\hat\theta_{0}=0,\nonumber\\
&&E(\{\hat\theta_{\kappa}\})=|b'^{\kappa}\rgl\lgl b'^{\kappa}|,\:\:
\kappa=1,...,m+1\nonumber,\\
&&E(\hat\theta_{0})=
I-\sum_{\kappa=1}^{m+1}|b'^{\kappa}\rgl\lgl b'^{\kappa}|.
\nonumber
\end{eqnarray}
Here, $I$ is the identity in ${\cal H}$.
\end{proof}

\begin{theorem}
SLD CR bound is attainable iff
\begin{eqnarray}
\tilde{J}={\rm Im}{\sf L}^*{\sf L}=0
\label{eqn:lsls0}
\end{eqnarray}
$\lgl l_j|l_i\rgl$ is real for any $i,j$.
Conversely, if $(\ref{eqn:lsls0})$ holds true, SLD CR bound is achieved 
by a simple measurement, i.e., a projection valued measurement.
\label{theorem:purelsls0}
\end{theorem}

\begin{proof}
If SLD CR bound is attainable, 
by virtue of lemma $\ref{lemma:genVZ}$-$\ref{lemma:genZJ}$,
we have $(\ref{eqn:imxx=0})$ and 
$(\ref{eqn:genxx=J})$, which lead directly to $(\ref{eqn:lsls0})$.

Conversely, if ${\rm Im}\lgl l_j|l_k\rgl =0$ for any $j,k$,
by virtue of commuting theorem, 
there exists such a simple measurement $E$  that
\begin{eqnarray}
\sum_k [J^{S-1}]^{j,k}|l_k\rgl={\bf M}(E,\,|\phi\rgl).
\nonumber
\end{eqnarray}
Elementary calculations show that the covariance matrix of this measurement
equals  $J^{S-1}$.
\end{proof}

Our theorem is equivalent to Fujiwara's one,
because 
by virtue of commuting theorem,
$\lgl l_j|l_i\rgl$ is real iff there exist such SLD's that
$L_i^S$ and $L_j^S$ commute for any $i,j$.
However, our condition is much easier to be checked,
because 
SLD's are not unique in the pure state model.

\begin{example}
Often, a model is defined by an initial state and  generators,
\begin{eqnarray}
&&\rho(\theta_0)=\rho\nonumber\\
&&\partial_i\rho(\theta)=i[H_i(\theta),\rho(\theta)]\nonumber\\
&&\rho(\theta)=\pi(|\phi(\theta)\rgl)\nonumber
\end{eqnarray}
Then, $\lgl l_j|l_i\rgl$ is real iff
$\lgl\phi(\theta)|[\, H_i(\theta),\, H_j(\theta)\, ]|\phi(\theta)\rgl$ is $0$.
Because of theorem $\ref{theorem:commute}$,\\
$\lgl\phi(\theta)|[\, H_i(\theta),\, H_j(\theta)\, ]|\phi(\theta)\rgl=0$ 
is equivalent to the existence of  
generators which commute with each other, $[H_i(\theta),H_j(\theta)]=0$.
\end{example}

Because of this example  and the theorem by Fujiwara, 
 we may metaphorically say 
that SLD CR bound is attainable iff any two parameter has 
`classical  nature' at $\theta$,
because often classical limit of a quantum system is obtained
by taking such a limit that commutation relations of observables tend to 0.
Throughout the paper, we say that a manifold ${\cal M}$ is
locally quasi-classical at $\theta$  
iff $\lgl l_j|l_i\rgl$ is real at $\theta$.
The following remark describes another `classical' aspect of 
the condition ${\rm Im}\lgl l_j|l_i\rgl=0$.

\begin{example}
The model ${\cal M}=\pi({\cal N})$, where
${\cal N}$ is a real span of some orthonormal basis of ${\cal H}$,
is locally quasi-classical at any point.
\end{example}

As is illustrated in this example, 
when the model ${\cal M}=\pi({\cal N})$ 
is locally quasi-classical at $\theta_0$,
${\cal N}$ 
behaves like an element of real Hilbert space around  $\theta_0$.
Metaphorically speaking, 
$|\phi(\theta)\rgl$'s phase parts  don't change around $\theta_0$
at all, and ${\cal N}$ looks like the family of
the square root of the probability distributions.

\section{The reduction theorem and the direct approach}
\label{sec:naimark}

\begin{theorem}
(Naimark's theorem, see Ref. $\cite{Holevo:1982}$, pp. 64-68.)\ \\
Any generalized measurement $M$ in ${\cal H}$ can be dilated
to a simple measurement $E$ 
in a larger Hilbert space ${\cal K}\subset{\cal H}$,
so that
\begin{eqnarray}
M(B)=PE(B)P
\label{eqn:naimark}
\end{eqnarray}
will hold,
where $P$ is the projection from ${\cal K}$ onto ${\cal H}$.
\end{theorem}

 Naimark's theorem, mixed with commuting theorem,
leads to the following reduction theorem, 
which sits at the heart of our direct approach.

\begin{theorem}
Let  ${\cal M}$ be a $m$-dimensional manifold in ${\cal P}_1$,
and ${\sf B}_{\theta}$ be a system 
$\{|\phi'\rgl,\,|l'_i\rgl\: |\; i=1,...,m\}$
of vectors in $2m+1$-dimensional Hilbert space ${\cal K}_{\theta}$
such that
\begin{eqnarray}
\lgl \phi'| l'_j\rgl&=&\lgl \phi| l_j\rgl=0,\nonumber\\
\lgl l'_i| l'_j\rgl&=&\lgl l_i| l_j\rgl,\nonumber\\
\lgl \phi'| \phi'\rgl&=&\lgl \phi | \phi\rgl=1
\nonumber
\end{eqnarray}
for any $i,j$.
Then,
for any locally unbiased estimator $M$ at $\theta$ in ${\cal H}$,
there is a simple `locally unbiased' measurement $E$ in ${\cal K}_{\theta}$,
\begin{eqnarray}
|x^i\rgl={\bf M}^i(E, |\phi'\rgl)\in{\cal K}_{\theta}
\label{eqn:xe}
\end{eqnarray}
\begin{eqnarray}
\lgl x^i|\phi'\rgl&=&0,\\
{\rm Re}\lgl x^i| l'_j\rgl&=&\delta^i_j
\:(i,j=1,...,m),
\label{eqn:naimark:unbiased}
\end{eqnarray}
whose  `covariance matrix' $V[E]$ equals $V[M]$,
\begin{eqnarray}
V[M]=V[E]\equiv 
\left[\int (\hat\theta^i-\theta^i)(\hat\theta^j-\theta^j)
						\lgl \phi'| E(d\hat\theta )| \phi'\rgl\right], 
\label{eqn:ve}                                  
\end{eqnarray}
\label{theorem:simple}
\end{theorem}

\begin{proof}
 For any locally unbiased measurement $M$,
there exists a Hilbert space ${\cal H}_{M}$ and a simple measurement 
$E_{M}$ in ${\cal H}_{M}$ which satisfies $(\ref{eqn:naimark})$ 
by virtue of Naimark's theorem. Note that $E_{M}$ is also locally unbiased.
Mapping 
$span_{\bf C}\{|\phi\rgl, |l_i\rgl,{\bf M}^i(M, |\phi\rgl)\: |\:i=1,...,m\}$
isometrically onto ${\cal K}_{\theta}$ 
so that $\{|\phi\rgl,|l_i\rgl\: |\:i=1...m\}$ 
are mapped to $\{|\phi'\rgl,|l'_i\rgl\: |\: i=1,...,m\}\:$,
we denote the images of $\{{\bf M}^i(M, |\phi\rgl)\: |\:i=1,...,m\}$ by 
$\{|x^i\rgl\: |\: i=1,...,m\}$.

Then, by virtue of the commuting theorem,
we can construct a simple measurement $E$ 
in ${\cal K}_{\theta}$ satisfying the equations 
$(\ref{eqn:xe})$ - $(\ref{eqn:ve})$.
\end{proof}

In our direct approach,
the reduction theorem reduces
the determination of ${\cal V}$ to 
the determination of the set of matrices
\begin{eqnarray}
V={\rm Re}{\sf X}^*{\sf X}
\nonumber
\end{eqnarray}
where a system
${\bf X}=[|x^1\rgl, |x^2\rgl, ..., |x^m\rgl]$ of 
elements of ${\cal K}_{\theta}\ominus\{|\phi'\rgl\}$
which satisfies 
$(\ref{eqn:imxx=0})$ and $(\ref{eqn:naimark:unbiased})$.
In the same way,
we minimize
\begin{eqnarray}
\Tr G{\rm Re}{\sf X}^*{\sf X},
\nonumber
\end{eqnarray}
where
$\{|x^i\rgl\: |\: i=1,...,m\}\in {\cal K}_{\theta}\ominus\{|\phi'\rgl\}$,
under the restriction such that
the equations $(\ref{eqn:imxx=0})$ and  $(\ref{eqn:naimark:unbiased})$
are satisfied,
instead of minimization of $\tr GV$
where $V$ runs through ${\cal V}$.

Now, the problem is simplified to the large extent,
because 
we only need to treat 
with vectors $\{|x^i\rgl\: |\: i=1,...,m\}$
in finite dimensional Hikbert space  ${\cal K}_{\theta}$
instead of measurements, or operator valued measures.

We  conclude this  section with 
a corollary of reduction theorem,
which is seemingly paradoxical, since
historically, non-projection-valued measurement
is introduced to describe  measurements of
non-commuting observables.

\begin{corollary}
When the dimension of ${\cal H}$ is larger than or equal to $2m+1$,
for any unbiased measurement $M$ in ${\cal H}$,
there is a simple measurement $E$ in ${\cal H}$
which has the same covariance matrix as that of $M$.
\label{corollary:ve=vm1}
\end{corollary}

\begin{proof}
Chose
 $\{ |l'_i\rgl\: |\:i=1,...,m\}$
to be $\{ |l_i\rgl\: |\:i=1,...,m\}$.
\end{proof}

Especially, if  ${\cal H}$ is infinite dimensional,
as is the space of wave functions,
the assumption of the corollary is always satisfied.

\section{Lagrange's method of indeterminate 
coefficients in the pure state estimation theory}
\label{sec:lagrange}

Now, we apply our direct approach to the problems presented in the end of 
section $\ref{sec:nonclassical}$.
To minimize the functional 
$\Tr G{\rm Re}{\sf X}^*{\sf X}$ of vectors in ${\cal K}_{\theta}$,
 Langrange's indeterminate coefficients method is employed.
First,
denoting an ordered pair $\{|l_i'\rgl\: |\: i=1,...,\}$
of vectors in ${\cal K}_{\theta}$ also by ${\sf L}$,
the symbol which is used also for 
an ordered pair $\{|l_i\: |\: i=1,...,\}$
of vectors in ${\cal H}$,
we define a function $Lag({\sf X})$ by 
\begin{eqnarray}
Lag({\sf X})\equiv
{\rm Re}{\Tr}{\sf X}^*{\sf X}G
-2{\Tr}(({\rm Re}{\sf X}^*{\sf L}-I_m)\Xi)
-{\Tr}{\rm Im}{\sf X}^*{\sf X}\Lambda,
\label{eqn:lagrangean}
\end{eqnarray}
where $\Xi, \Lambda$ are matrices whose components are 
Langrange's indeterminate coefficients.
Here, $\Lambda$ can be chosen to be antisymmetric, for
\begin{eqnarray} 
{\Tr}{\rm Im}{\sf X}^*{\sf X}\Lambda
={\Tr}{\rm Im}{\sf X}^*{\sf X}(\Lambda-\Lambda^T)/2
\nonumber
\end{eqnarray}
holds true and only skew symmetric part of $\Lambda$ appears
in $(\ref{eqn:lagrangean})$.

From here, we follow the routine of  
Langrange's method of indeterminate coefficients.
Differentiating $L({\sf X}+\varepsilon\delta{\sf X})$ 
with respect to $\varepsilon$
and substituting  $0$ into $\varepsilon$ in the derivative,
we get
\begin{eqnarray}
{\rm Re}{\Tr}(\delta{\sf X}^*
(2{\sf X}G-2{\sf L}\Xi-2i{\sf X}\Lambda))=0.
\nonumber
\end{eqnarray}
Because $\delta{\sf X}$ is arbitrary,
\begin{eqnarray}
{\sf X}(G-i\Lambda)={\sf L}\Xi
\label{eqn:xg-il}
\end{eqnarray}
is induced. 

Multipling ${\sf X}^*$ to both sides of 
$(\ref{eqn:xg-il})$,
the real part of the outcomming equation, 
together with $(\ref{eqn:lagrange:restriction1})$, yields
\begin{eqnarray}
\Xi={\rm Re}{\sf X}^*{\sf X}G=VG.
\label{eqn:lagrange:xi}
\end{eqnarray}
Substituting $(\ref{eqn:lagrange:xi})$ into $(\ref{eqn:xg-il})$ ,
we obtain
\begin{eqnarray}
{\sf X}(G-i\Lambda)={\sf L}VG.
\label{eqn:basic0.1}
\end{eqnarray}

If $(\ref{eqn:lagrange:restriction1})$,
$(\ref{eqn:imxx=0})$,
$(\ref{eqn:basic0.1})$ and $V={\rm Re}{\sf X}^*{\sf X}$
are solved for $V$, ${\sf X}$ and 
real skew symmetric matrix $\Lambda$,
our problems will be perfectly solved.
However,
so far, solutions only for special cases are known.

The rest of this section is devoted to the proof of 
the theorem which claim a little stronger assertion than
the corollary $\ref{corollary:ve=vm1}$.

Though the  real linear space $span_{\bf R}\{ {\sf L}\}$ is 
always $m$-dimensional for the parameters not to be redundant,
the dimension of the complex linear space 
$span_{\bf C}\{ {\sf L}\}$, or  the rank of ${\sf L}$,
is not necessarily equal to $m$.
If ${\rm rank}_{\C}{\sf L}={\rm rank}G=m$ is assumed,
since the rank of the matrix $V$ is $m$ as is proved soon,
the rank of the left hand side of $(\ref{eqn:basic0.1})$
is $m$, so is the rank of the right hand side,
implying 
that $G-i\Lambda$ is invertible.
The rank of the matrix $V$ is $m$
because $(\ref{eqn:lagrange:restriction1})$ implies that
the dimension of 
$span_{\bf R}\{ {\sf X}\}$ is $m$.

Since
${\sf X}$ is given by ${\sf L}VG(G-i\Lambda)^{-1}$,
we can conclude that $span_{\bf C}\{ {\sf X}\}$ 
should be a subspace of
$span_{\bf C}\{ {\sf L}\}$.
Therefore, by the same argument as in the proof of 
the corollary $\ref{corollary:ve=vm1}$,
we obtain the following theorem.

\begin{theorem}
Suppose that the dimension ${\cal H}$ is larger than or equal to $m+1$,
and that the dimension of the complex linear space
$span_{\bf C}\{ {\sf L}\}$ is $m$.
Then, for any strictly positive weight matrix $G$,
the attainable CR type bound $\CR(G)$
is attained by a simple measurement.
\end{theorem}

\section{The model with two parameters}
\label{sec:2parameter}

In this section, we determine ${\cal V}$ for the arbitrary 2-parameter
pure state model.

The equation$(\ref{eqn:basic0.1})$, 
mixed  with $(\ref{eqn:imxx=0})$,  leads to
\begin{eqnarray}
(G-i\Lambda)V(G-i\Lambda)
=GV{\sf L}^*{\sf L}VG.
\label{eqn:basic1.0}
\end{eqnarray}
whose 
real part and imaginary part  are
\begin{eqnarray}
GVG-\Lambda V \Lambda =GV J^S VG,
\label{eqn:basic1.1}
\end{eqnarray}
and
\begin{eqnarray}
GV\Lambda +\Lambda VG=-GV\tilde{J}VG,
\label{eqn:basic1.2}
\end{eqnarray}
where $\tilde J$ denotes ${\rm Im}{\sf L}^*{\sf L}$,
respectively.

We assert that
when the matrix $G$ is strictly positive,
$(\ref{eqn:basic1.0})$ is equivalent to the existence of ${\sf X}$
which satisfies
$(\ref{eqn:lagrange:restriction1})$,
$(\ref{eqn:imxx=0})$,
$(\ref{eqn:basic1.0})$,
and $V={\rm Re}{\sf X}^*{\sf X}$.
If real positive symmetric matrix $V$ and 
real antisymmetric matrix $\Lambda$ satisfying $(\ref{eqn:basic1.0})$ exist,
${\sf X}$ which satisfies  $(\ref{eqn:basic0.1})$ and 
$(\ref{eqn:imxx=0})$ is given by 
${\sf X}=UV^{1/2}$, where $U$ is a $2m+1$ by $m$  complex matrix such that
$U^*U=I_m$.
If $G$ is strictly positive, ${\sf X}=UV^{1/2}$ also satisfies
$(\ref{eqn:lagrange:restriction1})$,
because 
\begin{eqnarray}
VG={\rm Re}{\sf X}^*{\sf L}VG\nonumber
\nonumber
\end{eqnarray}
is obtained by
multiplication of ${\sf X}^*$ to 
and taking real part of the both sides of $(\ref{eqn:basic0.1})$,
and our assertion is proved.

Hence, our task is to solve $(\ref{eqn:basic1.1})$ and  $(\ref{eqn:basic1.2})$
for real positive symmetric matrix $V$ and real antisymmetric matrix $\Lambda$,
if $G$ is strictly positive.
When $G$ is not strictly positive, 
after solving $(\ref{eqn:basic1.1})$ and  $(\ref{eqn:basic1.2})$,
we must check 
whether there exists an ordered pair ${\sf X}$ of vectors 
which satisfies
$(\ref{eqn:lagrange:restriction1})$, $(\ref{eqn:imxx=0})$
and 
$V={\rm Re}{\sf X}^*{\sf X}$.

In the remainder of  this section, we use the coordinate system 
where  $J^S$ is equal to the identity $I_m$.
Given an arbitrary coordinate system $\{\theta^i|\: i=1,...,m\}$,
such a coordinate system $\{\theta'^i|\: i=1,...,m\}$
is obtained by the following coordinate transform:
\begin{eqnarray}
\theta'^i=\sum_{j=1}^m [(J^{S})^{1/2}]_{ij}\:\theta^j \:(i=1,...,m)
\label{eqn:2para:trans1}
\end{eqnarray}
By this coordinate transform,
$V$ is transformed as: 
\begin{eqnarray}
V'&=&(J^{S})^{1/2} V (J^{S})^{1/2}.
\label{eqn:2para:trans2}
\end{eqnarray}
Therefore, the result in the originally given coordinate is 
obtained as a transformation of the result 
in the coordinate system $\{\theta'^i|\: i=1,...,m\}$,
by using
$(\ref{eqn:2para:trans2})$ in the converse way.

So far, we have not assumed $\dim{\cal M}=2$.
When $\dim{\cal M}=2$, covariance matrices are members of
 the space $Sym(2)$ of  $2\times 2$ symmetric matrices
which is parameterized by the real variables $x,y,$ and $z$, where
\begin{eqnarray}
Sym(2)=
\left\{V\left|V=
\left[
\begin{array}{cc}
z+x & y\\
y   & z-x
\end{array}
\right]
\right.
\right\}.
\nonumber
\end{eqnarray}
Before tackling the equations $(\ref{eqn:basic1.1})$ 
and  $(\ref{eqn:basic1.2})$, three useful facts about this parameterization
are noted.
First, letting $A$ is a symmetric real matrix
which is represented by $(A_x, A_y, A_z)$ in the $(x,y,z)$-space,
the set ${\cal PM}(A)$ of all matrices larger than $A$
is
\begin{eqnarray}
{\cal PM}(A)
=\{(x,y,z)\:|\: (z-A_z)^2-(x-A_x)^2-(y-A_y)^2\ge 0,\:(z-A_z)\ge 0 \},
\nonumber
\end{eqnarray}
that is, interior
 of a upside-down corn with its vertex at $A=(A_x, A_y, A_z)$.
Hence,
$\cal V$ is a subset of ${\cal PM}(I_m)$, or inside of 
an upside-down corn with its vertex 
at $(0,0,1)$ because of SLD CR bound. 
When the model $\cal M$ is locally quasi-classical at $\theta$,
$\cal V$ coincides with  ${\cal PM}(I_m)$.

Second, an  action of 
rotation matrix $R_{\theta}$ to $V$ such that 
$R_{\theta}V R_{\theta}^T$,
where 
\begin{eqnarray}
R_{\theta}=\left[
\begin{array}{cc}
\cos\theta & -\sin\theta\\
\sin\theta & \cos\theta
\end{array}
\right],
\nonumber
\end{eqnarray}
corresponds to the rotation in the $(x,y,z)$-space around 
$z$-axis by the angle $2\theta$.

Third, we have the following lemma.
\begin{lemma}
${\cal V}$ is rotationally symmetric around $z$-axis,
if the coordinate in $\cal M$ is chosen so that
$J^S$ writes the unit matrix $I_m$.
\label{lemma:Vrot}
\end{lemma}
\begin{proof}
$\cal V$ is the set of every matrix  which writes 
${\rm Re}{\sf X}^*{\sf X}$ 
using a $2m+1$ by $m$ complex matrix ${\sf X}$ satisfying 
$(\ref{eqn:lagrange:restriction1})$ and $(\ref{eqn:imxx=0})$.
Therefore,  the  rotational symmetry of $\cal V$
around $z$-axis is equivalent to 
the existence of  a $2m+1\times m$ complex matrix ${\sf Y}$ satisfying 
$(\ref{eqn:lagrange:restriction1})$, $(\ref{eqn:imxx=0})$
and
\begin{eqnarray}
\forall\theta\;\; R_{\theta}{\sf X}^*{\sf X}R_{\theta}^T={\sf Y}^*{\sf Y},
\label{eqn:rxxr}
\end{eqnarray}
for any $2m+1\times m$ complex matrix ${\sf X}$ 
which satisfies
$(\ref{eqn:lagrange:restriction1})$ and $(\ref{eqn:imxx=0})$.
Because of ${\sf L}^*{\sf L}=I_m+i\tilde{J}$,
elementary calculation shows that
\begin{eqnarray}
{\sf L}^*{\sf L}={\rm Re}R_{\theta}{\sf L}^*{\sf L}R_{\theta}^T,
\nonumber
\end{eqnarray}
or, that for some unitary matrix $U$ in ${\cal K}_{\theta}$,
\[
{\sf L}^*U=R_{\theta}{\sf L}^*,
\]
which leads, together with $(\ref{eqn:lagrange:restriction1})$, to
\begin{eqnarray}
{\rm Re}{\sf L}^*U{\sf X}=R_{\theta}.
\end{eqnarray}
Therefore,
\[
{\sf Y}=U{\sf X}R_{\theta}^T
\]
satisfies $(\ref{eqn:rxxr})$, and we have the lemma.
\end{proof}

Because of lemma $\ref{lemma:Vrot}$,
${\cal V}$ is determined
if the boundary of the intersection $\tilde{\cal V}$ of $\cal V$ 
and the $zx$-plane
is calculated.
Note that
the `inner product' $\Tr GV$
of $G$ and $V$
does not take its minimum at $V\in \tilde{\cal V}$
unless $G$ is in the $zx$-plane.
Therefore, to obtain $bd\tilde{\cal V}$,
only diagonal weight matrix is needed to be considered.

Let us begin with the case of a positive definite weight matrix.
In this case, we only need to deal with 
$(\ref{eqn:basic1.1})$ and $(\ref{eqn:basic1.2})$.
Let
\begin{eqnarray}
\tilde{J}=\left[
\begin{array}{cc}
0 & -\beta\\
\beta & 0
\end{array}
\right],
\nonumber
\end{eqnarray}
and 
\begin{eqnarray}
\Lambda=\left[
\begin{array}{cc}
0 & -\lambda\\
\lambda & 0
\end{array}
\right],\:
G=\left[
\begin{array}{cc}
1 & 0 \\
0 & g
\end{array}
\right],\:
V=\left[
\begin{array}{cc}
u & 0 \\
0 & v
\end{array}
\right],
\nonumber
\end{eqnarray}
where $g,v,$ and $u$ are positive real real numbers.
Note that
\begin{eqnarray} 
|\beta|\leq 1
\nonumber
\end{eqnarray}
holds, because ${\sf L}^*{\sf L}=I_m+\tilde{J}$
is nonnegative definite.
Then,  $(\ref{eqn:basic1.1})$ and $(\ref{eqn:basic1.2})$ writes
\begin{eqnarray}
u+v\lambda^2-u^2=0,\nonumber\\
vg^2+u\lambda^2-v^2g^2=0,\nonumber\\
vg\lambda+u\lambda+uv\beta g=0.
\label{eqn:2para:uvglmd}
\end{eqnarray}
The necessary and sufficient condition for 
$\lambda$ and positive $g$ to exist is,
after some calculations,
\begin{eqnarray}
(u-1)^{1/2}+(v-1)^{1/2}-|\beta|(uv)^{1/2}=0.
\label{eqn:2para:uv}
\end{eqnarray}
Note that  $u$ and $v$ are larger than  or equal to $1$, 
because $V\ge J^{S-1}=I_m$.
Substitution of  $u=z+x$ and $v=z-x$ into $(\ref{eqn:2para:uv})$,
after some calculations, leads to
\begin{eqnarray}
|\beta|(z+x-1)^{1/2}(z-x-1)^{1/2}
\pm(1-\beta^2)^{1/2}((z+x-1)^{1/2}+(z-x-1)^{1/2})=|\beta|\nonumber\\
\label{eqn:zx+-1}
\end{eqnarray} 
Fig.1 shows that
the lower sign in the equation $(\ref{eqn:zx+-1})$,
\begin{eqnarray}
|\beta|(z+x-1)^{1/2}(z-x-1)^{1/2}
+(1-\beta^2)^{1/2}((z+x-1)^{1/2}+(z-x-1)^{1/2})=|\beta|\nonumber\\
\label{eqn:2para:zx1}
\end{eqnarray}
gives a part of $bd\tilde{\cal V}$.
In $(\ref{eqn:2para:zx1})$, $x$ takes value ranging 
from $-\beta^2/(1-\beta^2)$ to  $\beta^2/(1-\beta^2)$ if $|\beta|$ is smaller 
than $1$.
When $|\beta|=1$, $x$ varies from $-\infty$ to $\infty$.
This restriction on the range of $x$ comes from
the positivity of $z-x-1$ and $z+x-1$.

\begin{figure}
\centering
\includegraphics[keepaspectratio,scale=0.8]{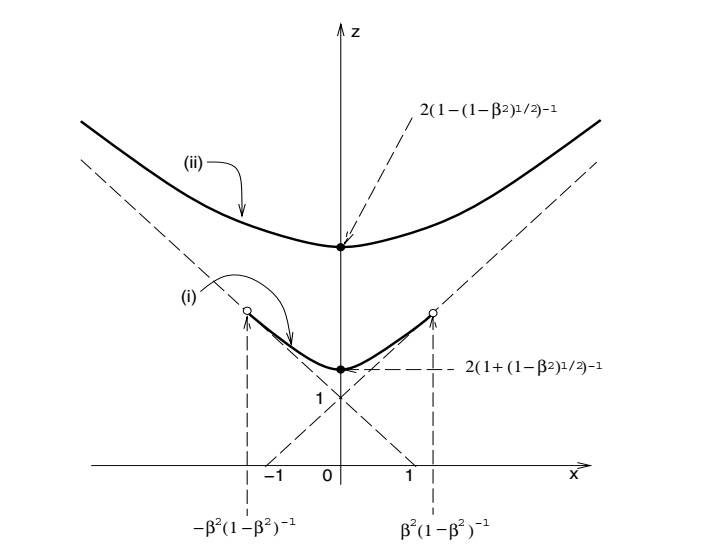}
\caption[Two stationary lines]{
Two stationary lines;\\
$(i)\: |\beta|(z+x-1)^{1/2}(z-x-1)^{1/2}
+(1-\beta^2)^{1/2}((z+x-1)^{1/2}+(z-x-1)^{1/2})=|\beta|;$\\
$(ii)\: |\beta|(z+x-1)^{1/2}(z-x-1)^{1/2}
-(1-\beta^2)^{1/2}((z+x-1)^{1/2}+(z-x-1)^{1/2})=|\beta|;$
}
\label{fig:2p:fig1}
\end{figure}

When the weight matrix $G$ is 
\begin{eqnarray}
G=\left[
\begin{array}{cc}
1 & 0\\
0 & 0
\end{array}
\right],\:\:\mbox{or}\:\:
G=\left[
\begin{array}{cc}
0 & 0\\
0 & 1
\end{array}
\right],
\label{eqn:2para:G}
\end{eqnarray}
we must treat the case of  $|\beta|=1$ and 
the case of $|\beta|<1$ separately.
If $|\beta|=1$, 
there exists no  $(2m+1)\times m $ complex matrix ${\sf X}$
which satisfies  $V={\rm Re}{\sf X}^*{\sf X}$, 
$(\ref{eqn:basic1.0})$,
$(\ref{eqn:lagrange:restriction1})$, and $(\ref{eqn:imxx=0})$.
On the other hand, if $|\beta|<1$,
such complex matrix ${\sf X}$ always exists and 
$V={\rm Re}{\sf X}^*{\sf X}$ is given by, in terms of $(x,y,z)$,
\begin{eqnarray}
z&=&-x+1,\: x\leq -\frac{\beta^2 }{1-\beta^2}\nonumber\\
&&\mbox{or}\nonumber\\
z&=&x+1,\: x\ge \frac{\beta^2 }{1-\beta^2}
\label{eqn:2para:zx2}
\end{eqnarray}

With the help of  $(\ref{eqn:2para:zx1})$ and $(\ref{eqn:2para:zx2})$, 
${\cal V}$ is depicted as Fig.2. 
The intersection of $z$-axis and $bd{\cal V}$ gives
\begin{eqnarray}
\CR(J^S)=\min\{J^S V|V\in {\cal V}\}=\frac{4}{1+(1-|\beta|^2)^{1/2}},
\label{eqn:minvv}
\end{eqnarray}
where the equality holds in any  coordinate of the model ${\cal M}$.
Simple calculation leads to following theorem.

\begin{theorem}
If a  model ${\cal M}$ has lager value of $|\beta|$ at $\theta$ than
another model ${\cal N}$ has at $\theta'$, 
the ${\cal V}_{\theta}({\cal M})$
is a subset of ${\cal V}_{\theta'}({\cal N})$.
\label{theorem:2pcalv}
\end{theorem}

By virtue of this theorem, 
$|\beta|$ can be seen as a measure of `uncertainty'
between the two parameters.
Two extreme cases are worthy of special attention;
When $|\beta|=0$, the model $\cal M$ is locally quasi-classical at $\theta$
and ${\cal V}$ is largest.
On the other hand,
if $|\beta|=1$ , ${\cal V}$ is smallest  and 
uncertainty between  $\theta^1$ and $\theta^2$ is maximum.
In the latter case, we say that the model is {\it coherent} at $\theta$. 

\begin{figure}[htp]
\centering
\includegraphics[keepaspectratio,scale=0.4]{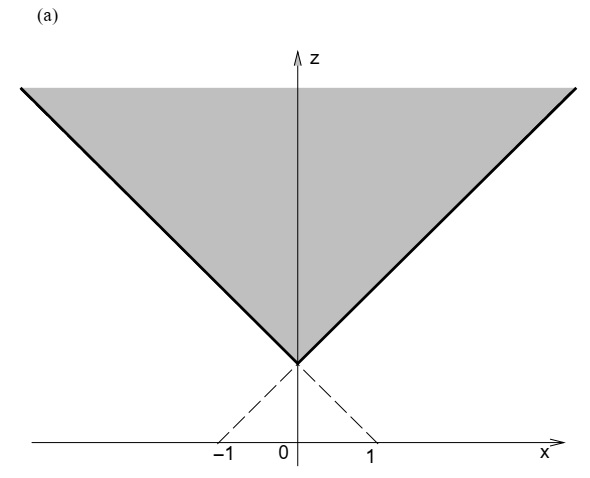}
\includegraphics[keepaspectratio,scale=0.4]{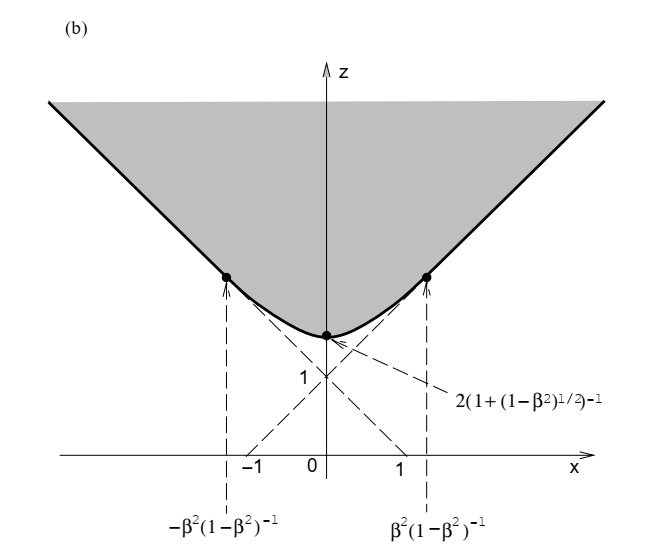}
\includegraphics[keepaspectratio,scale=0.4]{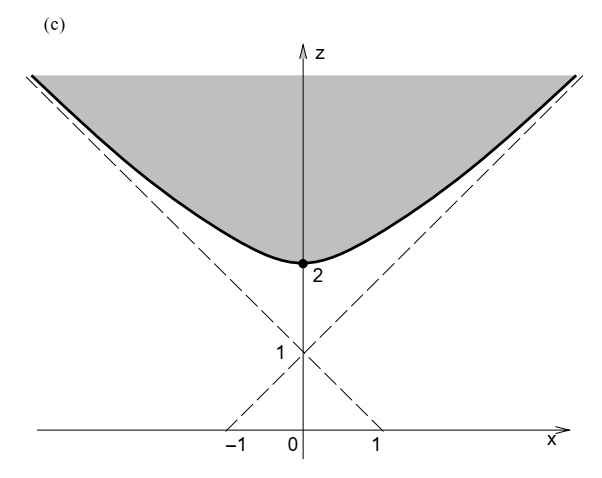}
\caption{(a)$\;|\beta|=0;\:$(b)$\;0<|\beta|<1;\;$(c)$\;|\beta|=1$.}
\label{c1}
\end{figure}

\pagebreak

\begin{example}
We define {\it generalized spin coherent model} \cite{Abe} by
\begin{eqnarray}
&{\cal M}_{s,m}&=\pi({\cal N})\nonumber\\
&{\cal N}_{s,m}&\nonumber\\
&=&\{ |\phi(\theta)\rgl\, |\, 
	|\phi(\theta)\rgl=\exp i\theta^1(\sin\theta^2 S_x-\cos\theta^2 S_y)
															|s, m\rgl,\,
	0\leq\theta^1<\pi,\,0\leq\theta^2<2\pi\},\nonumber\\
\label{eqn:gscoherent}
\end{eqnarray}
where $S_x$, $S_y$, $S_z$ are spin operators,
and $|s,m\rgl$ is defined by,
\begin{eqnarray}
&&S_z|j,m\rgl=\hbar m |s, m\rgl\nonumber\\
&&(S^2_x+S^2_y+S^2_z)|s,m\rgl=\hbar^2 s(s+1)|s, m\rgl.
\nonumber
\end{eqnarray}
$s$ takes value of half integers,
and m is a half integer such that $-j\leq m\leq j$.
Then after tedious calculations, we obtain 
\begin{eqnarray}
{\bf M}\left(\hlift_{|s, m\rgl}\,\partial_1\right)&=&
2i(\sin\theta^2 S_x-\cos\theta^2 S_y)|s, m\rgl\nonumber\\
{\bf M}\left(\hlift_{|s, m\rgl}\, \partial_2\right)&=&
2i\{-\sin\theta^1(\cos\theta^2 S_x + \sin\theta^2 S_y )
	+(\cos\theta-1)(S_z-m\hbar)\}|s, m\rgl,\nonumber\\
\nonumber
\end{eqnarray}
and
\begin{eqnarray}
J^S&=&2\hbar^2(s^2+s-m^2)\left[
\begin{array}{cc}
1 & 0\\
0 & \sin^2\theta^1
\end{array}
\right],\nonumber\\
\tilde{J}&=&\left[
\begin{array}{cc}
0 & 2m\hbar^2\sin\theta^1 \\
-2m\hbar^2\sin\theta^1 & 0
\end{array}
\right],\nonumber\\
\beta_{s,m} &=& \frac{m}{s^2+ s -m^2}.
\nonumber
\end{eqnarray}
If $m=\alpha s$, where $\alpha<1$ is a constant,
$\beta_{s,m}$ tends to zero as $s\rightarrow\infty$,
and the model ${\cal M}_{s,m}$ becomes locally quasi-classical.
However, if  $m=s$, the model ${\cal M}_{s,m}$ is coherent for any $s$.
\end{example}

\section{Multiplication of the imaginary unit}
\label{sec:complex}
As is shown in the previous section,
in the 2-parameter model,
$|\beta|$,
a good index of `uncertainty' between two parameters,
or a measure of how distinct the model is from the classical model.
It can be easily shown that, whatever coordinate
of the model ${\cal M}$ is chosen, 
$\pm i\beta$ are the eigenvalues of   
the matrix $J^{S-1}\tilde{J}$, 
which 
is deeply related to the complex structure of the model.
Actually,
that  matrix 
stands for the linear map ${\bf D}$ from ${\cal T}_{\rho}({\cal M})$
onto ${\cal T}_{\rho}({\cal M})$ defined as in the followings;
First, we  multiply the imaginary unit $i$ to 
$|l_X\rgl={\bf M}(\hlift(|l_X\rgl))$ and
and ${\bf M}^{-1}$ and $\pi_*$ are applied successively to
$i|l_X\rgl$.
since
$\pi_*({\bf M}^{-1}(i|l_X\rgl))$ 
is not a member of ${\cal T}_{\rho}({\cal M})$ generally,
we project $\pi_*({\bf M}^{-1}(i|l_X\rgl))\in {\cal T}_{\rho}({\cal P}_1)$
to  ${\cal T}_{\rho}({\cal M})$
with respect to the metric $\lgl *|* \rgl_{\rho}$,
and we obtain  ${\bf D}X\in {\cal T}_{\rho}({\cal M})$.



\begin{eqnarray}
\begin{array}{ccc}
        &\mbox{multiplication of}\:\: i& 	
\\
|l_X\rgl  \in span_{\bf R}{\sf L}&--\longrightarrow
&i|l_X\rgl  \in span_{\bf R}\{{\sf L}, i{\sf L}\}	
\\
\uparrow& &\:\downarrow {\bf M}^{-1},\, \pi_*
\\
\hlift_{|\phi\rgl},\, {\bf M}& &\pi_*({\bf M}^{-1}(i|l_X\rgl))\in {\cal T}_{\rho}({\cal P}_1)
\\
\uparrow& &\mbox{project}\,\downarrow \mbox{w.r.t.\,  $\lgl*,*\rgl_{\rho}$}
\\
X\in{\cal T}_{\rho}({\cal M})&--\longrightarrow
& {\bf D}X\in{\cal T}_{\rho}({\cal M})
\\
  &{\bf D}&
\end{array}
\nonumber
\end{eqnarray}

The following theorems are  straightforward 
 consequences of the above discussion.

\begin{theorem}
The  absolute value
of an  eigenvalues of ${\bf D}$, or equivalently, of $J^{S-1}\tilde{J}$, 
is smaller than or equal to $1$.
\label{theorem:coherent:beta_i}
\end{theorem}

\section{The coherent model}
\label{sec:coherent}
As for the model with arbitrary number of parameters,
the model is said to be {\it coherent} at $\theta$
iff all of the eigenvalues of $(J^S)^{-1}\tilde{J}$ are $\pm i$.
When the number of parameters is $2$,
this definition of coherency reduces to $|\beta|=1$.
It should be noted that the eigenvalues of $J^{S -1}\tilde{J}$ are ,
whether the model is coherent or not, of the form $\pm i\beta_j$ or 0.
Therefore, the number of parameters of the coherent model is even.

In this section, we determine the attainable CR type bound of 
the coherent model with arbitrary numbers of parameters.
The coherent model is worthy of attention firstly 
because the coherent model is `the maximal uncertainty' model,
secondly because there are many physically important coherent models.


Because $J^{S-1}\tilde{J}$ is 
a representation of ${\bf D}$, 
or of multiplication of the imaginary unit $i$,
the following theorem.
\begin{theorem}
The model $\cal M$ is coherent at $\theta$ iff
\begin{eqnarray}
(J^{S-1}\tilde{J})^2=-I_m
\label{eqn:dd=-1}
\end{eqnarray}
holds true.
\end{theorem}

\begin{theorem}
The model $\cal M$ is coherent at $\theta$ iff 
$span_{\R}\{i{\sf L}\}$ is identical to $span_{\R}{\sf L}$,
or equivalently, iff 
$span_{\R}\{ {\sf L}, i{\sf L}\}$ is identical to 
$span_{\R}{\sf L}$.
\label{theorem:coherent:spaniL}
\end{theorem}

This theorem  leads to the following theorem.

\begin{theorem}
The model $\cal M$ is coherent at $\theta$ iff
the dimension of $span_{\C}{\sf L}$ is $m/2$.
\end{theorem}
\begin{proof}
First, we  assume that 
\begin{eqnarray}
\dim_{\C}span_{\C}{\sf L}=m/2.
\label{eqn:dimL=m/2}
\end{eqnarray}
Because $span_{\R}{\sf L}$ is a $m$-dimensional subspace of 
$span_{\R}\{{\sf L},i{\sf L}\}$ whose dimension is
smaller than or equal to $m$ because of $(\ref{eqn:dimL=m/2})$,
we have 
$span_{\R}\{{\sf L},i{\sf L}\}=span_{\R}{\sf L}$, 
or coherency of the model at $\theta$.

Conversely, let us assume that the model is coherent at $\theta$.
If we take an orthonormal basis $\{e_j|\: j=1,...m\}$ of 
${\cal T}_{\rho}({\cal M})$ such that
$e_{j+m/2}={\bf D}e_i\:(i=1,2,...,m/2)$,
then 
${\bf M}(\hlift(e_{j+m/2}))=i{\bf M}(\hlift(e_{j}))\:(i=1,2,...,m/2)$
holds true,
and any element $|u\rgl$ of  
$span_{\R}{\sf L}= span_{\R}\{{\sf L}, i{\sf L}\}$ 
writes
\begin{eqnarray}
|u\rgl &=&\sum_{j=1}^m a_j {\bf M}(\hlift(e_{j})) \nonumber\\
&=&\sum_{j=1}^{m/2} (a_j+ i a_{j+(m/2)}){\bf M}(\hlift(e_{j})),
\nonumber
\end{eqnarray}
implying that the dimension of $span_{\C}{\sf L}$ is $m/2$.
\end{proof}

In 1996, Fujiwara and Nagaoka \cite{FujiwaraNagaoka:1996}
determined the attainable CR type bound of 
the two parameter coherent model.
In the following, more generally, 
we calculate the bound of the coherent model with arbitrary
number of parameters.

\begin{lemma}
In the case of the coherent model, ${\rm Re}{\sf L}^*{\sf X}=I_m$ 
or its equivalence ${\rm Re}{\sf L}^*({\sf X}-{\sf L}J^{S-1})=0$,
implies
\begin{eqnarray}
{\sf L}^*{\sf X}=I_m+i{\tilde J}J^{S-1}.
\nonumber
\end{eqnarray}
\end{lemma}
\begin{proof}
\begin{eqnarray}
{\rm Im}{\sf L}^*({\sf X}-{\sf L}J^{S-1})&=&
-{\rm Re}i{\sf L}^*({\sf X}-{\sf L}J^{S-1})\nonumber\\
&=&0.
\label{eqn:coherent:imlx0=0}
\end{eqnarray}
Here, $span_{\R}{\sf L}=span_{\R}i{\sf L}$ is 
used to deduce the last equality.
$(\ref{eqn:coherent:imlx0=0})$ and 
${\rm Re}{\sf L}^*({\sf X}-{\sf L}J^{S-1})=0$
implies
\begin{eqnarray}
{\sf L}^*{\sf X}={\sf L}^*{\sf L}J^{S-1}=I_m+i{\tilde J}J^{S-1}.
\nonumber
\end{eqnarray}
\end{proof}

Multiplication of  ${\sf L}^*$ to the both sides of $(\ref{eqn:basic0.1})$,
together with the lemma  presented above, yields
\begin{eqnarray}
(I_m+i{\tilde J}J^{S-1})(G-i\Lambda)=(J^S+i\tilde{J})VG.
\label{eqn:coherent:g-il}
\end{eqnarray}
By virtue of $(\ref{eqn:dd=-1})$,
both of the real part and the imaginary part of 
$(\ref{eqn:coherent:g-il})$ give the same equation,
\begin{eqnarray}
G+{\tilde J}J^{S-1}\Lambda=J^SVG,
\nonumber
\end{eqnarray}
or
\begin{eqnarray}
G^{1/2}VG^{1/2}-
G^{1/2}J^{S-1/2}G^{1/2}=
\left(\, G^{1/2}J^{S-1}{\tilde J}J^{S-1}G^{1/2}\,\right)
\left(\, G^{-1/2}\Lambda G^{-1/2}\,\right).
\nonumber\\
\label{eqn:vj=jl}
\end{eqnarray}
Therefore, letting $a_i$ and $b_i$ denote  the eigenvalues of 
$G^{1/2}J^{S-1}{\tilde J}J^{S-1}G^{1/2}$ and
$G^{-1/2}\Lambda G^{-1/2}$ respectively, we have
\begin{eqnarray}
&&\left[\, (G^{1/2}J^{S-1}{\tilde J}J^{S-1}G^{1/2}),\,
(G^{-1/2}\Lambda G^{-1/2})\,\right]=0\nonumber\\
&&\Tr\,\left\{\,(G^{1/2}J^{S-1}{\tilde J}J^{S-1}G^{1/2})
(G^{-1/2}\Lambda G^{-1/2})\,\right\}
= \sum_i |a_i||b_i|,
\nonumber
\end{eqnarray}
because $(\ref{eqn:vj=jl})$ and  SLD CR inequality implies 
that $(G^{1/2}J^{S-1}{\tilde J}J^{S-1}G^{1/2})(G^{-1/2}\Lambda G^{-1/2})$
is positive Hermitian.

On the other hand, $(\ref{eqn:basic0.1})$ or its equivalence,
\begin{eqnarray}
{\sf X}G^{1/2}(I_m-iG^{-1/2}\Lambda G^{-1/2})
={\sf L}VG^{1/2},
\label{eqn:coherent:i-iglg}
\end{eqnarray}
implies 
$|b_i|=1\,(i=1,...,m)$, because
the rank of $I_m-iG^{-1/2}\Lambda G^{-1/2}$ is
shown to be $m/2$
from $(\ref{eqn:coherent:i-iglg})$, and
the rank of matrices ${\sf L}$, ${\sf X}$, and $V$.
\begin{eqnarray}
\rank_{\C}{\sf X}=\rank_{\R}{\sf X}=m,
\nonumber
\end{eqnarray}
hold true,
where the last equation is valid by virtue of ${\rm Im}{\sf X}^*{\sf X}=0$.

After all, letting $\Tr\abs A$ denote
 the sum of the absolute values of the eigenvalues of $A$, 
we have the following theorem.

\begin{theorem}
\begin{eqnarray}
\CR (G)
=\Tr\, GJ^{S-1}+\Tr\abs\, GJ^{S-1}{\tilde J}J^{S-1},
\nonumber
\end{eqnarray}
where letting $|A|=(AA^*)^{1/2}$,
the covariance matrix $V$ such that
\begin{eqnarray}
V=J^{S-1}+G^{-1/2}\left|G^{1/2}J^{S-1}{\tilde J}J^{S-1}G^{1/2}\right|G^{-1/2}.
\nonumber
\end{eqnarray}
attain the minimum.
\end{theorem}

To check the coherency of the model, the following theorem,
which is induced from 
theorem $\ref{theorem:coherent:beta_i}$, is
useful.
\begin{theorem}
the model is coherent at $\theta$
iff 
\begin{eqnarray}
|{\rm det} J^S|=|{\rm det} \tilde{J}|.
\nonumber
\end{eqnarray}
\label{theorem:coherent:check}
\end{theorem}

\begin{example}(squeezed state model)
Squeezed state model, which has four parameters, is defined by
\begin{eqnarray}
{\cal M}=\{\rho(z, \xi)\:|\; \rho(z,\xi)=\pi(|z,\xi\rgl),\; z,\xi\in\C\},
\nonumber
\end{eqnarray}
where
\begin{eqnarray}
|z,\xi\rgl&=&D(z)S(\xi)|0\rgl,\nonumber\\
D(z)&=&\exp(za^{\dagger}-\ol{z} a),\nonumber\\
S(\xi)&=&
\exp\left\{\frac{1}{2}(\, \xi a^{\dagger2}-\ol{\xi} a^2\,)\right\}.
\label{eqn:defsqz}
\end{eqnarray}
Here, the operator $a$ is defined as $a=(Q+iP)/\sqrt{2\hbar}$
where $P$ and $Q$ satisfy the canonical commutation relation
$[P,\, Q]=-i\hbar$.
Letting $z=\sqrt{\frac{2}{\hbar}}(\theta^1+i\theta^2)$, 
and $\xi=\theta^3 e^{-2i\theta^4}\;(0\leq \theta^3, 0\leq\theta^4 \leq 2\pi)$,
we have
\begin{eqnarray}
{\bf M}\left(\hlift_{|z,\xi\rgl}\partial_{\theta^1}\right)
&=&\frac{2i}{\hbar}(-P+\theta^2)|z,\xi\rgl,
\nonumber\\
{\bf M}\left(\hlift_{|z,\xi\rgl}\partial_{\theta^2}\right)
&=&\frac{2i}{\hbar}(Q-\theta^1)|z,\xi\rgl,
\nonumber\\
{\bf M}\left(\hlift_{|z,\xi\rgl}\partial_{\theta^3}\right)
&=&-i(e^{2i\theta^4}a^2- e^{-2i\theta^4}a^{\dagger 2})|z,\xi\rgl,
\nonumber\\
{\bf M}\left(\hlift_{|z,\xi\rgl}\partial_{\theta^4}\right)&=&
i4(\sinh^2 \theta^3)(a^{\dagger}a+\frac{1}{2})|z,\xi\rgl,\nonumber\\
&&-i2(\sinh\theta^3 )(\cosh\theta^3)(e^{2i\theta^4}a^2-
+e^{-2i\theta^4}a^{\dagger 2}))
|z,\xi\rgl,\nonumber\\
\end{eqnarray}
and
\begin{eqnarray}
J^S&=&\frac{2}{\hbar}\left[
\begin{array}{cccc}
\cosh 2\theta^3-\sinh 2\theta^3 \cos 2\theta^4 &
\sinh 2\theta^3\sin 2\theta^4 &0&0\\
\sinh 2\theta^3\sin 2\theta^4 & 
\cosh 2\theta^3+\sinh 2\theta^3 \cos 2\theta^4 &0&0\\
0&0&\hbar&0\\
0&0&0&\hbar\sinh^2 2\theta^3
\end{array}
\right],\nonumber\\
\tilde{J}&=&\frac{2}{\hbar}\left[
\begin{array}{cccc}
0&1&0&0\\
-1&0&0&0\\
0&0&0&-\frac{\hbar}{2}\sinh 2\theta^3\\
0&0&\frac{\hbar}{2}\sinh 2\theta^3&0
\end{array}
\right]
\nonumber
\end{eqnarray}

Coherency of this model is easily checked 
by theorem \ref{theorem:coherent:check},
\begin{eqnarray}
|{\rm det} J^S|=|{\rm det} \tilde{J}|=\frac{4}{\hbar^2}\sinh^2 2\theta^3.
\nonumber
\end{eqnarray}
\end{example}

\begin{example}(spin coherent model)
As is pointed out by Fujiwara $\cite{FujiwaraNagaoka:1996}$,
{\it spin coherent model} ${\cal M}_{s,s}$,
where ${\cal M}_{s,m}$ is defined by $(\ref{eqn:gscoherent})$,
is coherent.
\end{example}

\begin{example}(total space model)
 {\it The total space model} is the space of all the pure state ${\cal P}_1$
in finite dimensional Hilbert space ${\cal H}$.
By virtue of theorem \ref{theorem:coherent:spaniL},
the coherency of the model is proved by
checking that $span_{\C}{\sf L}$ is invariant by the multiplication of 
the imaginary unit $i$.
Let $|l\rgl$ be a horizontal lift of a tangent vector at $|\phi\rgl$.
Then, $i|l\rgl$ is 
also a horizontal lift of another tangent vector at $|\phi\rgl$,
because $|\phi\rgl+i|l\rgl dt$ is a member of $\tilde{\cal H}$.
\end{example}

\section{Informationally exclusive, and independent parameters}

In a $m$-parameter model ${\cal M}$,
we say parameter $\theta^i$ and $\theta^j$ are 
{\it informationally independent} at $\theta_0$, iff
\begin{eqnarray}
{\rm Re}\lgl l_i|l_j\rgl |_{\theta=\theta_0}
={\rm Im}\lgl l_i|l_j\rgl |_{\theta=\theta_0}
=0,
\nonumber
\end{eqnarray}
because if the equation holds true,
in the estimation of  the parameters $\theta^1,\theta^2$ of 
the $2$-parameter submodel ${\cal M}(1,2|\theta_0)$ of ${\cal M}$,
where
\begin{eqnarray}
{\cal M}(1,2|\theta_0)\equiv
\left\{\rho(\theta)\:\left|\: 
\theta=(\,\theta^1,\,\theta^2,\,\theta^3_0 ,...,\,\theta^m_0\,),\,
(\theta^1,\,\theta^2\,)\in \R^2\right.\right\},
\label{eqn:calm12}
\end{eqnarray}
both of the parameters can be estimated up to the accuracy
which is achieved 
in the estimation of the parameter of the $1$-parameter submodels 
${\cal M}(1|\theta_0)$ and ${\cal M}(2|\theta_0)$,
where
\begin{eqnarray}
{\cal M}(1|\theta_0)&\equiv&
\{\rho(\theta)\:|\: \theta=(\theta^1,\theta^2_0,...,\theta^m_0),\,
\theta^1\in\R\},
\nonumber\\
{\cal M}(2|\theta_0)&\equiv&
\{\rho(\theta)\:|\: \theta=(\theta^1_0,\theta^2,...,\theta^m_0),\,
\theta^2\in\R\}.
\nonumber
\end{eqnarray}
On the other hand, iff
\begin{eqnarray}
{\rm Re}\lgl l_1|l_2\rgl |_{\theta=\theta_0}=0,
\label{eqn:def:excv}
\end{eqnarray}
and ${\cal M}(1,2|\theta_0)$ is coherent, or equivalently,
\begin{eqnarray}
{\rm Im}\lgl l_1|l_2\rgl |_{\theta=\theta_0}
=\left.(\lgl l_1|l_1\rgl\lgl l_2|l_2\rgl)^{1/2}\right|_{\theta=\theta_0}
\nonumber
\end{eqnarray} 
hold true,
we say the parameters are {\it informationally exclusive} at $\theta_0$.

Fujiwara and Nagaoka $\cite{FujiwaraNagaoka:1996}$
showed that in the coherent model with
two orthogonal parameters, the attainable CR type bound is
achieved by applying the best measurement for each parameter 
alternatively to the system.
This fact implies that if two parameters are informationally exclusive,
the one of them do not contain any information about the other.
In fact, we have the following theorem. 

\begin{theorem}
If two parameters $\theta^1$ and  $\theta^2$ are informationally exclusive,
any unbiased measurement $M$ in ${\cal M}(1,2|\theta^0)$
which estimates $\theta^1$ as accurately as possible,
{\it i.e.},
\begin{eqnarray}
\int(\hat\theta^1-\theta^1)^2\Tr\rho(\theta_0)M(d\hat\theta)
=\left[J^{S-1}\right]^{11}
\label{eqn:excv:1}
\end{eqnarray}
can extract no information about
$\theta^2$ from the system, {\it i.e.},
\begin{eqnarray}
\forall B\subset {\R}^2 \:\:\:
\Tr \left(M(B)
\left. \frac{\partial \rho}{\partial \theta^2}
\right|_{\theta=\theta_0}
\right)={\rm Re}\lgl\phi|M(B)|l_2\rgl=0,
\label{eqn:excv:2}
\end{eqnarray}
and vice versa.
\end{theorem}

\begin{proof}
We prove the theorem only for the measurements which writes
\begin{eqnarray}
M(B)=\int_B |\hat\theta\rgl\lgl\hat\theta|\mu(d\hat\theta).
\nonumber
\end{eqnarray}
The proof for general case will be discussed elsewhere.
If $(\ref{eqn:excv:1})$ holds true,
as in the proof of lemma $\ref{lemma:genVZ}$
(see Ref.\cite{Holevo:1982}, p.88),
\begin{eqnarray}
\lgl\hat\theta|\left\{(\hat\theta^1-\theta^1)|\phi\rgl
-\frac{1}{\lgl l_1|l_1 \rgl}|l_1\rgl\right\}=0
\label{eqn:excv:3}
\end{eqnarray}
must hold.
On the other hand,
because of coherency of ${\cal M}(1,2|\theta^0)$ 
and $(\ref{eqn:def:excv})$,
for some real number $a$, we have
\begin{eqnarray} 
|l_2\rgl=ia|l_1\rgl,
\nonumber
\end{eqnarray}
by use of which it is shown that 
$(\ref{eqn:excv:2})$ is equivalent to
\begin{eqnarray}
{\rm Im}(\,\lgl\phi|\hat\theta\rgl\lgl\hat\theta|l_1\rgl\,)=0.
\nonumber
\end{eqnarray}
This equation is obviously true if 
$(\ref{eqn:excv:3})$ is true, and we have the theorem.
\end{proof}

\section{Direct sum of models}

For the submodels 
\begin{eqnarray}
{\cal M}_1\equiv{\cal M}(1,2,...,m_1\,|\,\theta_0),\:
{\cal M}_2\equiv{\cal M}(m_1,m_1+1,...,m\,|\,\theta_0) 
\nonumber
\end{eqnarray}
of ${\cal M}$,
where ${\cal M}(1,2,...,m_1\,|\,\theta_0)$ and 
${\cal M}(m_1,m_1+1,...,m\,|\,\theta_0)$ are defined almost 
in the same way as the definition $(\ref{eqn:calm12})$
${\cal M}(1,2\,|\,\theta_0)$,
we write 
\begin{eqnarray}
{\cal M}|_{\theta_0}={\cal M}_1\oplus {\cal M}_2|_{\theta_0},
\nonumber
\end{eqnarray}
and 
say that 
${\cal M}$ is sum of ${\cal M}$ and ${\cal M}$ at $\theta_0$.
$m-m_1$ is denoted by $m_2$.

\begin{lemma}
If any parameter of ${\cal M}_1$ is informationally independent
of any parameter of ${\cal M}_2$ at $\theta_0$, and
the weight matrix $G$ writes
\begin{eqnarray}
G=\left[
\begin{array}{cc}
G_1 & 0\\
0 & G_2
\end{array}
\right],
\nonumber
\end{eqnarray}
then
\begin{eqnarray}
\CR(G, \theta_0, {\cal M})=
\CR(G_1, \theta_0, {\cal M}_1)+\CR(G_2, \theta_0, {\cal M}_2)
\nonumber
\end{eqnarray}
\label{lemma:independent}
\end{lemma}

When the assumption of the lemma is satisfied,
${\cal M}_1$ and ${\cal M}_2$ are said to be 
{\it informationally independent} at $\theta_0$.

\begin{proof}
Let  $\hat\theta_1(\omega)$ and  $\hat\theta_2(\omega)$
be the vector whose components are
the estimates of  $\theta^1, \theta^2,...,\theta^{m_1}$,
and $\theta^{m_1+1}, \theta^{m_1+2},...,\theta^{m}$,
\begin{eqnarray}
\hat\theta_1(\omega)
&=&\left(\,\hat\theta^1(\omega),\,\hat\theta^2(\omega),..., 
\,\hat\theta^{m_1}(\omega)\,\right),
\nonumber\\
\hat\theta_2(\omega)
&=&
\left(\,\hat\theta^{m_1+1}(\omega),\,\hat\theta^2(\omega),..., 
\,\hat\theta^{m}(\omega)\,\right).
\nonumber
\end{eqnarray}
Then, if $(M,\hat\theta, \Omega)$ is locally unbiased,
$(M,\hat\theta_i, \Omega)$ is locally unbiased.
Therefore, we have,
\begin{eqnarray}
\{ M\, |\, \mbox{$(M,\hat\theta_i, \Omega)$ is locally unbiased}\}
\subset\{ M\, |\, \mbox{$(M,\hat\theta, \Omega)$ is locally unbiased}\}
\nonumber
\end{eqnarray}
which yields,
\begin{eqnarray}
& &\min\left\{\left.\Tr GV[\hat{\theta}|M]\, \right|\, 
	\mbox{$(M,\hat\theta, \Omega)$ is locally unbiased}\right\}\nonumber\\
&=&\min\left\{\left.\Tr G_1 V[\hat{\theta}_1 |M]\,\right|\, 
	\mbox{$(M,\hat\theta, \Omega)$ is locally unbiased}\right\}\nonumber\\
& &+\min\left\{\left.\Tr G_2 V[\hat{\theta}_2 |M]\,\right|\, 
   \mbox{$(M,\hat\theta, \Omega)$ is locally unbiased}\right\}\nonumber\\
&\ge& \min\left\{\left.\Tr G_1 V[\hat{\theta}_1 |M]\,\right|\, 
	\mbox{$(M,\hat\theta_1, \Omega)$ is locally unbiased}\right\}\nonumber\\
& &+\min\left\{\left.\Tr G_2 V[\hat{\theta}_2 |M]\, \right|\, 
	\mbox{$(M,\hat\theta_2, \Omega)$ is locally unbiased}\right\},\nonumber
\end{eqnarray}
or its equivalence,
\begin{eqnarray}
\CR(G, {\cal M})\ge
\CR(G_1, {\cal M}_1)+\CR(G_2, {\cal M}_2).
\label{eqn:cr>cr+cr}
\end{eqnarray}
Because ${\cal M}_1$ and ${\cal M}_2$ are informationally independent,
${\sf L}$ for ${\cal M}$ writes
\begin{eqnarray}
{\sf L}=\left[
\begin{array}{cc}
{\sf L}_1 & 0 \\
0 & {\sf L}_2
\end{array}
\right],
\nonumber
\end{eqnarray}
in the appropriate coordinate,
where ${\sf L}_1=[|l_1\rgl,|l_2\rgl,...,|l_{m_1}\rgl]$,
and ${\sf L}_2$ is in the same manner.
In that coordinate,  ${\sf X}$ writes
\begin{eqnarray}
{\sf X}=\left[
\begin{array}{cc}
{\sf X}_1 & {\sf X}_{12}\\
{\sf X}_{21} & {\sf X}_2
\end{array}
\right].
\nonumber
\end{eqnarray}
Therefore, if 
\begin{eqnarray}
{\rm Re} {\sf X}_i^*{\sf L}_i=I_{m_i},\,
{\rm Im}{\sf X}_i^*{\sf X}_i=0\,(i=1,2)
\nonumber
\end{eqnarray}
holds true,
the measurements corresponding to ${\sf X}$ is locally unbiased,
and
\begin{eqnarray}
& &\CR(G, {\cal M})\nonumber\\
&=&\min\left\{\, \Tr G {\sf X}^*{\sf X}\,\left|\,
	{\rm Re} {\sf X}^*{\sf L}=I_m,\,{\rm Im}{\sf X}^*{\sf X}=0\right. 
	\,\right\}
\nonumber\\
&\leq&\min\left\{\left. \sum_{i=1}^2\Tr G_i {\sf X}_i^*{\sf X_i}\, \right| \,
	{\rm Re} {\sf X}_i^*{\sf L_i}=I_{i_m},\,{\rm Im}{\sf X}_i^*{\sf X}_i =0,
		\,(i=1,2)\,\right\}
\nonumber\\
&=&\CR(G_1, {\cal M}_1)+\CR(G_2, {\cal M}_2),
\nonumber
\end{eqnarray}
which, mixed with $(\ref{eqn:cr>cr+cr})$ leads to the lemma.
\end{proof}

\chapter{Berry's phase in quantum estimation theory}
\section{Berry's phase}
In this section, we review the geometrical theory of Berry's phase.

Berry's phase was discovered by M.V. Berry in 1984\cite{Berry}, 
and confirmed by
many experimental facts\cite{Shapere}.
In 1987,  Aharonov and Anandan \cite{AA} pointed out that
Berry's phase is naturally interpreted as a curvature
in the fiber bundle over ${\cal P}_1$.
Actually, Berry's phase is nothing but the Uhlmann's curvature
restricted to ${\cal P}_1$ \cite{Uhlmann:1986}\cite{Uhlmann:1993}.

Uhlmann's RPF in the space ${\cal P}_1$ takes value 
in the set of unimodular complex number.
On the other hand,
Berry's phase takes value in real numbers.
They are related as
\begin{eqnarray}
\mbox{Berry's phase}=-i\ln (\mbox{Uhlmann's RPF}).
\nonumber
\end{eqnarray}

The Berry's phase for the infinitesimal loop $(\ref{loop})$
is calculated up to the second order of $d\theta$
as 
\begin{eqnarray}
\frac{1}{2i}\lgl\phi(0)|F_{ij}|\phi(0)\rgl d\theta^i d\theta^j + o(d\theta)^2
\nonumber
\end{eqnarray}
Because Berry's phase is independent of the choice of SLD,
we can take $L^S_i$ to be $2\partial_i\rho$.
Then, the phase is equal to 
\begin{eqnarray}
\frac{1}{2}\tilde{J}_{ij}  d\theta^i d\theta^j + o(d\theta)^2,
\nonumber
\end{eqnarray}
where $\tilde{J}_{ij}$ is equal to
$\frac{1}{2}{\rm Im}\lgl l_i|l_j \rgl$ .

Mathematically,  
\begin{eqnarray}
\sum _{i,j}\tilde{J}_{ij}d\theta^i d\theta^j
\label{eqn:berrycurv}
\end{eqnarray}
corresponds to the curvature form.

\section{Berry's phase in quantum estimation theory}
It must be noted that the curvature form is 
deeply related to the multiplication of the imaginary unit ${\bf D}$.
Actually, The curvature form $(\ref{eqn:berrycurv})$ 
is identical to a map from 
${\cal T}_{\theta}({\cal M})\times {\cal T}_{\theta}({\cal M})$
to $\R$ such that
\begin{eqnarray}
\ds{\sum_{ij}}\tilde{J}_{ij}d\theta^id\theta^j\: :\: 
(\partial_i,\partial_j)\longrightarrow
\left\lgl \partial_i,\frac{1}{2}{\bf D}\partial_j \right\rgl.
\nonumber
\end{eqnarray}

Hence, the eigenvalues of ${\bf D}$
can be interpreted in terms of Berry's phase.
Concretely speaking,
taking the coordinate system which is orthonormal at $\theta$ 
in terms of the metric $\lgl*,*\rgl_{\theta}$,
they are the half of the Berry's phase obtained when the state
goes around the infinitesimal loop $(\ref{loop})$.
Especially, when the model is only with two  parameters,
the eigenvalues of ${\bf D}$ are the half of the Berry's phase per
unit area, where unit of the area is naturally induced from
the metric $\lgl *,* \rgl_{\theta}$.

Therefore, we can roughly say that
the more the Berry's phase for  the loop $(\ref{loop})$, 
the harder it is to estimate $\theta^i$ and $\theta^j$ simultaneously.
Namely, $\theta^i$ and $\theta^j$ are informationally independent
iff the Berry's phase for  the loop $(\ref{loop}g)$ vanishes and 
$\lgl \partial_i, \partial_j\rgl_{\theta}=0$ holds.
On the other hand, iff the Berry's phase for  the loop $(\ref{loop})$
is maximal and $\lgl \partial_i, \partial_j\rgl_{\theta}=0$ holds, 
{\it i.e.,} $\beta=1$, the two parameters
are informationally exclusive.

This discussion is parallel to that 
in the section $\ref{sec:estphase}$,
which was about relations between Uhlmann's parallelism and
the noncommutative nature of 
the quantum estimation theory of the faithful model.
Because Berry's phase is nothing but the restriction of
Uhlmann's RPF to the pure state model, this parallelism is
natural.

What about the models with arbitrary number of parameters?
By virtue of theorem $\ref{theorem:purelsls0}$,
if Berry's phase for any closed loop vanishes,
the model is localy quasi-classical.
For general pure state models, we have the following theorem.

\begin{theorem}
For any pure state model,
\begin{eqnarray}
&&\CR(J^{S})\nonumber\\
&=&\Tr\{{\rm Re}(I_m+i J^{S-1/2}\tilde{J}J^{S-1/2})^{1/2}\}^{-2}\nonumber\\
&=&\sum_{\alpha \in\{\mbox{eigenvalues of ${\bf D}$}\}} 
\frac{2}{1+(1-|\alpha|^2)^{1/2}}.
\nonumber
\end{eqnarray}
\end{theorem}

The estimation theoretical siginificance of $\CR(J^S)$
is hard to verify. However, this value remains invariant under
any transform of the coordinate in the model ${\cal M}$,
and can be an good index of distance between ${\cal V}$ and $J^{S-1}$.

\begin{proof} 
Because 
 $\CR(J^S)$
is invariant by any affine coordinate transform in the model ${\cal M}$,
we choose a coordinate in which
$J^S$ writes $I_m$ and $\tilde{J}$ writes
\begin{eqnarray}
\tilde{J}=\left[
\begin{array}{ccccccccc}
0& -\beta_1& 		0&\cdots&0&0&0&\cdots&0\\
\beta_1&0 & 		0&\cdots&0&0&0&\cdots&0\\
0& 0& 		\ddots&\ddots&\vdots&\vdots&\vdots&\ddots&\vdots\\
\vdots&\vdots&\ddots&\ddots&0&0&0&\cdots&0\\
0&0&		\cdots &0&0& -\beta_{l}&0&\cdots&0\\
0&0&		\cdots &0&\beta_{l}&0 &0&\cdots&0\\
0&0&		\cdots &0&    0    &0&0&\ddots&\vdots\\
\vdots&\vdots&\ddots&\vdots&\vdots&\vdots&\ddots&\ddots&\vdots\\
0&0&		\cdots &0&    0    &0 &\cdots&\cdots&0
\end{array}
\right]. 
\nonumber
\end{eqnarray}
Then, The model  ${\cal M}$ is decomposed into
the direct sum of the submodels one or two parameter ${\cal M}_{\kappa}$,
\begin{eqnarray}
{\cal M}=\bigoplus_{\kappa} {\cal M}_{\kappa},
\nonumber
\end{eqnarray}
where 
any two submodels  ${\cal M}_{\kappa}$ and ${\cal M}_{\kappa'}$
are informationally independent, and 
$\tilde{J}$ of a two parameter submodel ${\cal M}_{\kappa}$
is 
\begin{eqnarray}
\left[
\begin{array}{cc}
0& -\beta_{\kappa} \\
\beta_{\kappa}&0 
\end{array}
\right].
\nonumber
\end{eqnarray}
Because the weight matrix $J^S=I_m$ writes in the form of 
direct some of the weight matrix on the model ${\cal M}_{\kappa}$, 
by virtue of lemma $\ref{lemma:independent}$ 
and the equation $(\ref{eqn:minvv})$,
we have the theorem.
\end{proof}

\section{Berry's phase in the global theory of quantum estimation}

In this section, we present a geometrical sufficient condition 
for the pure state model ${\cal M}$ to be quasi-classical
in the sense of section $\ref{sec:questf}$.

For simplicity,
we say that the manifold ${\cal N}$ in ${\cal P}_1$ 
is a {\it horizontal lift} of the model ${\cal M}$ if
\begin{eqnarray}
\pi({\cal N})&=&{\cal M},\\
\forall |\phi(\theta)\rgl\in {\cal N},& &
\frac{1}{2}\frac{\partial}{\partial \theta^i}|\phi(\theta)\rgl
\in{\cal LS}_{|\phi(\theta)\rgl}.
\nonumber
\end{eqnarray} 
The horizontal lift ${\cal N}$ exists iff ${\cal M}$ 
is quasi-classical.

\begin{theorem}
If  the model ${\cal M}$ is parallel, that model is  
quasi- classical
in the sense of section $\ref{sec:questf}$.
\label{th:p-plgl1}
\end{theorem}
\begin{proof}
First, apply  Schmidt's orthonormalization to
the horizontal lift ${\cal N}$ of ${\cal M}$,
to obtain the orthonomal basis ${\bf B}=\{ |e_i\rgl\,|\, i=1, 2, ...\}$
such that ${\cal N}$ is a subset of the real span of ${\bf B}$.
We immerse Hilbert space ${\cal H}$ into $L^2(\R,\,\C)$
as 
\begin{eqnarray}
|\phi(\theta)\rgl
=\sum_i a_i|e_i\rgl\mapsto {\bf \iota}\sum_i a_i(\theta)\psi_i(x),
\nonumber
\end{eqnarray}
where $\{\ psi_i(x)\,|\, i=1, 2, ...\}$ is an orthonormal basis
in $L^2(\R,\,\C)$. Then, letting $E(dx)=|x\rgl\lgl x|dx$,
the triplet $(\hat\theta(\theta),\, E\, , \R)$ is 
one of the best estimators.
This assertion is easily proved by caluculating 
the Fisher information matrix of the family,
\begin{eqnarray}
\{ p(x\, |\, \theta)=\sum_i (a_i(\theta))^2 |\psi_i(x)|^2\}
\end{eqnarray}
of probability distributions.
\end{proof}

The converse of the latter theorem is, however, not true, because
the following counter-examples exist.

\begin{example}
We consider the {\it position shifted model} which is defined by 
\begin{eqnarray} 
&&{\cal M}_x=\pi({\cal N}_x)\nonumber\\
&&{\cal N}_x=\{ |\phi(\theta)\rgl\: |\: 
		|\phi(\theta)\rgl=
		 \mbox{\it const.}\times(x-\theta)^2 e^{-(x-\theta)^2+ig(x-\theta)},\:
		\theta\in\R \},
\nonumber
\end{eqnarray}
where $c$ is a normalizing constant,
$g$ the function such that
\begin{eqnarray}
g(x)=\left\{
\begin{array}{cc}
0 & (x\geq 0),\\
\alpha & (x< 0).
\end{array}
\right.
\nonumber
\end{eqnarray}

Then, as easily checked,
${\cal N}_x$ is a horizontal lift of the model ${\cal M}_x$,
and 
$\lgl\phi(\theta)|\phi(\theta')\rgl$ is not real 
unless $\alpha=n\pi\,(n=0,1,...)$.
However, SLD CR bound is uniformly attained by
the measurement obtained 
by the spectral decomposition $E_x(dx)=|x\rgl\lgl x| dx$ 
of the position operator,where $|x_0\rgl=\delta(x-x_0)$.
as is  checked by
comparing SLD Fisher information of the model ${\cal M}_x$
and the classical Fisher information of 
the probability distribution family
\begin{eqnarray}
\{ p(x|\theta)\,|\, 
	p(x|\theta)=|\lgl\phi(\theta)|x\rgl|^2,\;
	\theta\in\R\}.
\nonumber
\end{eqnarray}

Note that $|\phi(\theta)\rgl$ is 
an eigenstate of the Hamiltonian 
\begin{eqnarray}
H(\theta)=-\frac{\hbar^2}{2m}\frac{d^2}{dx^2}
	+\frac{\hbar^2}{m}\left( 2(x-\theta)^2+\frac{1}{(x-\theta)^2}\right),
\nonumber
\end{eqnarray}
whose potential  has two wells
with infinite height of wall between them.
\end{example}

\begin{example}
Let ${\cal H}$ be $L^2([0,\,2\pi],\,\C)$,
and  define a one parameter model ${\cal M}$ such that,
\begin{eqnarray}
&&{\cal M}=\pi({\cal N})\nonumber\\
&&{\cal N}=\{|\phi(\theta)\rgl\: |\: 
				|\phi(\theta)\rgl=\mbox{\it const.}\times(2-\cos\omega)\, 
	e^{i\alpha(f(\omega-\theta)+\theta)},\:(0\leq\omega, \theta< 2\pi) \},
\nonumber\\
\label{eqn:ring}
\end{eqnarray}
where $\alpha$ is a real number and
$f$ the function defined by
\begin{eqnarray}
f(\omega-\theta)=\left\{
\begin{array}{cc}
\omega-\theta & (\omega-\theta\geq 0)\\
\omega+2\pi-\theta &(\omega-\theta < 0)
\end{array}
\right. .
\nonumber
\end{eqnarray}
Physically, $(\ref{eqn:ring})$ is an eigenstate of the Hamiltonian $H$
such that,
\begin{eqnarray}
H(\theta)=-\frac{\hbar^2}{2m}\left(\frac{d}{d\omega}-i\alpha\right)^2
+\frac{A-B\cos(\omega-\theta)}{2-\cos(\omega-\theta)},
\nonumber
\end{eqnarray}
which characterize the dynamics of an electron
confined to the one-dimensional ring which encircles magnetic flax
$\Phi=2\pi\alpha c/e$, where $m$ is the mass of the electron,
$-e$ the charge of the electron, $c$ the velocity of light,
and $A$, $B$ the appropriately chosen constant.

It is easily checked that ${\cal N}$ is  a  horizontal lift of the model
${\cal M}$,
and that 
the model ${\cal M}$ is not parallel unless $\alpha= n\pi\,(n=0,1,...)$.
However, consider
the projection valued measure $E_{\omega}$ such that
\begin{eqnarray}
E_{\omega}(d\omega)=|\omega\rgl\lgl\omega| d\omega,
\nonumber
\end{eqnarray}
where $|\omega_0\rgl=\delta(\omega-\omega_0)$.
Then, it is easily checked that
the classical Fisher information of 
the probability distribution family
\begin{eqnarray}
\{ p(\omega|\theta)\,|\,  
p(\omega|\theta)=|\lgl \phi(\theta)|\omega\rgl |^2,\,
0\leq\omega,\theta< 2\pi\}
\nonumber
\end{eqnarray}
is equal to the SLD Fisher information of ${\cal M}$.
\end{example}

\section{Antiunitary operators}

The  transformation $A$ 
\begin{eqnarray}
|\tilde{a}\rgl=A|a\rgl,\:\:\:  |\tilde{b}\rgl=A|b\rgl
\nonumber
\end{eqnarray}
is said to be {\it antiunitary} iff
\begin{eqnarray}
\lgl \tilde{a} |\tilde{b}\rgl&=&\ol{\lgl a| b \rgl},\nonumber\\
A(\alpha|a\rgl+\beta|b\rgl)&=&\ol{\alpha}A|a\rgl+\ol{\beta}A|b\rgl,
\nonumber
\end{eqnarray}
where $\ol{*}$ means complex conjugate.

Fix  an orthonormal basis ${\sf B}=\{|i\rgl\,|\, i=1,2, ... ,d\}$,
and we can then define antiunitary operator $K_{\sf B}$ which
takes complex conjugate of any components in this basis, 
\begin{eqnarray}
K_{\sf B}=\sum_i\alpha_i |i\rgl= \sum_i \ol{\alpha}_i|i\rgl.
\nonumber
\end{eqnarray}
For different basis ${\sf B}, {\sf B}'$,
we have 
\begin{eqnarray}
K_{\sf B}=U K_{\sf B}'U^*,
\nonumber
\end{eqnarray}
where $U$ is a unitary operator corresponding 
to the change of the basis.

Suppose that any member of the manifold 
${\cal N}=\{|\phi \rgl\}$ in $\tilde{\cal H}$ 
is invariant by the antiunitary operator $A$,
and let $|\tilde{\phi}\rgl=A|\phi\rgl,\,|\tilde{\phi}'\rgl=A|\phi'\rgl$.
Then, we have
\begin{eqnarray}
\lgl \phi |\phi'\rgl=\lgl \tilde{\phi'}|\tilde{\phi}\rgl=\lgl \phi'|\phi\rgl
\in\R.
\nonumber
\end{eqnarray}
Conversely, if $\lgl \phi|\phi'\rgl$ is real for any 
$|\phi\rgl,\,|\phi'\rgl\in{\cal N}$,
by Schmidt's orthonormalization, we can obtain the basis ${\sf B}$
such that  ${\cal N}$ is subset of the real span of ${\sf B}$,
which means any member of ${\cal N}$ is invariant by 
the antiunitary operator $K_{\sf B}$.

Therefore, the premise of the statement of 
theorems $\ref{th:p-plgl1}$-$\ref{th:p-plgl2}$
is satisfied iff the horizontal lift of the model is 
invariant by some antiunitary operator.

\section{Time reversal symmetry}
As an example of the antiunitary operator,
we discuss {\it time reversal operator} 
(see Ref.\cite{Sakurai}, pp. 266-282).
The time reversal operator $T$ is an antiunitary operator
in $L^2(\R^3,\,\C)$ which transforms 
the wave function $\psi(x)\in L^2(\R^3,\,\C)$ as:
\begin{eqnarray}
T \psi(x)=\ol{\psi(x)}=K_{\{|{\bf x}\rgl\}}\psi(x).
\nonumber
\end{eqnarray}
The term `time reversal' came from the fact
that if  $\psi({\bf x}, t)$ is a solution of
the Sch\"odinger equation
\begin{eqnarray}
i\hbar\frac{\partial\psi}{\partial t}=
\left(-\frac{\hbar^2}{2m}\nabla^2+V\right)\psi,
\nonumber
\end{eqnarray}
then $\ol{\psi({\bf x}, -t)}$ is also its solution.

The operator $T$ is sometimes called {\it motion reversal operator},
since it transforms the momentum eigenstate 
$e^{i{\bf p}\cdot{\bf x}/\hbar}$ corresponding to eigenvalue ${\bf p}$
to the eigenstate  $e^{-i{\bf p}\cdot{\bf x}/\hbar}$
corresponding to eigenvalue $-{\bf p}$.

Define the {\it position shifted model} by
\begin{eqnarray}
{\cal M}_{\bf x}=
\{\rho(\theta)\,|\, 
\rho(\theta)=\pi(\psi({\bf x}-{\bf x}_0)\,),\,{\bf x}_0 \in\R^3\},
\nonumber
\end{eqnarray}
and suppose that any member of 
the horizontal lift ${\cal N}_{\bf x}$
of the model ${\cal M}_{\bf x}$
has time reversal symmetry.
Then, since time reversal operator $T$ is antiunitary,
the model  ${\cal M}_{\bf x}$
is quasi-classical in the wider sense.
The spectral decomposition of the position operator
gives optimal measurement.

Now, we discuss the generalization of time reversal operator.
The antiunitary transform
\begin{eqnarray}
T_{\alpha}\, : \, e^{i{\bf p}\cdot{\bf x}/\hbar}
\rightarrow
e^{i\alpha({\bf p})}\,e^{-i{\bf p}\cdot{\bf x}/\hbar}
\nonumber
\end{eqnarray}
is also called
motion reversal operator, or time reversal operator.

If any member $\psi({\bf x}-{\bf x}_0)$ of the horizontal lift 
${\cal N}_{\bf x}$ 
of the position shifted model ${\cal M}_{\bf x}$
is invariant by the time reversal operator $T_{\alpha}$,
\begin{eqnarray}
\int_{\R^3} 
\psi({\bf x}-{\bf x}_0)\,\ol{\psi({\bf x}-{\bf x}'_0)}\,
d{\bf x}
\in\R
\label{eqn:psipsi}
\end{eqnarray}
holds true for any ${\bf x}_0,\, {\bf x}'_0$,
which is equivalent to 
the premise of theorems $\ref{th:p-plgl1}$-$\ref{th:p-plgl2}$.

Conversely, if $(\ref{eqn:psipsi})$ holds true,
Fourier transform of $(\ref{eqn:psipsi})$ leads to
\begin{eqnarray}
|\Psi({\bf p})|^2=|\Psi(-{\bf p})|^2,
\nonumber
\end{eqnarray}
where
\begin{eqnarray} 
\Psi({\bf p})=
\frac{1}{\sqrt{2\pi}}
\int\psi({\bf x})e^{-i{\bf p}\cdot{\bf x}/\hbar} d{\bf x}.
\nonumber
\end{eqnarray}
Therefore,
any member of ${\cal N}_{\bf x}$ is transformed  to itself
by the time reversal operator $T_{\alpha}$
such that
\begin{eqnarray}
T_{\alpha}\, : \, e^{i{\bf p}\cdot{\bf x}/\hbar}
\rightarrow
e^{i(\beta({\bf p})+\beta(-{\bf p})\,)}\,e^{-i{\bf p}\cdot{\bf x}/\hbar},
\nonumber
\end{eqnarray}
where
\begin{eqnarray} 
e^{i\beta({\bf p})}=\frac{\Psi({\bf p})}{|\Psi({\bf p})|}.
\nonumber
\end{eqnarray}

\begin{theorem}
 $(\ref{eqn:psipsi})$ is equivalent to the existence of
the time reversal operator which transforms
any member of the horizontal lift ${\cal N}_{\bf x}$ to itself.
\end{theorem}

\chapter{Uncertainty principle in view of quantum estimation theory}
\section{The position-momentum shift model}
In this section and the next, we examine the position-momentum uncertainty
in view of quantum estimation theory.

First, it must be emphasized that 
so-called `Heisenberg's uncertainty',
\begin{eqnarray}
\lgl (\Delta X)^2\rgl\lgl (\Delta P)^2\rgl \ge \frac{\hbar^2}{4},
\label{eqn:robertson}
\end{eqnarray}
where $\lgl (\Delta X)^2\rgl$ stands for
\begin{eqnarray} 
\lgl\phi|(X-\lgl\phi|X|\phi\rgl)^2|\phi\rgl,
\nonumber
\end{eqnarray}
has nothing to do with the Heisenberg's gedanken experiment
which deals with the simultaneous measurement of the position
and the momentum.

$\lgl (\Delta X)^2\rgl$ (, or $\lgl (\Delta P)^2\rgl$) 
in $(\ref{eqn:robertson})$
is the variance of the data when
only position (, or momentum) is measured. 
Therefore, 
$(\ref{eqn:robertson})$ corresponds to the
experiment where 
the position is measured for the one of the group of identical particles
and the momentum for the other group of identical particles.
As a matter of fact, 
the inequality $(\ref{eqn:robertson})$,
is  derived by H. P. Robertson $\cite{Robertson}$
and some careful researcher call
the inequality Robertson's uncertainty
(Heisenberg himself had nothing to do with the inequality).

The purpose of this  chapter is 
to examine the simultaneous measurement of the position
and the momentum from the estimation theoretical viewpoint.
However, since
 the measurement obtained by the spectral decomposition of 
 the position operator
differs from that of the momentum operator,
the simultaneous measurement in exact sense is impossible.
Here, we formulate the problem as a estimation of the {\it shift parameters}
$x_0$ and $p_0$,
in the {\it position-momentum shifted model}
\begin{eqnarray}
{\cal M}_{xp}=\{\rho(\theta)\: |\: 
		\rho(x_0,p_0)=\pi(D(x_0,p_0)|\phi_0\rgl),\,
		\theta=(x_0,p_0)\in\R^2 \},
\nonumber
\end{eqnarray}
where 
\begin{eqnarray}
D(x_0,p_0)=\exp \frac{i}{\hbar}(p_0 X- x_0 P)
\label{eqn:weilop}
\end{eqnarray}
and $|\phi_0\rgl=\phi_0(x)$ is a member of $L^2(\R,\C)$ such that
\begin{eqnarray}
\lgl\phi_0|X|\phi_0\rgl=\lgl\phi_0|P|\phi_0\rgl=0.
\nonumber
\end{eqnarray}

\section{The estimation of the shift parameters}
\label{sec:estshift}
Our purpose is to examine how efficiently we can estimate 
the shift parameters $\theta=(x_0,p_0)$.

The  horizontal lifts of 
$\partial/\partial x_0$, $\partial/\partial p_0$ are
\begin{eqnarray}
\hlift_{|\phi(\theta)\rgl}
\left(\frac{\partial}{\partial x_0}\right)
&=&-\frac{2i}{\hbar}\Delta P_{\theta}|\phi(\theta)\rgl,\nonumber\\
\hlift_{|\phi(\theta)\rgl}
\left(\frac{\partial}{\partial p_0}\right)
&=&\frac{2i}{\hbar}\Delta X_{\theta}|\phi(\theta)\rgl,
\nonumber
\end{eqnarray}
where 
\begin{eqnarray}
\lgl  A \rgl_{\theta}&\equiv& \lgl\phi(\theta)| A |\phi(\theta)\rgl,\nonumber\\
\Delta A_{\theta} &\equiv& A -\lgl  A \rgl_{\theta}
\nonumber
\end{eqnarray}
and the SLD Fisher information matrix $J^S(\theta)$ is,
\begin{eqnarray}
&\left[J^S(\theta)\right]_{x_0,x_0}&
=\frac{4}{\hbar^2}\lgl (\Delta P_{\theta})^2 \rgl,\nonumber\\
&&=\frac{1}{\hbar^2}\lgl (\Delta P_{(0,0)})^2 \rgl,\nonumber\\
&\left[J^S(\theta)\right]_{p_0,p_0}&=
\frac{4}{\hbar^2}\lgl (\Delta X_{(0,0)})^2 \rgl, \nonumber\\
&\left[J^S(\theta)\right]_{x_0,p_0}&=
\frac{4}{\hbar^2} Cov(X,P)_{(0,0)},
\nonumber
\end{eqnarray}
where 
$Cov(X,P)_{\theta}\equiv 
\frac{1}{2}
\lgl (\Delta X_{\theta}\Delta P_{\theta} 
			+\Delta P_{\theta}\Delta X_{\theta}) $.
The absolute value $\beta(\theta)$ 
of the eigenvalue of ${\bf D}$ at $\theta$ is calculated as,
\begin{eqnarray}
\beta(\theta)&=&
\frac{1}{2i}\frac{\lgl [P,X]\rgl_{(0,0)}}
{\sqrt{\lgl(\Delta X_{(0,0)})^2 \rgl\lgl(\Delta P_{(0,0)})^2 \rgl
-(Cov(X,P)_{(0,0)})^2}}
\nonumber\\
&=&\frac{-\hbar}
{2\sqrt{\lgl(\Delta X_{\theta})^2\rgl\lgl (\Delta P_{\theta})^2 \rgl
-(Cov(X,P)_{\theta})^2}}.
\label{eqn:sftb}
\end{eqnarray}
Notice that $J^S(\theta)$ and $\beta(\theta)$  are independent of 
the true value of parameters.

We are interested in 
the attainable CR type bound
and in the index $\beta$ of noncommutative nature of the model. 

As for $\beta$, $(\ref{eqn:sftb})$ indicates that
the larger the formal `covariance matrix'
\footnote{ Note that this formal `covariance matrix' is not
equal to 
the covariance matrix of any measurement related to position or momentum }
 of $X$ and $P$
\begin{eqnarray}
\left[
\begin{array}{cc}
\lgl (\Delta P_{(0,0)})^2\rgl & Cov(X,P)_{(0,0)}\\
Cov(X,P)_{(0,0)} & \lgl (\Delta X_{(0,0)})^2\rgl 
\end{array}
\right]
\nonumber
\end{eqnarray}
is,
the smaller the noncommutative nature between $x_0$ and $p_0$.

As for the attainable CR type bound,
because  
the SLD Fisher information matrix is proportional to the
formal `covariance matrix' and $\beta$ decreases as 
the determinant of the `covariance matrix' increases, 
we can metaphorically say that
the larger the `covariance matrix' implies the possibility
of more efficient estimate of $x_0$ and $p_0$,
which is seemingly paradoxical.

We examine these points in  the {\it shifted harmonic oscillator  model}
 ${\cal M}_{xp,n}$, which is defined to be 
the position-momentum shifted model in which 
 $\phi_0(x)$ is equal to 
the $n$th eigenstate $|n\rgl$ of the harmonic oscillator.
For ${\cal M}_{xp,n}$, we have 
\begin{eqnarray}
&[J^S(\theta)]_{x_0,x_0}&=\frac{4}{\hbar} \left(n+\frac{1}{2}\right), 
\nonumber\\
&[J^S(\theta)]_{p_0,p_0}&=\frac{4}{\hbar}\left(n+\frac{1}{2}\right), 
\nonumber\\
&[J^S(\theta)]_{x_0,p_0}&=[J^S(\theta)]_{(x_0,p_0)}
=0,
\nonumber
\end{eqnarray}
and
\begin{eqnarray}
\beta(\theta)=\frac{-1}{2(n+1/2)}.
\label{eqn:sfnb}
\nonumber
\end{eqnarray}
Hence, if $n$ is large, 
`noncommutative nature' of the parameters is small.

What about the efficiency of the estimate ?
We define ${\cal M}'_{xp,n}$ by
normalizing the parameters in ${\cal M}_{xp,n}$ as
\begin{eqnarray}
\theta=(x_0, p_0)\longrightarrow
\theta=(n+1/2)^{-1/2}x_0,(n+1/2)^{-1/2} p_0).
\nonumber
\end{eqnarray}
Then, directly from the definition, 
\begin{eqnarray}
\CR (G,\theta,{\cal M}_{xp,n})=\frac{1}{n+1/2}\CR(G,\theta,{\cal M}'_{xp,n}),
\label{eqn:crmcrm}
\end{eqnarray}
and the SLD Fisher information matrix of  ${\cal M}'_{xp,n}$
is equal to $\frac{4}{\hbar}I_m$ for any $n$.

Because the index $\beta$ is unchanged by the change of the parameter, 
\begin{eqnarray}
\CR (G,\theta,{\cal M}'_{xp,0})\ge 
\CR (G,\theta,{\cal M}'_{xp,1})\ge
... \ge \CR (G,\theta,{\cal M}'_{xp,n})\ge ...,
\nonumber
\end{eqnarray}
which, combined with $(\ref{eqn:crmcrm})$ leads to
\begin{eqnarray}
\CR (G,\theta,{\cal M}_{xp,0})\ge 
\CR (G,\theta,{\cal M}_{xp,1})\ge
... \ge \CR (G,\theta,{\cal M}_{xp,n})\ge ...,
\nonumber
\end{eqnarray}
for any $\theta$ and any  $G$.
Therefore, 
if the `covariance matrix' larger,
the more efficient estimate of the parameter is possible.
Especially, when $n=0$, or in the case of the so-called 
`minimum uncertainty state', the efficiency of the estimation is
the lowest.

Especially, when $n=0$, or in the case of the so-called 
`minimum uncertain state',
the position parameter $x_0$ and the momentum parameter $p_0$
are maximally `noncommutative' in the sense $\beta$ is larger than
that of any other ${\cal M}_{xp,n}\,(n\neq 0)$.
It is easily shown that $\beta$ is maximal, or coherent,
iff $|\phi_0\rgl$ is in the squeezed state, or
$|\phi_0\rgl=S(\xi)|0\rgl$,
where $S(\xi)$ is 
the operator defined by $(\ref{eqn:defsqz})$.
In addition, the efficiency of the estimation is
lower than any other ${\cal M}_{xp,n}\,(n\neq 0)$.

If $n$ is very  large, how efficiently can we estimate? 
Given $N$ particles, we divide them into two groups,
to one of which we apply the best measurement for
$x_0$ and to the other of which we apply the best measurement for $p_0$ .  
The parameter $x_0$ and $p_0$ is estimated only from the data from the
first group and the second group, respectively.
Then, the attained efficiency of the estimation of
$x_0$ is
\begin{eqnarray}
\frac{\hbar}{(N/2)\,4(n+1/2)}
\nonumber
\end{eqnarray}
and the efficiency of the estimation of $p_0$ is
\begin{eqnarray}
\frac{\hbar}{(N/2)\,4(n+1/2)}
\nonumber
\end{eqnarray}
which are combined to yield the efficiency 
of this estimate par sample,
\begin{eqnarray}
 g_{x_0} [V[M]]_{x_0,x_0}+g_{p_0} [V[M]]_{p_0,p_0}
=\frac{\hbar g_{x_0}}{2(n+1/2)}
+\frac{\hbar g_{p_0}}{2(n+1/2)},
\label{eqn:effmle}
\end{eqnarray}
in this estimation scheme.
Therefore, we have
\begin{eqnarray}
\CR (\diag (g_1,g_2),{\cal M}_{xp,n})
\leq
\frac{\hbar g_{x_0}}{2(n+1/2)}
+\frac{\hbar g_{p_0}}{2(n+1/2)},
\nonumber
\end{eqnarray}
which implies that arbitrarily precise estimate is possible
if ${\cal M}_{xp,n}$ with large enough $n$ is fortunately given.

The efficiency $(\ref{eqn:effmle})$
is achievable by the following maximum likelihood estimator
up to the first order of $1/N$:
\begin{eqnarray}
\hat{x}_0 &=& 
\argmax_{x_0} \sum_{j=1}^{N/2} \ln p(x_j - x_0),\nonumber\\
\hat{p}_0 &=& 
\argmax_{p_0} \sum_{j=1}^{N/2} \ln \tilde{p}(p_j - p_0),
\label{eqn:mlexp}
\end{eqnarray}
where
$x_1,x_2,..., x_{N/2}$ 
and
 $p_1,p_2,..., p_{N/2},$ be 
is 
data produced by the measurement of 
the position and the momentum of the given states,
and  their probability distribution is denoted by
$p(x)$ and $\tilde{p}(p)$, respectively.

\section{Planck's constant and Uncertainty}
In this section, we focus on Planck's constant.
Let
\begin{eqnarray}
P'=\hbar^{-1/2} P,\, X'=\hbar^{-1/2} X,\nonumber\\
p'_0=\hbar^{-1/2} p_0,\, x'_0=\hbar^{-1/2} x_0,
\nonumber
\end{eqnarray}
and define
\begin{eqnarray}
{\cal M}'_{xp}=\{\rho(p'_0, x'_0)\,|\, 
			\rho(p'_0, x'_0)=\pi(\exp i(p'_0 X'-x'_0 P')|\phi_0\rgl\,)\}.
\nonumber
\end{eqnarray}
Since Planck's constant does not appear
in the commutation relation $[P', X']= -i$,
the attainable CR type bound of ${\cal M}_{xp}'$ is
not dependent on $\hbar$ 
if definition of $|\phi_0\rgl$ does not include $\hbar$.

If $\CR({\cal M}'_{xp})$ has some finite value,
the identity
\begin{eqnarray}
\CR({\cal M}_{xp})=\hbar\CR({\cal M}'_{xp})
\nonumber
\end{eqnarray}
implies 
\begin{eqnarray}
\lim_{\hbar\rightarrow 0} \CR({\cal M}_{xp})=0.
\nonumber
\end{eqnarray}
Therefore, in the limit of $\hbar\rightarrow 0$,
position and momentum can be simultaneously 
measured as precisely as needed.

However, it must be noticed that
the attainable CR type bound of the {\it position shifted model}
\begin{eqnarray}
{\cal M}_x=\{\rho(x_0)\,|\, 
	\rho(x_0)=\pi(\exp (-ix_0 P)|\phi_0\rgl\,)\}
\nonumber
\end{eqnarray}
and of the {\it momentum  shifted model}
\begin{eqnarray}
{\cal M}_p=\{\rho(p_0)\,|\, 
	\rho(p_0)=\pi(\exp (ip_0 P)|\phi_0\rgl\,)\}
\nonumber
\end{eqnarray}
also tends to zero as $\hbar\rightarrow 0$,
and that the ratio
\begin{eqnarray}
\frac{\CR({\cal M}_{xp})}{\sqrt{\CR({\cal M}_x)\:\CR({\cal M}_p)}}
=
\frac{\CR({\cal M}_{xp}')}{\sqrt{\CR({\cal M}'_x)\:\CR({\cal M}'_p)}}
\nonumber
\end{eqnarray}
is independent of $\hbar$, where
${\cal M}'_x$ and ${\cal M}'_p$ are defined
in the same manner as ${\cal M}'_{xp}$.
Therefore, noncommutative nature of the model is unchanged
even if $\hbar$ tends to $0$.
Actually, as in $(\ref{eqn:sfnb})$the index $\beta$ of the noncommutative nature of the model
is independent of $\hbar$.

\begin{remark}
Notice the discussion in this section is essentially valid for
the mixed position-momentum shifted model,
\begin{eqnarray}
\{\rho(p_0, x_0)\,|\, 
	\rho(p_0, x_0)
=D(x_0,\,p_0\,)
\rho_0 D^*(x_0,\,p_0\,),\,(x_0,\,p_0)\in\R^2\},
\nonumber
\end{eqnarray}
where the state $\rho_0$ is mixed, 
and $D(x_0,\,p_0\,)$ is the operator defined by $(\ref{eqn:weilop})$.
\end{remark}

\section{Semiparametric estimation of the shift parameters}
In section $\ref{sec:estshift}$, our conclusion is that
if we are fortunate enough, we can estimate 
the average of the position and the momentum 
with arbitrary accuracy at the same time.

One may argue that this is because we make
full use of knowledge about the shape of the wave function
of the given state. However, this argument is not thoroughly true.

In the classical estimation theory, we have the following very strong result.
Suppose that we are intersected in the mean value $\theta\in\R$ 
of the probability distribution,
and that  the shape of the probability distribution is unknown 
except it is  symmetric around $\theta$. 
In other words, we set up the {\it semiparametric model} such that,
\begin{eqnarray}
\{ p(x\, |\, \theta, g)\: |\: 
				p(x\, |\, \theta, g)=g(x-\theta),\: \theta\in\R,\:
				\mbox{$g(x)$ is symmetric around $0$}\},\nonumber\\
\end{eqnarray}
and estimate the parameter $\theta\in\R$ from
the data $x_1,x_2,...,x_N\in\R$.

If $g(x)$ is known, the variance of  the best consistent estimator
is given by
\begin{eqnarray}
\frac{1}{N J}+o\left(\frac{1}{N}\right)
\label{eqn:semicr}
\end{eqnarray}
where $J$ is the Fisher information,
\begin{eqnarray}
J=\int \left(\frac{d}{dx}\ln g(x)\right)^2 g(x)dx.
\end{eqnarray}
In the case where  $g(x)$ is not known,
the {\it theorem} 2.2 in the  Ref. $\cite{Bickel}$ insists
that the bound $(\ref{eqn:semicr})$ is attainable:

\begin{theorem}
If $g(x)$ is absolutely continuous, 
the bound $(\ref{eqn:semicr})$ is attainable
by some consistent estimate 
(see pp. 649-650 in th Ref $\cite{Bickel}$).
\end{theorem}

By the use of this theorem, 
if $p(x)$ and $\tilde{p}(p)$ defined 
in the end of section $\ref{sec:estshift}$ are 
symmetric about $x_0$ and $p_0$ respectively,
we can use the semiparametric estimates, instead 
of the maximum likelihood estimates $(\ref{eqn:mlexp})$,
and can achieve the same efficiency as  $(\ref{eqn:mlexp})$.
Then, if we are so fortunate that the $|\phi_0\rgl$
is happen to be $|n\rgl$ with very large $n$,
our estimate is quite accurate.

\chapter{Time-energy uncertainty in view of hypothesis test}
\section{Conventional discussion about time-energy uncertainty}
This chapter is the result of the joint research with Mr. S. Osawa of
Tokyo Institute of Technology,
which aimed at a consistent and physically meaningful formulation of 
the time-energy uncertainty.
  
  There are various formulation of the time-energy uncertainty relation,
corresponding to the variety of the interpretation of 
uncertainty $\Delta t$ of time $t$.
Some authors introduce `time operator' which forms canonical pair with
the Hamiltonian $H$ of the system.
In this chapter, it is shown that the formulation of 
the time-energy uncertainty is quite reasonable, consistent, and physically
meaningful.

An acceptable interpretation of  $\Delta t$  is the
time interval during which 
the state of a system can hardly 
be distinguished from the initial state. 
For example, it is derived in the explanation  
of the sudden approximation in Messiah $\cite{Messiah}$.   
The outline is as follows.

\vspace{1cm} 
We suppose the Hamiltonian to change-over in a continuous way 
from a certain initial time $t_0$ to a certain final time $t_1$. We put
\begin{eqnarray}
    \Delta t \equiv t_1-t_0
\nonumber
\end{eqnarray}
  and denote by $H(t)$ the Hamiltonian at time $t$. 

 Let $|0 \rangle $ denote the state vector of the system at time $t_0$ , 
and $U(t_1, t_0)$ the time evolution operator 
from $t_0$ to $t_1$. 
 
  The sudden approximation consists in writing 
\begin{eqnarray}
U(t_1,t_0)|0\rangle \approx |0\rangle .
\nonumber
\end{eqnarray}
Messiah regarded
\begin{eqnarray}
   w\equiv  
\lgl 0 | U^ {\dag}(t_1, t_0)\,(I-|0\rgl\lgl 0|)\,U(t_1, t_0)|0\rgl 
\label{eq:w},
\end{eqnarray}
as  a `probability of finding 
the system in a state other than the initial state'.
One obtains the expansion of $w$ in powers of $\Delta t$ 
by the perturbation method. Put 
\begin{eqnarray}
   \overline {H} = \frac{1}{ \Delta t} \int_{t_0}^{t_1} H(t) dt.
\nonumber
\end{eqnarray}
We then have 
\begin{eqnarray}
 w= 
\frac {\Delta t^2}{\hbar^2}\lgl 0 |\overline{H}
(I-|0\rgl\lgl 0|) \overline{H} |0 \rgl 
+ O(\Delta t)^3. 
\nonumber
\end{eqnarray}
Since 
\begin{eqnarray}
\langle 0 |\overline{H}Q_0 \overline{H} |0 \rangle 
=\langle 0|\overline{H}^2|0 \rangle - \langle 0|\overline{H}|0 \rangle ^2
=\lgl(\Delta \overline{H})^2\rgl,
\nonumber
\end{eqnarray}
we have
\begin{eqnarray}
 w= \frac{\Delta t^2 \lgl(\Delta \overline{H})^2\rgl }{\hbar^2}+ O(\Delta t)^3.
\nonumber
\end{eqnarray}
 Thus the condition for the validity of the sudden approximation, 
$w\ll1$, requires that 
\begin{equation}
 \Delta t \ll \frac{\hbar}{\sqrt{\lgl(\Delta \overline{H})^2\rgl}}
\label{eqn:messTE}
\end{equation}
 
\vspace{1cm}
We can point out a  defect in  this discussion. 
In Messiah's  discussion,
the following testing scheme is 
 implicitly assumed;
Let $M_{ms}$ the measurement which takes value 
on the set $\Omega=\{0,\, 1\}$ such that,
\begin{eqnarray}
M_{ms}(\{ 0 \})=|0\rgl\lgl 0|,\;\;
M_{ms}(\{ 1 \})=I-|0\rgl\lgl 0|\;\; ;
\nonumber
\end{eqnarray}
If  the outcome of the measurement 
is  $0$, we accept that the system is in the initial state.
However, if there is a testing scheme which works better than this,
$\Delta t$ should be smaller than  $(\ref{eqn:messTE})$ implies.
Hence, for the Messiah's discussion to be valid,
the optimality of the testing scheme must be shown.
In this study, we investigate this point  
from the viewpoint of hypothesis testing.

\section{Time-energy uncertainty as a hypothesis test}
Here, we formulate the problem as a hypothesis testing 
(see Ref. \cite{Lehmann:1986}).
Consider the model
\begin{eqnarray}
{\cal M}=\{\rho(t)\,|\, 
		\rho(t)=U(t_1,t_0)\,\rho(t_0)\,U^{\dagger}(t_1,t_0),\, t\in\R\}
\nonumber
\end{eqnarray}
and the hypothesises
\begin{eqnarray} 
\begin{array}{ccc}
 H_0: & \rho(t)=\rho(t_0) & \mbox{( null hypothesis)}, \\
 H_1: & \rho(t)=\rho(t_1) & \mbox{ (alternative hypothesis)}.
\end{array}
\label{eqn:TEqutest}
\end{eqnarray}

Then, we choose a test
which maximize the probability $\gamma$ 
of 
when $H_1$ is really true.
The maximization is to be done 
under the restriction
that the probability of rejecting   $H_0$  
when $H_0$ holds true is smaller than the {\it significance level} $\alpha$,
for, otherwise, the test which always reject $H_0$ is chosen. 
We call $\gamma$ the {\it power of the test}.

Two steps are needed to maximize 
the power $\gamma$ of  the test. 
The first step is to find the most powerful tests 
of the following classical hypothesis testing of the parameter family
\begin{eqnarray}
{\cal M}(M)
=\{ p_M(x\,|\, t)\, |\,  
	p_M(x\,|\, t)dx=\tr \rho(t) M(dx),\;\; t\in\R\}
\nonumber
\end{eqnarray}
of the probability distributions which is deduced from ${\cal M}$
by the measurement $M$; The hypothesis are set to be
\begin{eqnarray} 
\begin{array}{ccc}
H_0: &  p_M (x | t)= p_M (x | t_0) & \mbox{( null hypothesis)}, \\
H_1: &  p_M (x | t)=p_M (x | t_1) &\mbox{ (alternative hypothesis)},
\end{array}
\label{eq:TEctest}
\end{eqnarray}
corresponding to $(\ref{eqn:TEqutest})$.
We denote by $\gamma_M$ the power of the most powerful tests
in this hypothesis testing.
In the second step, we adjust the measurement $M$ to maximize $\gamma_M$,
and obtain the optimal tests and its power.

\section{Power of the test}
 Let us consider the power of test and the optimum measurement
when we are given $N$ copies of the state.
 To begin with, consider the first step. 

Whatever the significance levle of the test is,
Stein's lemma in the  classical statistics
gives the maximum power $\gamma_M$ in the classical
hypothesis testing $(\ref{eq:TEctest})$ as
\begin{eqnarray}
\lim_{N\rightarrow\infty}\frac{1}{N}\ln(1-\gamma_M(N)\,)
=-D(\,p_M(x | t_0\,)\,||\,p_M(x | t_1)\,),
\label{eqn:gM}
\end{eqnarray}
where $D(\,p(x)\,||\,q(x)\,)$
is {\it Kullback Divergence} defined by  
\begin{eqnarray}
D(\,p(x)\,||\,q(x)\,)\equiv
\int p(x)\ln\frac{p(x)}{q(x)}dx.
\nonumber
\end{eqnarray}
When $N$ is very large, $(\ref{eqn:gM})$ roughly writes
\begin{equation}
  \gamma_M  \approx 1- \exp[-N D(\,p_M(x | t_0\,)\,||\,p_M(x | t_1)\,)\,],
\label{eqn:gM2}
\end{equation}
where the argument $N$ in $\gamma_M(N)$ is dropped for notational simplicity.

Expansion of $D(\,p_M(x | t_0\,)\,||\,p_M(x | t_1)\,)$ in powers of
$\Delta t=t_1 -t_0$ gives
\begin{eqnarray}
D(\,p_M(x | t_0\,)\,||\,p_M(x | t_1)\,)
=\frac{1}{2}J_M(\theta_0)\,(\Delta t)^2 + o(\Delta t)^2,
\nonumber
\end{eqnarray}
where $J_M(\theta_0)$ is the classical Fisher information
of the classical model ${\cal M}(M)$.
Hence, when $\delta t$ is very small and $N$ is very large,
$\gamma_M$ can be written roughly as
\begin{equation}
\gamma_M 
\approx 1-\exp\left(-\frac{N}{2}J_M(t_0)(\Delta t)^2\,\right)
+o(\Delta t)^2
\qquad(\Delta t \ll 1),
\label{eqn:gM3}
\end{equation}

Let us move to the second step.
By virtue of $(\ref{eqn:gM2})$, 
when $\Delta t$ is very small, the maximization of
$\gamma_M$ is equivalent to that of $J^M(\theta_0)$,
the answer of which is given as
\begin{eqnarray}
\max_M J_M(\theta_0)=J^S(\theta_0),
\nonumber
\end{eqnarray}
where $J^S(\theta_0)$ is the SLD Fisher information of 
the model ${\cal M}$ (see Theorem 1 in Ref. \cite{Nagaoka:1989}).
$J^S(\theta_0)$ is easily calculated as
\begin{eqnarray}
J^S(\theta_0)=\frac{4}{\hbar^2}\lgl\Delta H^2\rgl,
\label{eqn:TEjs}
\end{eqnarray}
and finally we obtain the power $\gamma_max$ of the optimum test
when $N$ is very large and $\Delta t$ is very small;
\begin{eqnarray}
\gamma_{max}
\approx 
1-\exp\left(-\frac{2N}{\hbar^2}(\Delta t)^2\lgl(\Delta H)^2\rgl\right)
+o(\Delta t)^2
\qquad(\Delta t \ll 1,\,N\gg 1).
\nonumber
\end{eqnarray}

Now we can show the condition that $\rho(t)\,(t\ge t_1)$ 
can hardly be distinguished from $\rho(t_0)$  using $N$ data 
when $\Delta t \ll 1$ and $N \gg 1$ are satisfied;
the condition writes
\begin{eqnarray}
\gamma_{max}\approx 
1-
\exp\left(-\frac{2N}{\hbar^2}(\Delta t)^2\lgl(\Delta H)^2\rgl\right)
+o(\Delta t)^2 
\ll 1,
\nonumber
\end{eqnarray}
or equivalently,
\begin{eqnarray}
\frac{2N \Delta t^2\lgl(\Delta H)^2\rgl}{\hbar^2}\ll 1
\quad(N\gg 1,\,\Delta t\ll 1). 
\label{eqn:htestTE} 
\end{eqnarray}
Notice that Messiah's condition $(\ref{eqn:messTE})$
is identical to $(\ref{eqn:htestTE})$ when $\Delta t$ is very small.

\section{The optimal measurement}
In this section, it is shown that the test 
based on the measurement $M_{ms}$ is one of the optimal tests.

The classical Fisher information $J_{M_{ms}}(t_0)$ 
of the model ${\cal M}(M_{ms})$ 
is given by
\begin{eqnarray}
J_{M_{ms}}(t_0)=\lim_{t \to t_0} 
\left(\frac{\dot{p}(0|t)^2}{p(0|t)}
+\frac{\dot{p}(1|t)^2}{p(1|t)}\right),
\nonumber
\end{eqnarray}
where
\begin{eqnarray}
p(0|t)&=&\tr \rho(t) M_{ms}({0})\nonumber\\
p(1|t)&=&1-p(0|t).
\nonumber
\end{eqnarray} 
Expansion of $p(0|t)$ in powers of $t-t_0$ gives
\begin{eqnarray}
p(0|t)&=&\dot{p}(0|t_0)(t-t_0)+\frac{1}{2}\ddot{p}(0|t_0)(t-t_0)^2
			+ o(t-t_0)^2\nonumber\\
      &=&-\frac{1}{\hbar^2}(\lgl(\Delta H)^2\rgl)\,(t-t_0)^2 +o(t-t_0)^2,
\nonumber
\end{eqnarray}
which leads to
\begin{eqnarray}
J_{M_{ms}}(t_0)=\frac{4}{\hbar^2}\Delta H^2.
\label{eqn:TEjm}
\end{eqnarray}
From $(\ref{eqn:TEjs})$ and $(\ref{eqn:TEjm})$, 
$M_{ms}$ is one of the optimum measurements.

\part*{Part III\\ 
The general  model theory}
\addcontentsline{toc}{chapter}%
{Part III :
The general  model theory}
\chapter{Geometrical structure}
\section{$w$-connection and $e$-connection in the wider sense}
\label{sec:gentwt}
This chapter presents the results obtained in the joint research with
Dr. A. Fujiwara of Osaka University.
In this section, we treat with the generalization of 
 the $w$-connection in the faithful model theory
to the general case.

To define the generalized $w$-connection, we need to specify 
the logarithmic derivative ${\bf L}(X)$ among 
the various roots of the matrix equation $(\ref{eqn:deflogd})$. 
Moreover, we need to check whether 
a  connection can be defined by $\ref{eqn:defpttw}$ or not.
For the system of vectors
\begin{eqnarray}
\{
{\bf L}(\partial_{\zeta^i}|_{W_1})W_2
-(\tr\pi(W_2) 
{\bf L}(\partial_{\zeta^i}|_{W_1})\,) I\: 
|\: i=1,2,..., 2dr-1
\}
\label{eqn:spantrans}
\end{eqnarray}
might not be a linearly independent.
However, for any choice
of the logarithmic derivative, 
if $W_2$ is  near enough to $W_1$,
 $(\ref{eqn:spantrans})$ is linearly independent,
because the mapping 
\begin{eqnarray}
L\longrightarrow LW-\tr(\pi(W)L)I
\nonumber
\end{eqnarray}
is continuous. 

Therefore, we define
\begin{eqnarray}
{\cal N}(W_1, \varepsilon)
=\{ W\: |\:  W\in{\cal N}, W\in B(W_1, \varepsilon)\},
\nonumber
\end{eqnarray}
where  
\begin{eqnarray}
B(W_1, \varepsilon)
\equiv
\left\{ W\: 
\left|\: \sqrt{\lgll W-W_1, W-W_1\rgll} \leq\varepsilon
\right.
\right\},
\nonumber
\end{eqnarray}
and we restrict ourselves to the consideration of
 ${\cal N}(W_1, \varepsilon)$ with enough small $\varepsilon$.
Then, a choice of the logarithmic derivative
defines a generalization of the $w$-connection.

In the almost same way, we can define generalized $e$-connection
in the manifold 
\begin{eqnarray}
{\cal M}(\theta, \varepsilon)
=\{ \rho\: |\:  \rho\in{\cal M}, W\in B(\rho_1, \varepsilon)\},
\nonumber
\end{eqnarray}
with small $\varepsilon$.
Or, we can deduce a generalized $e$-connection
from a $w$-connection as in $(\ref{eqn:ptwte})$.
Obviously, the equation$(\ref{eqn:etortion2})$ 
and 
$(\ref{eqn:mttw})$ 
hold true in any generalization 
of the $w$-connection and $e$-connection.
The generalized $w$-connection and the generalized $e$-connection
deduced from that $w$- connection satisfies theorem $\ref{theorem:lslsdecom}$.

\section{Vanishing conditions}
\label{sec:genvanish}
In this section, we examine
the conditions  that $e$-tortion,
$w$-tortion, and/or  Uhlmann curvature would vanish. 

We consider following conditions
for the model ${\cal M}={\cal P}_r$ and 
for the manifold ${\cal N}={\cal W}_r$.
Here, $W$ denotes a member of ${\cal W}_r$.

\begin{itemize}
\item[(A)] Algebraic conditions
	\begin{enumerate}
		\item The SLD can be chosen to satisfy
				$W^*[L^S_{X}, L^S_{Y}]W=0$.
		\item The SLD can be chosen to satisfy
				$[L^S_{X}, L^S_{Y}]W=0$.
		\item The SLD can be chosen to satisfy
				$[L^S_{X}, L^S_{Y}]=0$.
	\end{enumerate}
\item[(G)] Geometric conditions
	\begin{enumerate}
		\item Uhlmann curvature vanishes at $\pi(W)$.
		\item There is a generalized $w$-tortion which vanishes at $W$
				for horizontal vectors.
	\end{enumerate}
\end{itemize}

\begin{lemma}
\begin{itemize}
\item[(1)]  G1 is equivalent to  A1
\item[(2)]  G2 is equivalent to  A2
\end{itemize}
\end{lemma}
\begin{proof}
(1) : Since Uhlmann curvature is the vertical component of
$w$-tortion 
${\bf M}\Ttw(\hlift(X),\hlift(Y))=(1/4)[L_X^S,L_Y^S]W$,
Uhlmann curvature vanishes iff
$(1/4)[L_X^S,L_Y^S]W$ is orthogonal to
the vertical subspace,
or equivalently, iff
for any skew Hermitian matrix $A$,
\begin{eqnarray}
{\rm Re}\,\tr \left\{(WA)^*(1/4)[L_X^S,L_Y^S]W\right\}
=(-1/4){\rm Re}\,\tr \left\{AW^*[L_X^S,L_Y^S]W\right\}
=0
\nonumber
\end{eqnarray}
holds true,
which is equivalent to $W^*[L_X^S,L_Y^S]W=0$.

(2) : The statement is derived directly from $(\ref{eqn:mttw})$.
\end{proof}

\begin{lemma}
As for the Algebraic conditions,
\begin{eqnarray}
A1\,\Leftarrow\, A2\Leftarrow\, A3,\nonumber
\nonumber
\end{eqnarray}
but non of the converses do not hold true.
\end{lemma}
\begin{proof}
The former half of the assertion is trivial.
As for the $A1\Rightarrow A2$,
 we have the following counter-example:
\begin{eqnarray}
W=\left(
\begin{array}{cc}
\sigma\\
0
\end{array}
\right),\:\:\:
L_{1}^SW=\left(
\begin{array}{cc}
0\\
{\bf a}^*\sigma
\end{array}
\right),\:\:\:
L_{2}^SW=\left(
\begin{array}{cc}
B\sigma\\
{\bf 0}^*
\end{array}
\right),
\nonumber
\end{eqnarray}
where $\sigma$ is a $r\times r$ reversible matrix, 
${\bf a}$ a member of $R^m$ which is not eigenvector of 
$r\times r$ matrix $B$, and $0$ the vector all of whose components are zeros.
The following is the counter-example of $A2\Rightarrow A3$:
\begin{eqnarray}
W=\left(
\begin{array}{cc}
\sigma\\
0
\end{array}
\right),\:\:\:
L_1^SW=\left(
\begin{array}{cc}
0\\
B\sigma
\end{array}
\right),\:\:\:
L_2^SW=\left(
\begin{array}{cc}
0\\
CB\sigma
\end{array}
\right),
\label{eqn:gence23}
\end{eqnarray}
where $d$ is larger than or equal to $r+2$, 
$B$ is a $(d-r)\times r$ full rank matrix, 
$C$ is a $(d-r)\times (d-r)$ matrix,
 and $BB^*$ does not commute with $C$.
\end{proof}

If the model is faithful or pure, the conditions A1-A3 are equivalent.
Therefore, not only G1$\Leftarrow$G2,
but also G1$\Rightarrow$G2 holds true.
However, as is understood by the above two lemmas,
this is not the case generally:

\begin{theorem}
 G2 implies G1, but not vice versa. 
\end{theorem}

\section{Pure state model revisited}
\label{sec:puret0}
In the pure state model, we can also introduce
the generalized $e$- and $w$- connections.
One remarkable fact about the pure state model is 
that the generalized $e$-connection in  the total space model ${\cal P}_1$ 
can be chosen so that the $e$-tortion vanishes. 

Actually, by taking $L^S_X=2X\rho$, we have
\begin{eqnarray}
\frac{1}{4}[[L^S_X, L^S_Y],\rho]
&=&[[X\rho, Y\rho],\rho]\nonumber\\
&=& \frac{1}{4}
[[(|l_X\rgl\lgl\phi|+|\phi\rgl\lgl l_X|),
(|l_Y\rgl\lgl\phi|+|\phi\rgl\lgl l_Y|) ], |\phi\rgl\lgl\phi|]\nonumber\\
&=&0.
\nonumber
\end{eqnarray}
Therefore,
if we  choose  SLD and the logarithmic derivative 
for horizontal vectors this way, 
Uhlamnn curvature is equal to the $w$-tortion for the horizontal vectors,
which means that  
the pure state model is locally quasi-classical 
iff there is a generalized $w$-connection whose tortion vanishes.

\chapter{Attainability of SLD CR bound}
\section{Commutative SLD and attainability of the bound}
In the faithful model,
a necessary and sufficient condition for SLD CR bound  
to be attained
is that the SLD's are commutative,
and in the pure state model, that condition 
is the existence of the commutative SLD's.
Hence, 
one might come up with the conjecture that
in general, attainability of SLD CR bound is
equivalent to the existence of commutative SLD's,
which is the algebraic condition A1 in section $\ref{sec:genvanish}$.

Actually, it is easily shown that SLD CR bound can be achieved
if SLD's commute. The optimal measurement
is the one in the theorem $\ref{theorem:mpCR}$,
that is, the simultaneous spectral decomposition of 
commutative SLD's.
However, in the followings, the case of $(\ref{eqn:gence23})$ 
is shown to be a counter-example of 
the converse of the statement. Therefore, 
the algebraic condition A1 is only a sufficient condition
for the attainable SLD CR bound.

Let 
${\cal K}$ be a $d+\max\{r,d-r\}$-dimensional complex vector space
such that ${\cal H}$ is a its subspace,
$P$ the  projection from 
${\cal K}$ onto ${\cal H}$,
$L^j\:(j=1,2)$ the matrices which satisfy 
\begin{eqnarray}
L^j W=
\sum_{k=1}^2[J^{S-1}]^{j,k}
{\bf M}\left(\hlift_W \left(\frac{\partial}{\partial\theta^k}\right)\right),
\label{eqn:defL^j}
\end{eqnarray}
$\tilde{L}^j\:(j=1,2)$ the matrices in ${\cal K}$
such that
\begin{eqnarray}
\tilde{L}^1\equiv\left(
\begin{array}{ccc}
0 & B^*C_1 & 0\\
C_1B &0 & C_1D \\
0 & D^*C_1 &0
\end{array}
\right),\:
\tilde{L}^2 \equiv \left(
\begin{array}{ccc}
0 & B^*C_2 & 0\\
C_2B &0 & C_2D \\
0 & D^*C_2 &0
\end{array}
\right),
\nonumber
\end{eqnarray}
where 
\begin{eqnarray}
C_1&\equiv&[J^{S-1}]^{1,1}I_{d-r}+[J^{S-1}]^{1,2}C,\nonumber\\
C_2&\equiv&[J^{S-1}]^{2,1}C+[J^{S-1}]^{1,1}I_{d-r},
\label{eqn:c1c2}
\end{eqnarray}
and $D$ will be defined soon.
Notice 
$\tilde{L}^j\:(j=1,2)$ are defined so that
they satisfy
\begin{eqnarray}
P\tilde{L}^j W=L^j W
\: (j=1,2),
\nonumber
\end{eqnarray}
where $W$ in the left hand side of the equation means
\begin{eqnarray}
\left(
\begin{array}{c}
W\\
0
\end{array}
\right).
\label{eqn:w=wo}
\end{eqnarray}

By virtue of
\begin{eqnarray}
\tilde{L}^1\tilde{L}^2=
\left(
\begin{array}{ccc}
B^*C_1C_2B & 0 & B^*C_1C_2D\\
0 & C_1(BB^*+DD^*)C_2  & 0\\
D^*C_1C_2B & 0 & D^*C_1C_2D
\end{array}
\right)
\nonumber
\end{eqnarray}
and  $(\ref{eqn:c1c2})$,
if we choose $D$ such that
\begin{eqnarray}
D=(\alpha I - BB^*)^{1/2},
\nonumber
\end{eqnarray}
where $\alpha$ is the maximum eigenvalue of the matrix $BB^*$,
$\tilde{L}^1$ and $\tilde{L}^2$ commute.

Let $E$ be a projection valued measurement such that
\begin{eqnarray}
\tilde{L}^i=\int \hat\theta^i E(d\hat\theta)\: (i=1,2),
\nonumber
\end{eqnarray}
and $M$ a measurement deduced from $E$ as
\begin{eqnarray}
M(B)=PE(B)P.
\nonumber
\end{eqnarray}
Then, we have
\begin{eqnarray}
V_{\pi(W)}[M]&=&V_{\pi(\tilde{W})}[E]\nonumber\\
&=&[{\rm Re}\,\tr \pi(\tilde{W}) \tilde{L}^j \tilde{L}^k]\nonumber\\
&=&[{\rm Re}\,\tr \pi(W) L^j L^k]\nonumber\\
&=&J^{S-1}.
\nonumber
\end{eqnarray}
Because, as is shown in the previous section,
any  matrices $A_1, A_2$ which satisfy
\begin{eqnarray}
A_j W=
{\bf M}\left(\hlift_W \left(\frac{\partial}{\partial\theta^j}\right)\right)\: 
(j=1,2)\nonumber
\end{eqnarray}
do not commute with each other,
$(\ref{eqn:gence23})$ is a counter-example of the conjecture.

\section{A necessary condition and the main conjectures}
\begin{theorem}
If SLD CR bound is achieved,
G2($\Leftrightarrow$A1) in section $\ref{sec:genvanish}$
holds true.
\end{theorem}
\begin{proof}
Let $M$ be a measurement which satisfies $(\ref{eqn:genxx=J})$,
$E$ the Naimark dilation of $M$,
${\cal K}\supset{\cal H}$ the Hilbert space which $E$ lives in,
$P$ the projection from ${\cal H}$ onto ${\cal K}$,
$W$ a member of ${\cal W}_d$ such that $\rho(\theta)=\pi(W)$,
$L^j\,(j=1,...,m)$ the operators in ${\cal H}$
which satisfy $(\ref{eqn:defL^j})$,
$\tilde{L}^j\,(j=1,...,m)$ the operators ${\cal K}$ which satisfy
\begin{eqnarray}
{\bf M}^j(\hat\theta, E, W)=\tilde{L}^j W,
\nonumber
\end{eqnarray}
where $W$ in the equation means $(\ref{eqn:w=wo})$.
Then, for SLD CR bound to be attained,
\begin{eqnarray}
[V[E]]_{jj}(=[V[M]]_{jj})=tr \rho(\theta)(L^j)^2(=[J^{S-1}]^{jj}),
\nonumber
\end{eqnarray}
or, its equivalence,
\begin{eqnarray}
\tr \rho(\theta)(\tilde{L}^j)^2=\tr \rho(\theta)(L^j)^2
\label{eqn:extL=L}
\end{eqnarray}
must hold true.
Since 
\begin{eqnarray}
P\tilde{L}^j W= L^j W
\nonumber
\end{eqnarray}
follows directly from their definitions,
we have
\begin{eqnarray}
\tr W^*(\tilde{L}^j)^2 W 
&=&\tr W^*(\tilde{L}^j)P(\tilde{L}^j)+\tr W^*(\tilde{L}^j)(1-P)(\tilde{L}^j)
\nonumber\\
&=&\tr W^*(L^j)^2W=\tr \rho(\theta)(L^j)^2,
\nonumber
\end{eqnarray}
which, combined with $(\ref{eqn:extL=L})$, leads to
\begin{eqnarray}
P\tilde{L}^j W=L^j W.
\nonumber
\end{eqnarray}
Hence, we have
\begin{eqnarray}
&&P(\tilde{L}^j \tilde{L}^k-\tilde{L}^k \tilde{L}^j)W
\nonumber\\
&=&P(\tilde{L}^j P\tilde{L}^k-\tilde{L}^k P\tilde{L}^j)W
\nonumber\\
&=&P(\tilde{L}^j P\tilde{L}^k P-\tilde{L}^k P\tilde{L}^j P)W.
\nonumber
\end{eqnarray}
Since 
$[\tilde{L}^j, \tilde{L}^k]=0$ follows from the definition,
this means
\begin{eqnarray}
[P\tilde{L}^j P, P\tilde{L}^k P]W=0.
\nonumber
\end{eqnarray}
Because we can take the SLD's such that
\begin{eqnarray}
P\tilde{L}^j P =\sum_k [J^{S-1}]^{j,k}L^S_k,
\nonumber
\end{eqnarray}
we have the theorem.
\end{proof}

\begin{conjecture}
If G2($\Leftrightarrow$A1) in section $\ref{sec:genvanish}$
holds true, the equality in CR inequality can be achieved.
\end{conjecture}

In the pure state model, 
this conjecture is true as is mentioned 
in the end of section $\ref{sec:gentwt}$.
Because of $G1\Leftrightarrow G2$ and theorem $\ref{th:estuhl1}$,
the conjecture is valid also in the faithful model.

\begin{conjecture}
$w$-tortion is a good index of noncommutative nature of the model.
\end{conjecture}

This statement is proved to be true in the 2-parameter pure state model.



\chapter{Conclusions}

As for a geometrical side of the thesis, $w$-connection is proposed 
as a medium to unify Nagaoka's information geometry 
and Uhlmann's parallelism.

Our conjecture is that the tortion of 
$w$-connection is a good measure of noncommutative nature of the model.
This conjecture is proved 
for the 2-parameter pure state model.
The following seem to support the conjecture:
\begin{itemize}
\item[(1)] The attainable CR type bound of 
			the general pure state model 
			with the weight matrix $J^S$.
\item[(2)] The condition for the pure state model and the faithful model 
			to be locally quasi-classical.
\end{itemize}
The proof (or disproof) of the conjecture is an open problem.

As for the global property of the model,
the faithful model is quasi-classical iff the model is parallel.
However, in the pure state model,
being parallel is sufficient condition but not necessary condition.
Therefore, Uhlmann's RPF might not characterize the global
property of the model in general.
However, the condition for being parallel seems to have
intrinsic relation with some kind of symmetry.

As for the determination of the attainable CR type bound,
we succeeded in the case of the 2-parameter pure state model 
and the coherent model. 
The CR type bound with weight matrix $J^S$ is also 
calculated for arbitrary pure state model.

We successfully applied the result in the estimation theory
to the analysis of the position-momentum uncertainty. 
The main points are that the mean value of the position and the momentum
can be estimated up to arbitrary efficiency and 
that Planck's constant has nothing 
to do with noncommutative nature of the position-momentum shifted model.
As for the analysis of the time-energy uncertainty,
we succeeded in the formulation of the problem
in a good shape without the help of `time operator'.

\end{document}